\documentclass[acmsmall,]{acmart}

\acmJournal{PACMPL}
\acmVolume{1}
\acmNumber{CONF} 
\acmArticle{1}
\acmYear{2018}
\acmMonth{1}
\acmDOI{} 
\startPage{1}

\setcopyright{none}

\usepackage{graphicx}
\usepackage{subcaption}
\usepackage{etoolbox}

\usepackage{url}

\makeatletter
\AtBeginDocument{%
  \def\doi#1{\url{https://doi.org/#1}}}
\makeatother
\usepackage{wrapfig}

\usepackage[inline]{enumitem}
\newtoggle{full}
\settoggle{full}{true}

\usepackage{mathpartir}

\usepackage{tikz}
\usetikzlibrary{cd, calc, positioning, decorations.markings, decorations.pathmorphing, decorations.pathreplacing, shapes }

\usepackage{cancel}

\usepackage{xspace}
\usepackage{fancyvrb}

\definecolor{dkblue}{rgb}{0,0.1,0.5}
\definecolor{lightblue}{rgb}{0,0.5,0.5}
\definecolor{dkgreen}{rgb}{0,0.4,0}
\definecolor{dk2green}{rgb}{0.4,0,0}
\definecolor{dkviolet}{rgb}{0.6,0,0.8}
\definecolor{mantra}{rgb}{0.2,0.6,0.2}
\definecolor{gotcha}{rgb}{0.8,0.2,0}
\definecolor{ocre}{RGB}{243,102,25} 
\definecolor{dkolive}{RGB}{85, 107, 47}
\definecolor{pine}{RGB}{1, 121, 111}
\definecolor{DarkSlateBlue}{RGB}{72,61,139}
\definecolor{dkred}{RGB}{139, 0, 0}
\definecolor{lsgray}{RGB}{199,216,233}
\definecolor{gbgray}{RGB}{220,220,220}
\definecolor{lightyellow}{RGB}{255,255,224}
\definecolor{papayawhip}{RGB}{255,239,213}
\definecolor{lightyg}{RGB}{229, 255, 153}
\definecolor{cbblue}{RGB}{0,114,178}
\definecolor{cbgreen}{RGB}{0,158,115}
\definecolor{cbred}{RGB}{213,94,0}

\usepackage[nomargin,inline,draft]{fixme}
\fxsetup{theme=color,mode=multiuser}
\FXRegisterAuthor{FF}{aFF}{FF} 
\FXRegisterAuthor{NY}{aNY}{NY} 
\FXRegisterAuthor{DC}{aDC}{DC} 
\FXRegisterAuthor{LG}{aLG}{LG} 

\newcommand{\code}[1]{\texttt{\small #1}} 
\newcommand{\rulename}[1]{\withcolor\colorrules{\scriptsize\textsc{[#1]}\xspace}}

\newcommand{\SEP}{\ensuremath{~~\mathbf{|\!\!|}~~ }}

\ifdefined\NOCOLOR
  \newcommand{\withcolor}[2]{#2} 
\else
  \newcommand{\withcolor}[2]{\colorlet{currbkp}{.}\color{#1}{#2}\color{currbkp}}
\fi

\ifdefined\NO\COLOR
\else
  
\fi

\newcommand{\colorlp}{blue} 
\newcommand{\colorgt}{violet} 
\newcommand{\colorht}{stColor} 
\definecolor{stColor}{rgb}{0.1, 0.4, 0.1}
\newcommand{\colorrules}{DarkSlateBlue} 

\let\DefTirNameOld\DefTirName
\renewcommand{\DefTirName}[1]{\withcolor\colorrules{\DefTirNameOld{[{#1}]}}}

\newcommand{\hla}[1]{\colorbox{lightyellow}{#1}}
\newcommand{\hlb}[1]{\colorbox{papayawhip}{#1}}
\newcommand{\hlc}[1]{\colorbox{lightyg}{#1}}

\newcommand{\dte}[1]{\SortCol{#1}} 

\newcommand{\tunit}{\dte{\code{unit}}}
\newcommand{\tbool}{\dte{\code{bool}}}

\newcommand{\tnat}{\dte{\code{nat}}}
\newcommand{\tint}{\dte{\code{int}}}

\newcommand{\tS}{\dte{\code{S}}}
\newcommand{\tplus}[2]{\dte{#1\code{+}#2}}
\newcommand{\tpair}[2]{\dte{#1\code{*}#2}}

\newcommand{\roleFmt}[1]{\roleCol{\textbf{\texttt{\upshape #1}}}}
\newcommand{\roleCol}[1]{{\color{pine}#1}}

\definecolor{maroon}{RGB}{128,0,0}
\definecolor{labelColor}{RGB}{72,61,140}
\newcommand{\lblFmt}[1]{{{\color{labelColor}{#1}}}}
\newcommand{\elllbl}{\lblFmt \ell}

\definecolor{SortColor}{RGB}{72,61,139}
\newcommand{\SortCol}[1]{{\color{SortColor}{#1}}}

\newcommand{\p}{{\roleFmt{p}}}
\newcommand{\q}{{\roleFmt{q}}}
\newcommand{\pr}{{\roleFmt{r}}}
\newcommand{\ps}{{\roleFmt{s}}}
\newcommand{\pt}{{\roleFmt{t}}}
\newcommand{\pu}{{\roleFmt{u}}}
\newcommand{\ph}{{\roleFmt{h}}}
\newcommand{\pk}{{\roleFmt{k}}}
\newcommand{\pl}{{\roleFmt{l}}}
\newcommand{\pext}{{\roleFmt{ext}}}

\newcommand{\phtwo}{{\roleFmt{h$\color{pine}{_2}$}}}
\newcommand{\phthree}{{\roleFmt{h$\color{pine}{_3}$}}}
\newcommand{\pkone}{{\roleFmt{k$\color{pine}{_1}$}}}

\newcommand{\pkthree}{{\roleFmt{k$\color{pine}{_3}$}}}
\newcommand{\plone}{{\roleFmt{l$\color{pine}{_1}$}}}
\newcommand{\pltwo}{{\roleFmt{l$\color{pine}{_2}$}}}

\newcommand{\dlt}[1]{\withcolor{\colorlp}{#1}} 

\newcommand{\lT}{\dlt{ \ensuremath{L}}}
\newcommand{\lsend}[4]{\dlt{!{#1};\{\lblFmt{#2_i}(\SortCol{#3_i}). #4_i \}_{i \in I}}}

\newcommand{\lrecv}[4]{\dlt{?{#1};\{\lblFmt{#2_i}(\SortCol{#3_i}). #4_i \}_{i \in I}}}
  \newcommand{\lrecvj}[4]{\dlt{?{#1};\{\lblFmt{#2_j}(\SortCol{#3_j}). #4_j \}_{j \in J}}}
\newcommand{\lrcv}[1]{\dlt{?{#1};}}
\newcommand{\lsnd}[1]{\dlt{!{#1};}}
\newcommand{\lrec}[2]{\dlt{\mu #1 . #2}}
\newcommand{\lX}{\dlt{X}}
\newcommand{\lend}{\dlt{\mathtt{end}}}

\newcommand{\dgt}[1]{\withcolor{\colorgt}{#1}} 

\newcommand{\G}{\dgt{\ensuremath{G}}}
\newcommand{\gmu}{\dgt{\mu}}
\newcommand{\msg}[2]{\dgt{ \ensuremath{#1\to#2:}}}
\newcommand{\msgi}[5]{\dgt{ \ensuremath{#1\to#2:\{\lblFmt{#3_i}(\SortCol{{#4}_i}). #5_i \}_{i \in I}}}}

\newcommand{\gX}{\dgt{X}}
\newcommand{\grec}[2]{\dgt{\mu #1 . #2}}
\newcommand{\gend}{\dgt{\mathtt{end}}}

\newcommand{\guarded}{\code{guarded}}

\newcounter{sarrow}

\newcommand{\dht}[1]{\withcolor{\colorht}{#1}}

\newcommand{\HH}{ \dht{\ensuremath{\mathcal{H}}} }

\newcommand{\hT}{\dht{\ensuremath{H}}}
\newcommand{\hTb}{\dht{\ensuremath{H'}}}
\newcommand{\hTi}{\dht{\ensuremath{H_i}}}

\newcommand{\hTone}{\dht{\ensuremath{H_1}}}

\newcommand{\hTtwo}{\dht{\ensuremath{H_2}}}

\newcommand{\hTthree}{\dht{\ensuremath{H_3}}}
\newcommand{\hTonei}{\dht{\ensuremath{H^1_i}}}
\newcommand{\hTtwoi}{\dht{\ensuremath{H^2_i}}}
\newcommand{\hTonej}{\dht{\ensuremath{H^1_j}}}
\newcommand{\hTtwok}{\dht{\ensuremath{H^2_k}}}
\newcommand{\hTt}{\dht{\ensuremath{H_t}}}
\newcommand{\hTn}{\dht{\ensuremath{H_n}}}
\newcommand{\hTN}{\dht{\ensuremath{H_N}}}
\newcommand{\hTE}{\dht{\ensuremath{H^E}}}
\newcommand{\hTEone}{\dht{\ensuremath{H^{E}_1}}}
\newcommand{\hTEtwo}{\dht{\ensuremath{H^{E}_2}}}
\newcommand{\hTEi}{\dht{\ensuremath{H^E_i}}}
\newcommand{\hTEj}{\dht{\ensuremath{H^E_j}}}

\newcommand{\hTEb}{\dht{\ensuremath{{H^E}'}}}
\newcommand{\hTL}{\dht{\ensuremath{H^L}}}
\newcommand{\hTLi}{\dht{\ensuremath{H^L_i}}}
\newcommand{\hTLone}{\dht{\ensuremath{H^L_1}}}
\newcommand{\hTLoneb}{\dht{\ensuremath{{H^L_1}'}}}

\newcommand{\hTLtwo}{\dht{\ensuremath{H^L_2}}}
\newcommand{\hTLtwob}{\dht{\ensuremath{{H^L_2}'}}}

\newcommand{\ghT}{\dht{\ensuremath{H^\dagger}}}
\newcommand{\ghTb}{\dht{\ensuremath{{H^{\dagger}}'}}}
\newcommand{\ghTi}{\dht{\ensuremath{H^\dagger_i}}}
\newcommand{\ghTj}{\dht{\ensuremath{H^\dagger_j}}}

\newcommand{\ghTk}{\dht{\ensuremath{H^\dagger_k}}}
\newcommand{\ghTone}{\dht{\ensuremath{H^\dagger_1}}}
\newcommand{\ghTtwo}{\dht{\ensuremath{H^\dagger_2}}}

\newcommand{\ghTP}{\dht{\ensuremath{H^{\dagger P}}}}
\newcommand{\ghTPone}{\dht{\ensuremath{H^{\dagger P}_1}}}

\newcommand{\ghTPtwo}{\dht{\ensuremath{H^{\dagger P}_2}}}

\newcommand{\gG}{\dgt{\ensuremath{G^\dagger}}}
\newcommand{\Gpar}{\dgt{\ensuremath{G_{\sf par}}}}
\newcommand{\hTparone}{\dht{\ensuremath{H^1_{\sf par}}}}
\newcommand{\hTpartwo}{\dht{\ensuremath{H^2_{\sf par}}}}
\newcommand{\ghTpar}{\dgt{\ensuremath{G^\dagger_{\sf par}}}}

\newcommand{\hmsg}[2]{\dht{ \ensuremath{#1\to#2:}}}
\newcommand{\nhmsg}[2]{\dht{ \ensuremath{#1\to#2}}}
\newcommand{\hmsgi}[5]{\dht{\ensuremath{#1\to#2:\{\lblFmt{#3_i}(\SortCol{#4_i}).
      #5_i \}_{i \in I}}}}
\newcommand{\hmsgj}[5]{\dht{\ensuremath{#1\to#2:\{\lblFmt{#3_j}(\SortCol{#4_j}).
      #5_j \}_{j \in J}}}}
\newcommand{\hmsgk}[5]{\dht{\ensuremath{#1\to#2:\{\lblFmt{#3_k}(\SortCol{#4_k}).#5_k \}_{k \in K}}}}

\newcommand{\hX}{\dht{X}}
\newcommand{\hY}{\dht{Y}}

\newcommand{\hrec}[2]{\dht{\mu #1 . #2}}
\newcommand{\hsend}[5]{\dht{{#1}!{#2};\{\lblFmt{#3_i}(\SortCol{#4_i}). #5_i \}_{i \in I}}}
\newcommand{\hrecv}[5]{\dht{{#1}?{#2};\{\lblFmt{#3_i}(\SortCol{#4_i}). #5_i \}_{i \in I}}}
\newcommand{\hsendj}[5]{\dht{{#1}!{#2};\{\lblFmt{#3_j}(\SortCol{#4_j}). #5_j \}_{j \in J}}}
\newcommand{\hrecvj}[5]{\dht{{#1}?{#2};\{\lblFmt{#3_j}(\SortCol{#4_j}). #5_j \}_{j \in J}}}
\newcommand{\hsendk}[5]{\dht{{#1}!{#2};\{\lblFmt{#3_k}(\SortCol{#4_k}). #5_k \}_{k \in K}}}
\newcommand{\hrecvk}[5]{\dht{{#1}?{#2};\{\lblFmt{#3_k}(\SortCol{#4_k}). #5_k \}_{k \in K}}}
\newcommand{\hrcv}[2]{\dht{{#1}?{#2};}}
\newcommand{\nhrcv}[2]{\dht{{#1}?{#2}}}
\newcommand{\hsnd}[2]{\dht{{#1}!{#2};}}
\newcommand{\nhsnd}[2]{\dht{{#1}!{#2}}}
\newcommand{\hpar}[2]{\dht{#1 | #2}}

\newcommand{\hend}{\dht{\mathtt{end}}}
\newcommand{\hpart}[1]{{\textsf{ipart}}({#1})}
\newcommand{\hepart}[1]{\textsf{epart}({#1})}

\newcommand{\isglobal}[1]{\dgt{isGlobal(#1)}}
\newcommand{\islocal}[1]{\dlt{isLocal(#1)}}
\newcommand{\hdepth}[1]{\color{black}{\sf{depth}}(#1)}
\newcommand{\loc}[1]{{\color{black}{{\sf{loc}}_{#1}\;}}}
\newcommand{\merge}{\color{black}{\bigsqcup}}
\newcommand{\mrg}{\color{black}{\sqcup}}
\newcommand{\lmerge}{\color{black}{\bigsqcup^L}}
\newcommand{\lmrg}{\color{black}{\sqcup^L}}

\newcommand{\ulb}{{\color{black}{\sf unmL^{bin}}}}
\newcommand{\ul}{{\color{black}{\sf unmL}}}

\newcommand{\ulba}[3]{{\color{black}{{\sf unmL^{bin}}\;#1\;#2\;#3}}}
\newcommand{\ula}[2]{{\color{black}{{\sf unmL}\;#1\;#2}}}
\newcommand{\ulbEa}[3]{{\color{black}{{\sf unmL^{bin}_E}\;#1\;#2\;#3}}}

\newcommand{\ulpb}{{\color{black}{\sf unmLP^{bin}}}}
\newcommand{\ulp}{{\color{black}{\sf unmLP}}}
\newcommand{\ulpba}[4]{{\color{black}{{\sf unmLP^{bin}}\;#1\;#2\;#3\;#4}}}
\newcommand{\ulpa}[2]{{\color{black}{{\sf unmLP}\;#1\;#2}}}
\newcommand{\uplb}{{\color{black}{\sf unmPL^{bin}}}}
\newcommand{\upl}{{\color{black}{\sf unmPL}}}
\newcommand{\uplba}[4]{{\color{black}{{\sf unmPL^{bin}}\;#1\;#2\;#3\;#4}}}
\newcommand{\upla}[2]{{\color{black}{{\sf unmPL}\;#1\;#2}}}
\newcommand{\upb}{{\color{black}{\sf unmP^{bin}}}}
\newcommand{\up}{{\color{black}{\sf unmP}}}

\newcommand{\upba}[3]{{\color{black}{{\sf unmP^{bin}}\;#1\;#2\;#3}}}
\newcommand{\upa}[2]{{\color{black}{{\sf unmP}\;#1\;#2}}}

\newcommand{\upEa}[2]{ {\color{black}{{\sf unmP_E}\;#1\;#2} } }

\newcommand{\bbeoa}[2]{{\color{black}{{\mathcal{B}^1_E}\;#1\;#2}}}
\newcommand{\bbeo}{\color{black}{\mathcal{B}^1_E}}
\newcommand{\bboa}[3]{{\color{black}{{\mathcal{B}^1_{#1}}\;#2\;#3}}}
\newcommand{\bbo}{\color{black}{\mathcal{B}^1}}
\newcommand{\bba}[2]{{\color{black}{\mathcal{B}\;#1\;#2}}}
\newcommand{\bbafor}[3]{{\color{black}{\mathcal{B}_{#1}\;#2\;#3}}}
\newcommand{\bb}{\color{black}{\mathcal{B}}}

\newcommand{\hToneone}{\dht{\ensuremath{H^1_1}}}
\newcommand{\hTonetwo}{\dht{\ensuremath{H^1_2}}}
\newcommand{\hTtwoone}{\dht{\ensuremath{H^2_1}}}
\newcommand{\hTtwotwo}{\dht{\ensuremath{H^2_2}}}

\newcommand{\proj}[2]{ #1 \upharpoonright_{#2}}
\newcommand{\projb}[2]{ #1 {\color{black}{\upharpoonright}}_{\color{black}{#2}} }
\newcommand{\projt}[2]{{#1}{\upharpoonright}{#2}}

\newcommand{\product}{{\textit {\lblFmt {prod}}}}
\newcommand{\price}{{\textit {\lblFmt {price}}}}
\newcommand{\publish}{{\textit {\lblFmt {publish}}}}
\newcommand{\dir}{{\roleFmt{d}}}
\newcommand{\ad}{{\roleFmt{ad}}}
\newcommand{\fone}{\ensuremath{\roleFmt{f}_{\roleCol{1}}}}
\newcommand{\ftwo}{\ensuremath{\roleFmt{f}_{\roleCol{2}}}}
\newcommand{\sales}{{\roleFmt{s}}}
\newcommand{\web}{{\roleFmt{w}}}

\newcommand{\gTstr}{\dgt{\ensuremath{G_{\sf str}}}}
\newcommand{\gTsales}{\dgt{\ensuremath{G_{\sf sales}}}}
\newcommand{\gTfin}{\dgt{\ensuremath{G_{\sf fin}}}}

\newcommand{\hTstr}{\dht{\ensuremath{H_{\sf str}}}}
\newcommand{\hTsales}{\dht{\ensuremath{H_{\sf sales}}}}
\newcommand{\hTfin}{\dht{\ensuremath{H_{\sf fin}}}}
\newcommand{\hTstrb}{\dht{\ensuremath{H'_{\sf str}}}}
\newcommand{\hTsalesb}{\dht{\ensuremath{H'_{\sf sales}}}}
\newcommand{\hTfinb}{\dht{\ensuremath{H'_{\sf fin}}}}

\newcommand{\ok}{{\textit {\lblFmt {ok}}}}
\newcommand{\stp}{{\textit {\lblFmt {stop}}}}
\newcommand{\wait}{{\textit {\lblFmt {wait}}}}
\newcommand{\go}{{\textit {\lblFmt {go}}}}

\newcommand{\oa}{{\roleFmt{oa}}}
\newcommand{\res}{{\roleFmt{res}}}
\newcommand{\ua}{{\roleFmt{ua}}}
\newcommand{\ow}{{\roleFmt{ow}}}

\newcommand{\init}{{\textit {\lblFmt {init}}}}
\newcommand{\deny}{{\textit  {\lblFmt {deny}}}}
\newcommand{\auth}{{\textit {\lblFmt {auth}}}}
\newcommand{\close}{{\textit {\lblFmt {close}}}}
\newcommand{\release}{{\textit {\lblFmt {release}}}}
\newcommand{\login}{{\textit {\lblFmt {login}}}}
\newcommand{\authcode}{{\textit {\lblFmt {code}}}}
\newcommand{\exchange}{{\textit {\lblFmt {exchange}}}}
\newcommand{\accesstoken}{{\textit {\lblFmt {token}}}}
\newcommand{\pass}{{\textit {\lblFmt {pass}}}}
\newcommand{\request}{{\textit {\lblFmt {request}}}}
\newcommand{\revoke}{{\textit {\lblFmt {revoke}}}}
\newcommand{\response}{{\textit {\lblFmt {response}}}}

\newcommand{\appid}{\dte{\code{id}}}
\newcommand{\scope}{\dte{\code{scp}}}
\newcommand{\name}{\dte{\code{name}}}
\newcommand{\pwd}{\dte{\code{pwd}}}
\newcommand{\cod}{\dte{\code{code}}}
\newcommand{\secret}{\dte{\code{secret}}}
\newcommand{\tkn}{\dte{\code{token}}}
\newcommand{\data}{\dte{\code{data}}}

\newcommand{\hresa}{\dht{\hT_{{\sf res\_acc}}}}

\newcommand{\hToaone}{\dht{\ensuremath{H_{\sf auth}}}}
\newcommand{\hToatwo}{\dht{\ensuremath{H_{\sf res}}}}
\newcommand{\ghToa}{\dgt{\ensuremath{G^\dagger_{\sf oa}}}}
\newcommand{\hresab}{\dht{\hT'_{{\sf res\_acc}}}}
\newcommand{\gGoa}{\dgt{\ensuremath{G^\dagger_{\sf op}}}}

\newcommand{\da}{\dht{\alpha}}
\newcommand{\dm}[2]{\dht{\ensuremath{#1\to#2}}}
\newcommand{\dme}[2]{\dht{\ensuremath{#1\to^{e} #2}}}
\newcommand{\dpi}{\dht{\pi^{\textsf{in}}}}
\newcommand{\dpii}{\dht{\pi^{\textsf{in}}_i}}
\newcommand{\dr}[2]{\dht{{#1}?{#2}}}
\newcommand{\dre}[2]{\dht{{#1}?^{e}{#2}}}
\newcommand{\dpo}{\dht{\pi^{\textsf{out}}}}
\newcommand{\dpoi}{\dht{\pi^{\textsf{out}}_i}}
\newcommand{\ds}[2]{\dht{{#1}!{#2}}}
\newcommand{\dse}[2]{\dht{{#1}!^{e}{#2}}}

\newcommand{\ms}{\dht{\sigma}}
\newcommand{\dtos}{\dht{\to^{\star}}}
\newcommand{\dexs}{\dht{!^{\star}}}
\newcommand{\dqus}{\dht{?^{\star}}}
\newcommand{\dqusi}{\dht{?^{\star_i}}}
\newcommand{\dms}[2]{\dht{\ensuremath{#1\dtos#2}}}
\newcommand{\dss}[2]{\dht{{#1}\dexs{#2}}}
\newcommand{\drs}[2]{\dht{{#1}\dqus{#2}}}

\newcommand{\hgchoice}[2]{\dht{\ensuremath{\boxplus_{i\in I}\da^{\p}_i(\lblFmt{#1_i});#2_i}}}

\newcommand{\hchoicep}[2]{\dht{\ensuremath{\boxplus_{i\in I}\ms^{\p}_i(\lblFmt{#1_i});#2_i}}}
\newcommand{\hun}[2]{\dht{\ensuremath{\bigvee_{i\in I}\pi^{\textsf{out}}_i(\lblFmt{#1_i});#2_i}}}
\newcommand{\hint}[2]{\dht{\ensuremath{\bigwedge_{i\in I}\pi^{\textsf{in}}_i(\lblFmt{#1_i});#2_i}}}
\newcommand{\fdel}{\dht{\ensuremath{\circ\langle\!\langle\bullet}}}
\newcommand{\bdel}{\dht{\ensuremath{\bullet\rangle\!\rangle\circ}}}
\newcommand{\hfdel}[3]{\dht{\ensuremath{#1\fdel#2;#3}}}
\newcommand{\hbdel}[3]{\dht{\ensuremath{#1\bdel#2;#3}}}
\newcommand{\hafdel}[3]{\dht{\ensuremath{#1\fdel[#2];#3}}}
\newcommand{\hpbdel}[3]{\dht{\ensuremath{[#1]\bdel#2;#3}}}
\newcommand{\hpfdel}[3]{\dht{\ensuremath{[#1]\fdel#2;#3}}}
\newcommand{\habdel}[3]{\dht{\ensuremath{#1\bdel[#2];#3}}}

\newcommand{\projdone}[3]{ #1 {\color{black}{\upharpoonright^{1}}}_{\color{black}{(#2,#3)}} }
\newcommand{\projdtwo}[3]{ #1 {\color{black}{\upharpoonright^{2}}}_{\color{black}{(#2,#3)}} }
\newcommand{\projd}[3]{ #1 {\color{black}{\upharpoonright^{d}}}_{\color{black}{#2,#3}} }
\newcommand{\de}{\mathit{dE}}
\newcommand{\pqe}{\{(\p,\q)\}}
\newcommand{\pqp}{(\p,\q)}

\newcommand{\disjpair}[3]{\texttt{disj\_pair}\;(#1,#2)\;#3}

\newcommand{\pzero}{{\roleFmt{p$\color{pine}{_0}$}}}

\newcommand{\qone}{{\roleFmt{q$\color{pine}{_1}$}}}
\newcommand{\bank}{{\roleFmt{bank}}}
\newcommand{\seller}{{\roleFmt{seller}}}
\newcommand{\buyer}{{\roleFmt{customer}}}
\newcommand{\card}{{\textit {\lblFmt {card}}}}
\newcommand{\Gd}{\dgt{G_d}}
\newcommand{\Gec}{\dgt{\G_{ec}}}
\newcommand{\user}{{\roleFmt{user}}}
\newcommand{\website}{{\roleFmt{website}}}
\newcommand{\shopeu}{{\roleFmt{EUshop}}}
\newcommand{\shopuk}{{\roleFmt{UKshop}}}
\newcommand{\location}{{\textit {\lblFmt {location}}}}
\newcommand{\open}{{\textit {\lblFmt {open}}}}
\newcommand{\fronteu}{{\textit {\lblFmt {EUfrontpage}}}}
\newcommand{\frontuk}{{\textit {\lblFmt {UKfrontpage}}}}

\ifdefined\NOCOLOR
  \newcommand{\chcolor}[1]{black} 
\else
  \newcommand{\chcolor}[1]{#1} 
\fi

           \newcommand{\appref}[1]{\ref{#1}}

\iftoggle{full}
         {
           \newcommand{\inApp}{}
         }
         {
           \newcommand{\inApp}{ of \cite{fullversion}}
           }

\AtEndPreamble{%
\theoremstyle{acmdefinition}
\newtheorem{remark}[theorem]{Remark}}

\newcommand{\myparagraph}[1]{\textbf{\emph{#1}}\ }

\AtBeginDocument{%
  \providecommand\BibTeX{{%
    \normalfont B\kern-0.5em{\scshape i\kern-0.25em b}\kern-0.8em\TeX}}}

\begin{document}

\title[Hybrid Multiparty Session Types - Full Version]{Hybrid Multiparty Session Types - Full Version}         
\subtitle{Compositionality for Protocol Specification through Endpoint Projection}                     


\author{Lorenzo Gheri}
\orcid{0000-0002-3191-7722}             
\affiliation{
  \institution{University of Liverpool}            
  \country{United Kingdom}                    
}
\email{lorenzo.gheri@liverpool.ac.uk}          

\author{Nobuko Yoshida}
\orcid{0000-0002-3925-8557}             
\affiliation{
  \institution{University of Oxford}            
  \country{United Kingdom}                    
}
\email{nobuko.yoshida@cs.ox.ac.uk}          

\begin{abstract}
Multiparty session types (MPST) are a specification and verification
framework for distributed message-passing systems.
The communication protocol of the system is
specified as a \emph{global type}, from which a
collection of \emph{local types} (local process implementations)
is obtained by \emph{endpoint projection}. 
A global type is a single disciplining entity for the whole system, specified by \emph{one designer} that has full knowledge of the communication protocol. On the other hand, distributed systems are often 
described in terms of their \emph{components}: a different designer is in charge of providing a subprotocol for each component. 
The problem of modular specification of global protocols has been
addressed in the literature, but the state of the art focuses only on
dual input/output
compatibility. Our work overcomes this limitation. We propose the
first MPST theory of \emph{multiparty compositionality for distributed
protocol specification} that is semantics-preserving, 
allows the composition of two or more
components, 
and retains full MPST expressiveness.
We introduce \emph{hybrid types} for describing subprotocols interacting
with each other, define a novel \emph{compatibility relation},
explicitly describe an algorithm for composing
multiple subprotocols into a \emph{well-formed global type},
and prove that compositionality preserves 
projection,
thus retaining semantic guarantees, such as liveness and deadlock freedom.
Finally, we test our work against real-world case studies and we
smoothly extend our novel compatibility to MPST with delegation
and explicit connections.

\end{abstract}

\ccsdesc[500]{Theory of computation~Distributed computing models}
\ccsdesc[500]{Theory of computation~Type theory}

\keywords{multiparty session types, compositionality, protocol design, concurrency}  

\maketitle

\section{Introduction}
\label{sec:intro}




With the current growth in
scale and complexity of systems,
their \emph{design}
has become of central importance
for industry and
society in general.
Choreographies for interactions
among multiple participants,
or \emph{(communication) protocols},
arise naturally in
numerous fields:
authorisation standards
\cite{oauth2,kerberos},
the BPMN graphical
specification for business processes
\cite{bpmn}, or
smart contracts for financial transactions
\cite{ethereum}.

The literature on programming languages
offers a variety of formal frameworks for
protocol description \cite{HYC2016,BLT-CA,chor-prog},
aimed at the verification of behavioural
properties of distributed implementations
that comply with the
communication discipline
prescribed by the protocol.
Such theories focus on distributed
implementations of participants, but
rarely feature modularity
in the design of protocols, which
are instead seen as standalone,
monolithic entities.
Mostly, when modularity is considered,
it is either conceived in terms of nesting
\cite{nested,aspects} or it substantially
modifies protocol description,
by adding additional structure
\cite{CMTV:2018,SGTV:2020,compchor}.
To the best of our knowledge,
only in 
\cite{BLTDezani} and \cite{smdg:facs21}
the result of composition
is a well-formed protocol.

\emph{
This paper presents hybrid multiparty session types:
a novel, general theory 
that offers compositionality for
distributed protocol specification,
improves on the state of the art, and
is immediately compatible with existing
multiparty session types systems.
}

Multiparty session types (MPST)
\cite{HYC2016,CDPY2015,verygentle}
provide a typing discipline for message-passing
concurrency, ensuring deadlock freedom for two or more
distributed processes.
A \emph{global type} or \emph{protocol},
which describes an entire interaction scenario,
is projected into a collection of \emph{local types}
onto the respective participants (endpoint projection).
MPST cater for the safe
implementation of distributed processes:
as long as the process for each participant
is independently type-checked against
its local type, its communication behaviour is disciplined by the
semantics of the global type, and its execution does not get stuck.


Although alternatives to the top-down approach (i.e., %
endpoint projection from a global type) have been proposed %
\cite{SY2019,DenielouYoshida2013,LTY2015,LangeYoshidaCAV19}, %
the benefits of an explicit, concise design
of the communication protocol for the whole system
have been recognised
by the research community, since
the first appearance of MPST
\cite{Honda2008Multiparty},
until more recent times,
e.g., see
\cite{GHH2021,cledou_et_al:LIPIcs.ECOOP.2022.27}.
Furthermore, the top-down approach
has been extended, e.g., to
fault tolerance \cite{10.1145/3485501},
timed specification \cite{BYY2014},
refinements \cite{ZFHNY2020,gheri_et_al:LIPIcs.ECOOP.2022.8},
cost awareness \cite{10.1145/3428223},
exception handling \cite{lagaillardie_et_al:LIPIcs.ECOOP.2022.4}, or
explicit connections and delegation \cite{HY2017,CASTELLANI2020128}.

Concretely, the underlying assumption
to top-down MPST systems is that a single designer has %
\emph{full knowledge} of the communication protocol
and can %
give its formal specification
in terms of a global type. %
Distributed systems, however, %
are designed modularly, by multiple designers.
Recently, the literature has addressed the
problem of obtaining a single coherent global type
from independently specified subprotocols
(components of a protocol)
and some solutions have been offered:
\citet{BLTDezani} achieve
direct composition of \emph{two} global types, through a
dual compatibility relation that matches inputs and outputs,
based on gateways \cite{gateways1,gateways2,BLT-isola2020}.
\citet{smdg:facs21} describe a
dual methodology beyond gateways, but 
severely restrict
the syntax for global types.
In contrast to this approach,
our theory substitutes
dual compatibility,
based only on input/output
matching, with the notion of
\emph{compatibility through projection}.
Thus, we improve on the state of the art:
\begin{enumerate*}
\item we can compose more than two subprotocols
into a well-formed global type and
\item we retain the full expressiveness
of MPST (including recursion and parallel composition).
\end{enumerate*}
See \S\ref{sec:related} for a broader, in-detail discussion.
Moreover, metathoretical results
about the semantics
of traditional MPST systems
\cite{DenielouYoshida2013,HYC2016}
immediately translate to ours
(\emph{semantics preservation}): 
from distributed specifications
in terms of subprotocols,
our theory synthesises a
global protocol for the whole system;
we prove once and for all,
as a metatheoretical result,
that such global protocol is
a traditionally well-formed global type.


\textbf{\emph{Contributions.}\ } This paper develops
\emph{a theory of compositionality for
distributed protocol description in MPST systems} and
 introduces the following novel MPST concepts:
\begin{itemize}
\item \emph{hybrid types}, a generalisation of both global and local types, for the specification of communicating subprotocols
(Definition \ref{def:hybrid-types});
\item \emph{generalised projection}
onto sets of roles (Definition \ref{def:proj}),
which well-behaves with respect to set inclusion (Theorem \ref{thm:proj-comp});
\item \emph{localiser} (Definition \ref{def:loc}),
a novel operator that isolates, in a subprotocol,
the inter-component communication
from the intra-component one;
\item \emph{compatibility} based on projection and localiser
(Equation \ref{eq:compatibility}, \S\ref{subsec:comp1});
\item \emph{build-back}, an explicit algorithm to compose
\emph{two or more} subprotocols into a more general one
(Definitions \ref{def:bb1} and \ref{def:bb} and
Theorems \ref{thm:comp1} and \ref{thm:comp}).
\end{itemize}
To the best of our knowledge, our approach is the first that:
\begin{itemize}
\item enables the correct composition
of \emph{two or more}
subprotocols into a global type, while capturing
\emph{full MPST expressiveness}: branching, parallel
composition, and recursion (Corollary \ref{cor:dist-spec});
\item operates at a purely syntactic level, thus
retaining previously developed MPST semantics results
(\emph{semantics preservation}); 
correctness is guaranteed by compositionality
resulting in a traditionally
well-formed global type and
preserving endpoint projection
(Corollary \ref{cor:dist-spec});
\item provides a notion of compatibility that is
\emph{more expressive than dual input/output matching}
and hence suitable for extension to
more sophisticated MPST systems (Example \ref{ex:del}).
\end{itemize}
We discuss the applicability
and generality of our work,
through \emph{case studies}.
\begin{enumerate*}
\item We give a distributed specification of
the real-world protocol OAuth 2.0 \cite{oauth2},
which showcases modularity features of our theory
(\S\ref{subsec:oauth})
and leads to an optimisation 
(\S\ref{subsec:oauth2}, Corollary \ref{cor:dist-spec2}).
\item We extend our theory beyond traditional MPST, to 
delegation and explicit connections (\S\ref{subsec:del}).
\end{enumerate*}

\emph{Outline.\ }\S\ref{sec:overview}
gives an overview of 
our development,
with a simple, but realistic,
application scenario. \S\ref{sec:hybrid} and \S\ref{sec:comp}
are dedicated to our technical contributions. \S\ref{sec:eval}
tests the strengths of our theory with
case studies. 
\S\ref{sec:related} discusses in detail, with examples,
related work. 
\S\ref{sec:conclusion} concludes with
future work. 
Further detail for definitions and proofs can be found
in Appendix \ref{sec:appendix}\inApp.

\section{Overview of Our Development}
\label{sec:overview}
This work achieves \emph{distributed protocol specification} for MPST:
different \emph{(protocol) designers}
specify protocols
(naively as global types, Figure \ref{fig:naive})
for different components of the
communicating system;
then, these compose into a
single global type
for the whole system.
Composition must \emph{preserve (endpoint) projection} (indicated with $\upharpoonright$):
local types,
for the distributed
implementation of 
\emph{roles} (or \emph{participants}),
need to be obtained 
by projection of each separate
component, but, also,
they need to be projections of
the same global type (obtained by composition), if we want
semantic guarantees (e.g., deadlock freedom) to hold.
In other words, our protocol-compositionality
theory relies on multiparty compatibility,
guaranteed by a well-formed global type, and
on semantics proofs from previous work
(e.g., \cite{DenielouYoshida2013}).
This approach
makes our development
\emph{semantics-preserving}: 
it endows existing MPST systems
with 
distributed protocol specification. 



\begin{figure}
  \begin{center}
    \pgfdeclarelayer{background}
\pgfsetlayers{background,main}

\begin{tikzpicture}[commutative diagrams/every diagram]


	\node(G1)at (0.5,0) {$\dgt{G_1}$};
        \node(gdots)at (1.5,0){$\dots$};
	\node(Gn)at (2.5,0) {$\dgt{G_n}$};
	\node(L11)at (0,-1.5) {$\dlt{L_{11}}$};
        \node(l1dots)at (0.5,-1.5) {$\dots$};
	\node(L1k)at (1,-1.5) {$\dlt{L_{1k_1}}$};
        \node(ldots)at (1.5,-1.5){$\dots$};
        \node(Ln1)at (2,-1.5) {$\dlt{L_{n1}}$};
        \node(lndots)at (2.5,-1.5) {$\dots$};
	\node(Lnk)at (3,-1.5) {$\dlt{L_{nk_n}}$};
        \node(nol)at (3,0) {};
        \node(nor)at (5,0) {};
        \node(eq)at (3.75,-1.5) {${\bf \color{magenta}{=}}$};
        \node(G)at (5.75,0) {$\G$};
        \node(L11b)at (4.5,-1.5) {$\dlt{L_{11}}$};
        \node(l1dotsb)at (5,-1.5) {$\dots$};
	\node(L1kb)at (5.5,-1.5) {$\dlt{L_{1k_1}}$};
        \node(ldotsb)at (6.25,-1.5) {$\dots$};
        \node(Lnkb)at (7,-1.5) {$\dlt{L_{nk_n}}$};


	\path[commutative diagrams/.cd,every arrow]
	(G1) edge node[left] {$\upharpoonright$} (L11)
        (G1) edge node[right] {$\upharpoonright$} (L1k)
        (Gn) edge node[left] {$\upharpoonright$} (Ln1)
        (Gn) edge node[right] {$\upharpoonright$} (Lnk)
	(G) edge node[left,xshift=-1mm] {$\upharpoonright$} (L11b)
        (G) edge node[left] {$\upharpoonright$} (L1kb)
        (G) edge node[right,xshift=1mm] {$\upharpoonright$} (Lnkb)
        ;
        \path[commutative diagrams/.cd,every arrow, thick, magenta, font=\footnotesize]
        (nol) edge node[below,yshift=-1mm] {composition} (nor);
\end{tikzpicture}
  \end{center}
  \caption{Distributed Protocol Specification, Naively}
  \label{fig:naive}
\end{figure}


Traditionally, a global type is a ``closed''
standalone entity that describes a
one-component communication protocol: all
interactions among participants are
\emph{internal} to such component.
We consider instead the distributed
specification of a system, in terms of
multiple components
(disjoint sets of participants).
Each participant can send both internal messages,
within its component, or \emph{external}, to
other components.
Therefore, we ``open''
the syntax of global types,
so that it allows not only for
intra-component communication,
but also for
inter-component communication.
By extending the syntax of
global types with an interface for
inter-component communication, we obtain
\emph{hybrid types}.
The communication protocol of each component of
the system is specified as a hybrid type;
multiple components can be composed into
a well-formed global type thanks to
a novel notion of
\emph{compatibility, based on projection.}

In what follows: we consider a three-component system:
a company with three departments, for each of which, a different
\emph{(protocol) designer} is in charge of describing the communication protocol.
The departments, with respective (internal) roles, are the following:
\begin{itemize*}
\item[(a)] the \emph{strategy team}, the roles of which are the director $\dir$ of the company and the advertisement team $\ad$;
\item[(b)] the \emph{sales department}, with a salesman $\sales$ and the website administrator $\web$;
\item[(c)] the \emph{finance department}, with two employees, $\fone$ and $\ftwo$.
\end{itemize*}
%
We assume that internal roles of different components are \emph{distinct}.

\myparagraph{Global Types for Intra-Component Communication.\ }
When no inter-component communication happens,
each protocol designer gives
a global type for the \emph{internal} communication of their
department (Figure \ref{fig:co-global}). In $\gTstr$,
the global type for 
the strategy department, the director $\dir$ sends the
product ID to the responsible for advertisement $\ad$; then, $\dir$
gives an $\ok$ or asks $\ad$ to $\stp$. For the sales
department ($\gTsales$)
$\sales$ decides whether $\web$ can publish
some content on the company website. In the financial
department ($\gTfin$), $\fone$ sends the product ID to $\ftwo$ and gets back
either a $\price$ or a $\stp$.

\begin{figure}
  \begin{subfigure}[t]{\textwidth}
  \begin{small}
\[
\begin{array}{l}
  \gTstr = \msg \dir \ad \dgt{\product(\tnat). \msg \dir \ad\{\ok.\gend,\stp.\gend\}}\qquad\\
  \gTsales = \msg \sales \web
  \dgt{\{\publish.\gend,\stp.\gend\}}\\
  \gTfin = \msg \fone \ftwo \dgt{\product(\tnat).\msg \ftwo \fone \{\price(\tnat).\gend,\stp.\gend\}}
\end{array}
\]
  \end{small}
\caption{Global Types for Internal Communication}
  \label{fig:co-global}
  \end{subfigure}
    \begin{subfigure}[t]{\textwidth}
    \vspace*{2mm}
      \begin{small}
        \[
        \begin{array}{l}
          \hTstrb  = \hmsg \dir\ad\dht{\product(\tnat).\hla{\hsnd \dir \sales \product(\tnat)}.}
          \dht{\hlb{\hsnd \dir \fone \product(\tnat)}.}
           \dht{\hmsg \dir \ad \{\ok.\hend,\stp.\hend\}}
\\
          \hTsalesb  =  \dht{\hla{\hrcv \dir \sales \product(\tnat)}.
            \hmsg \sales \web \{\publish.\hend,\stp.\hend\}}\\
          \hTfinb   = \dht{\hlb{\hrcv \dir \fone \product(\tnat)}.
            \hmsg \fone \ftwo \product(\tnat).}
            \dht{\hmsg \ftwo \fone \{\price(\tnat).\hend,\stp.\hend\}}
        \end{array}
        \]
      \end{small}
      \caption{Hybrid Types for Basic Inter-Department Interactions}
\label{fig:co-hybrid1}
    \end{subfigure}
       \begin{subfigure}[t]{\textwidth}
       \vspace*{1mm}
      \begin{small}
        \[
  \begin{array}{l}
    \hTstr=\hmsg\dir\ad \dht{\product(\tnat).\hsnd \dir \sales \product(\tnat).\hsnd \dir \fone \product(\tnat).}
    \hrec \hX {\hrcv \fone \dir}
    \dht{\left\{
        \ok.\hmsg\dir\ad\go.\hend,
        \wait.\hmsg\dir\ad\wait.\hX
      \right\}}
  \\
  \hTsales = \hrcv \dir \sales \dht{\product(\tnat).
  \hrec \hX {\hrcv \fone\sales}}
  \dht{\left\{
      \price(\tnat).\hmsg \sales \web \publish.\hend,
      \wait.\hmsg \sales \web \wait.\hX
    \right\}}\\
      \hTfin = \hrcv \dir \fone \dht{\product(\tnat).
      \hmsg \fone \ftwo \product(\tnat).
      }
    \dht{\hrec \hX {\hmsg \ftwo \fone}\left\{
      \begin{array}{l}
        \price(\tnat).\hsnd\fone\dir\ok.\hsnd\fone\sales\price(\tnat).\hend ,\\
        \wait.\hsnd\fone\dir\wait.\hsnd\fone\sales\wait.\hX
      \end{array}
      \right\}}\\
  \end{array}
  \]
  \end{small}
  \subcaption{More Expressive Hybrid Types}
  \label{fig:co-hybrid2}
  \end{subfigure}
    \caption{Types for Interactions in the Company}
    \label{fig:company}
\end{figure}

\myparagraph{Hybrid Types for Inter-Component Interactions.}
The components of a distributed system are expected
to communicate 
with each other.
Therefore, we introduce a
\emph{hybrid syntax of global and local constructs}
(and we call \emph{hybrid types} the terms of this syntax):
to the global-type syntax (e.g., $\nhmsg \dir \ad$), 
for intra-component communication,
we add
local send and receive constructs
(e.g., $\nhsnd \dir \sales$ and $\nhrcv \dir \fone$),
as the interface for
inter-component communication.
In our example,
a first message is sent by $\dir$,
with a product ID $\product$,
externally,
to the other two departments (Figure \ref{fig:co-hybrid1}): 
$\nhsnd \dir \sales$ and $\nhsnd \dir \fone$. These are \emph{dually}
received by the sales team $\nhrcv \dir \sales$ and by the finance
team $\nhrcv \dir \fone$, respectively
(as highlighted in Figure~\ref{fig:co-hybrid1}). 

\begin{remark}[Generalising Global and Local Types]
\label{rem:gloc-1}
We observe that hybrid types are a
generalisation of both global and local types.
A global type is a ``closed'' hybrid type,
where only internal messages are exchanged.
The intuition for local types is more subtle:
a local type can be interpreted
as a basic, \emph{one-participant component}
of a communicating system,
which
communicates only externally, with
participants of other components.
E.g., the local type
$\dlt{\lrcv \p \lblFmt{\ell_1}.}\dlt{\lsnd\pr \lblFmt{\ell_2}.\lend}$,
for the participant $\q$,
can be written as the hybrid type
$\dht{\hrcv \p \q \lblFmt{\ell_1}.}\dht{\hsnd \q \pr \lblFmt{\ell_2}.\hend}$:
$\q$ is the only internal participant that first receives
from $\p$ and then sends to $\pr$
(both $\p$ and $\pr$ are of other components).
Being able to
express global and local types
as hybrid types is fundamental:
it makes our results
correct and compatible
with existing MPST theories
(see Remark \ref{rem:gloc-2} below, and
Corollary \ref{cor:dist-spec}
in \S\ref{sec:comp}).
\end{remark}

\myparagraph{Expressiveness and Compatibility.} We describe a
more expressive version of the protocols (Figure \ref{fig:co-hybrid2})
that combines inter-component messages with branching and recursion.
Figure \ref{fig:diag} shows the communication for each
component of the system, as described by the protocol
designer of each department.
We imagine that the price of the product $\product$ is decided
within the finance department: the finance expert $\ftwo$
either gives a $\price$ or asks all processes to $\wait$
in a recursive loop; then, the decision is communicated to
the other departments. Figure \ref{fig:diag-fin} shows
the execution of the protocol for the finance
department, where $\ftwo$ makes such choice.
Figure \ref{fig:co-hybrid2} shows the formal
specifications, as hybrid types, of the three protocols.
%
We observe that, to compose $\hTstr$, $\hTsales$, and $\hTfin$
(and, in general, to compose, more than two communicating protocols),
dual relations are not sufficient for compatibility
(for a broader discussion see \S\ref{sec:related}). %
%
Our proposal is 
to give separately the specification of a
\emph{communication discipline for inter-component
interactions only}:
intra-component interactions are left
to the designer of each respective component and 
some \emph{chief designer}
gives the description of one more protocol,
for global guidance of inter-component communication.
For our example, we collect all the interactions
between any two different departments in the
protocol in Figure \ref{fig:diag-ext}, and we formalise
it with a \emph{compatibility (global) type}:\\[1mm]
\begin{small}
\centerline{$
  \begin{array}{l}
    \dgt{\G^\dagger} :=\hmsg \dir \sales \dht{\product(\tnat).}
    \dht{\hmsg \dir \fone \product(\tnat).}
    \dht{\hrec \hX {\hmsg\fone\dir}}\dht{\left\{
\begin{array}{l}
      \ok.\hmsg \fone \sales \price(\tnat).\hend,\\
      \wait.\hmsg \fone \sales \wait.\hX
    \end{array}
    \right\}}
 \end{array}
$}
\end{small}

\emph{Compatibility} of subprotocols $\hTstr$, $\hTsales$, and $\hTfin$,
with $\dgt{\G^\dagger}$, is achieved by asking that the
\emph{(generalised) projection} $\projb{}{}$, of $\dgt{\G^\dagger}$, with
respect to the internal participants of each subprotocol,
is equal to the \emph{localisation} $\loc{}{}\!\!$, of that subprotocol,
where ``localising a protocol'' means isolating
its inter-component communication
(by retaining only its \emph{local}
constructs). 
E.g., we consider $\hTstr$ and its internal participants $\{\dir,\ad\}$. 
\begin{small}
\[
\projb {\dgt{\G^\dagger}} {\{\dir,\ad\}} = \loc {} \hTstr =
\hsnd \dir \sales \dht{\product(\tnat). \hsnd \dir \fone \product(\tnat).\hrec \hX {\hrcv\fone\dir}}
  \dht{\left\{
      \ok.\hend,
      \wait.\hX
    \right\}}
\]
\end{small}
Analogously we require that
$\projb {\dgt{\G^\dagger}} {\{\sales,\web\}} = \loc {} \hTsales$
and $\projb {\dgt{\G^\dagger}} {\{\fone,\ftwo\}} = \loc {} \hTstr$.

We observe that not only we have enriched the syntax of global
types with local constructs to get hybrid types, but also we have
\emph{generalised projection} to \emph{sets} of participants,
introduced a new operator (\emph{localiser}) to isolate
external communication, and, based on these, defined compatibility.

\begin{figure*}
  \begin{subfigure}{.32\textwidth}
    \begin{center}
      \begin{tikzpicture}[font=\footnotesize]

\node[draw,
  minimum width=0.8cm,
  minimum height=0.4cm,
  rounded corners] (d) at (0,0)
     {$\dir$};

\node[draw,
  minimum width=0.8cm,
  minimum height=0.4cm,
  rounded corners] (ad) at (1.2,0)
     {$\ad$};

\node[draw,
  minimum width=0.8cm,
  minimum height=0.4cm,
  dotted, thick,
  rounded corners] (s) at (2.4,0)
     {$\sales$};

\node[draw,
  minimum width=0.8cm,
  minimum height=0.4cm,
  dotted, thick,
  rounded corners] (f1) at (3.4,0)
     {$\fone$};

\draw[-](d) edge (0,-3.4);
\draw[-](ad) edge (1.2,-3.4);
\draw[-,dotted,thick](s) edge (2.4,-4.9);
\draw[-,dotted,thick](f1) edge (3.4,-4.9);
\draw[-](0,-3.8) edge (0,-4.9);
\draw[-](1.2,-3.8) edge (1.2,-4.9);
\draw[-{Stealth}] (0,-0.9) --node[above]{$\product$} (1.2,-0.9);
\draw[-{Triangle[open]},thick,cbred] (0,-1.5) --node[above,xshift=-5mm]{$\product$} (2.4,-1.5);
\draw[-{Triangle[open]},thick,cbred] (0,-2.1) --node[above]{$\product$} (3.4,-2.1);
\draw[-{Triangle[open]},thick,cbred] (3.4,-2.7) --node[above]{$\ok$} (0,-2.7);
\draw[-{stealth}] (0,-3.2) --node[above]{$\go$} (1.2,-3.2);
\draw[-{Triangle[open]},thick,cbred] (3.4,-4) --node[above]{$\wait$} (0,-4);
\draw[-{Stealth}] (0,-4.5) --node[above]{$\wait$} (1.2,-4.5);
\draw[decorate,decoration={brace},thick] (3.5,-2.7) -- (3.5,-4);
\draw[-,thick] (-0.4,-4.7) -- (-0.1,-4.7);
\draw[-,thick] (-0.4,-4.7) -- (-0.4,-2.4);
\draw[->,thick] (-0.4,-2.4) -- (-0.1,-2.4);
\draw[-,dashed,gray,thin] (0,-2.4) -- (3.4,-2.4);
\draw[-,dashed,gray,thin] (0,-4.7) -- (3.4,-4.7);

\end{tikzpicture}
    \end{center}
    \subcaption{Strategy Department}
    \label{fig:diag-str}
  \end{subfigure}
  \begin{subfigure}{.32\textwidth}
    \begin{center}
      \begin{tikzpicture}[font=\footnotesize]

\node[draw,
  minimum width=0.8cm,
  minimum height=0.4cm,
  rounded corners] (s) at (0,0)
     {$\sales$};

\node[draw,
  minimum width=0.8cm,
  minimum height=0.4cm,
  rounded corners] (w) at (1.2,0)
     {$\web$};

\node[draw,
  minimum width=0.8cm,
  minimum height=0.4cm,
  dotted, thick,
  rounded corners] (d) at (2.4,0)
     {$\dir$};

\node[draw,
  minimum width=0.8cm,
  minimum height=0.4cm,
  dotted, thick,
  rounded corners] (f1) at (3.4,0)
     {$\fone$};

\draw[-](s) edge (0,-3.1);
\draw[-](w) edge (1.2,-3.1);
\draw[-,dotted, thick](d) edge (2.4,-4.9);
\draw[-,dotted, thick](f1) edge (3.4,-4.9);
\draw[-](0,-3.5) edge (0,-4.9);
\draw[-](1.2,-3.5) edge (1.2,-4.9);
\draw[-{Triangle[open]},thick,cbred] (2.4,-1) --node[above,xshift=5mm]{$\product$} (0,-1);
\draw[-{Triangle[open]},thick,cbred] (3.4,-2.1) --node[above]{$\price$} (0,-2.1);
\draw[-{stealth}] (0,-2.8) --node[above]{$\publish$} (1.2,-2.8);
\draw[-{Triangle[open]},thick,cbred] (3.4,-3.8) --node[above]{$\wait$} (0,-3.8);
\draw[-{Stealth}] (0,-4.5) --node[above]{$\wait$} (1.2,-4.5);
\draw[decorate,decoration={brace},thick] (3.5,-2.1) -- (3.5,-3.8);
\draw[-,thick] (-0.4,-4.7) -- (-0.1,-4.7);
\draw[-,thick] (-0.4,-4.7) -- (-0.4,-1.6);
\draw[->,thick] (-0.4,-1.6) -- (-0.1,-1.6);
\draw[-,dashed,gray,thin] (0,-1.6) -- (3.4,-1.6);
\draw[-,dashed,gray,thin] (0,-4.7) -- (3.4,-4.7);

\end{tikzpicture}
    \end{center}
    \subcaption{Sales Department}
    \label{fig:diag-sales}
  \end{subfigure}
  \begin{subfigure}{.32\textwidth}
    \begin{center}
      \begin{tikzpicture}[font=\footnotesize]

\node[draw,
  minimum width=0.8cm,
  minimum height=0.4cm,
  rounded corners] (f1) at (0,0)
     {$\fone$};

\node[draw,
  minimum width=0.8cm,
  minimum height=0.4cm,
  rounded corners] (f2) at (1.2,0)
     {$\ftwo$};

\node[draw,
  minimum width=0.8cm,
  minimum height=0.4cm,
  dotted, thick,
  rounded corners] (d) at (2.4,0)
     {$\dir$};

\node[draw,
  minimum width=0.8cm,
  minimum height=0.4cm,
  dotted, thick,
  rounded corners] (s) at (3.4,0)
     {$\sales$};

\draw[-](f1) edge (0,-3);
\draw[-](f2) edge (1.2,-3);
\draw[-,dotted, thick](d) edge (2.4,-4.9);
\draw[-,dotted, thick](s) edge (3.4,-4.9);
\draw[-](0,-3.3) edge (0,-4.9);
\draw[-](1.2,-3.3) edge (1.2,-4.9);
\draw[-{Triangle[open]},thick,cbred] (2.4,-0.8) --node[above,xshift=5mm]{$\product$} (0,-0.8);
\draw[-{Stealth}] (0,-1.3) --node[above]{$\product$} (1.2,-1.3);
\draw[-{Stealth}] (1.2,-1.8) --node[above]{$\price$} (0,-1.8);
\draw[-{Triangle[open]},thick,cbred] (0,-2.3) --node[above,xshift=-5mm]{$\ok$} (2.4,-2.3);
\draw[-{Triangle[open]},thick,cbred] (0,-2.8) --node[above,xshift=2mm]{$\price$} (3.4,-2.8);
\draw[-{Stealth}] (1.2,-3.5) --node[above]{$\wait$} (0,-3.5);
\draw[-{Triangle[open]},thick,cbred] (0,-4) --node[above]{$\wait$} (3.4,-4);
\draw[-{Triangle[open]},thick,cbred](0,-4.5) --node[above]{$\wait$} (3.4,-4.5);
\draw[decorate,decoration={brace},thick] (1.3,-1.8) -- (1.3,-3.5);
\draw[-,thick] (-0.4,-4.7) -- (-0.1,-4.7);
\draw[-,thick] (-0.4,-4.7) -- (-0.4,-1.6);
\draw[->,thick] (-0.4,-1.6) -- (-0.1,-1.6);
\draw[-,dashed,gray,thin] (0,-1.6) -- (3.4,-1.6);
\draw[-,dashed,gray,thin] (0,-4.7) -- (3.4,-4.7);

\end{tikzpicture}
    \end{center}
    \subcaption{Finance Department}
    \label{fig:diag-fin}
  \end{subfigure}
  \begin{subfigure}[t]{\textwidth}
\vspace*{2mm}
\begin{center}
  \begin{tikzpicture}[font=\footnotesize]

\node[draw,
  minimum width=0.8cm,
  minimum height=0.4cm,
  rounded corners] (p) at (0,0)
     {$\p$};

\node[] (pe) at (1.6,0){internal participant};


\node[draw,
  minimum width=0.8cm,
  minimum height=0.4cm,
  dotted, thick,
  rounded corners] (q) at (3.8,0)
     {$\q$};

\node[] (qe) at (5.5,0){external participant};

\node[] (ie) at (1.6,-0.7){internal interaction};

\node[] (ee) at (5.5,-0.7){external interaction};

\node[] (ch) at (8.6,0){choice};

\node[] (rec) at (9.1,-0.7){loop (recursion)};

\draw[-{Stealth}] (-0.3,-0.7) -- (0.3,-0.7);
\draw[-{Triangle[open]},thick,cbred] (3.4,-0.7) -- (4,-0.7);
\draw[decorate,decoration={brace},thick] (7.8,0.2) -- (7.8,-0.2);
\draw[-,thick] (7.3,-1) -- (7.6,-1);
\draw[-,thick] (7.3,-1) -- (7.3,-0.4);
\draw[->,thick] (7.3,-0.4) -- (7.6,-0.4);
\draw[-,dashed,gray,thin] (7.7,-0.4) -- (8,-0.4);
\draw[-,dashed,gray,thin] (7.7,-1) -- (8,-1);



\end{tikzpicture}
 \end{center}
  \end{subfigure}
  \caption{Communication for Each Department in the Company}
  \label{fig:diag}
\end{figure*}

\myparagraph{Compositionality and Correctness.}
Our theory (\S\ref{sec:hybrid} and \S\ref{sec:comp})
provides an explicit function $\bb$ that \emph{builds back}
a single global type for
the communication in the company,
from the distributed specification above:
$\G=\bba {(\dgt{\G^\dagger})} {([\hTstr,\hTsales,\hTfin])}$.
It holds that
$\projb \G \dir = \projb \hTstr \dir$, $\projb \G \fone = \projb \hTfin \fone$,
and analogously for all participants:
\emph{projection is preserved}. A figure,
representing such $\G$ for our example, can be found in
Appendix \appref{sec:appendix-fig}\inApp.

More generally (see Figure \ref{fig:overview}),
from a compatibility type $\gG$ and
hybrid types $\hTi$ for each component
(with set of internal participants $E_i$),
such that they are compatible
$\proj\gG {E_i}=\loc{}\hTi$,
our theory synthesises a 
global type 
$\G=\bba {(\dgt{\G^\dagger})} {([\hTone,\dots,\dht{\hT_n}])}$
(Definition \ref{def:bb} and Theorem \ref{thm:comp}).
Correctness of our theory
is given by Corollary \ref{cor:dist-spec}.
Formally, this result guarantees that
\emph{the local types, projections of $\G$ on each participant,
are the same as if obtained by the respective subprotocol $\hTi$}:
$\projb \G {\p}= \projb \hTi {\p}$ if $\p$ is a participant
of the $i$-th subprotocol $\hTi$.
We have achieved \emph{distributed protocol specification}:
we can both
obtain local types for implementation
in a distributed fashion,
by projection of the
respective component $\hTi$
(no designer or programmer
needs the full knowledge
of $\G$),
and rest assured that all local types
(for all participants, in all components)
are projections of a
single, well-formed global type.
This makes our development
compatible with existing MPST theories,
with no need for developing new semantics
(\emph{semantics preservation}): 
a well-formed global type
projecting on all participants
gives traditional \emph{multiparty compatibility},
which, thanks to the semantics results in the literature
\cite{DenielouYoshida2013,CDPY2015,HYC2016},
leads to guarantees, such
as liveness and deadlock freedom.


\begin{remark}[Hybrid Types and Generalised Projection]
\label{rem:gloc-2}
With reference to Figure \ref{fig:overview},
Let us consider the set of participants $E_i$
of the generic $\hTi$, and $\p\in E_i$.
Generalised projection takes a hybrid type
and returns a hybrid type; since \emph{global and local
types are hybrid types} (Remark \ref{rem:gloc-1}),
e.g., we can project $\gG$ onto $E_i$
for compatibility ($\proj\gG {E_i}=\loc{}\hTi$), or $\G$ onto $E_i$ and
verify that it is equal to $\hTi$
(see Theorem \ref{thm:comp}, \S\ref{sec:comp}).
Most importantly, Theorem \ref{thm:proj-comp} in \S\ref{sec:comp}
guarantees that projection composes over set inclusion
$\projb \G {\p}=\projb {(\projb\G{E_i})} {\p} = \projb \hTi {\p}$:
by projecting $\hTi$ and $\G$ onto the participant $\p$,
we obtain the same local type for $\p$.
Namely, we can obtain local types
from the specific component $\hTi$
and then \emph{implement them in a distributed fashion},
but also 
they all are
projections of a well-formed global type $\G$.
\end{remark}

To summarise (see Figure \ref{fig:overview}),
our proposal for distributed protocol specification
is the following:
\begin{enumerate}
\item a different designer specifies, for each component of the system, a hybrid type $\hTi$;
\item a chief designer gives the compatibility type $\dgt{\G^\dagger}$
to discipline inter-component interactions;
\item compatibility is a simple equality check: 
$\projb {\dgt{\G^\dagger}} {E_i} = \loc {} \hTi$ ($E_i$ is the set of participants for $\hTi$).
\end{enumerate}
In return, as a metatheoretical result
proved once and for all by our theory,
the designers obtain
$\bba {(\dgt{\G^\dagger})} {([\hTone,\dots,\dht{\hT_n}])}$,
\emph{a global type for the whole communication},
for which \emph{projections are preserved}
(and, hence, MPST semantic guarantees hold). 

In \S\ref{sec:hybrid} and \S\ref{sec:comp} we detail
our compositionality theory,
including generalised projection, localiser,
compatibility, build-back, and correctness results.

\begin{figure*}
\begin{subfigure}{0.28\textwidth}
  \begin{center}
    \begin{tikzpicture}[font=\footnotesize]

\node[draw,
  minimum width=0.8cm,
  minimum height=0.4cm,
  rounded corners] (d) at (0,0)
     {$\dir$};

\node[draw,
  minimum width=0.8cm,
  minimum height=0.4cm,
  rounded corners] (s) at (1.2,0)
     {$\sales$};

\node[draw,
  minimum width=0.8cm,
  minimum height=0.4cm,
  rounded corners] (f1) at (2.4,0)
     {$\fone$};

\draw[-](d) edge (0,-2.5);
\draw[-](s) edge (1.2,-2.5);
\draw[-](f1) edge (2.4,-2.5);
\draw[-](0,-2.8) edge (0,-3.9);
\draw[-](1.2,-2.8) edge (1.2,-3.9);
\draw[-](2.4,-2.8) edge (2.4,-3.9);
\draw[-{Triangle[open]},thick,cbred] (0,-0.8) --node[above]{$\product$} (1.2,-0.8);
\draw[-{Triangle[open]},thick,cbred] (0,-1.3) --node[above,xshift=-5mm]{$\product$} (2.4,-1.3);
\draw[-{Triangle[open]},thick,cbred] (2.4,-1.8) --node[above,xshift=5mm]{$\ok$} (0,-1.8);
\draw[-{Triangle[open]},thick,cbred] (2.4,-2.3) --node[above]{$\price$} (1.2,-2.3);
\draw[-{Triangle[open]},thick,cbred] (2.4,-3) --node[above,xshift=5mm]{$\wait$} (0,-3);
\draw[-{Triangle[open]},thick,cbred] (2.4,-3.5) --node[above]{$\wait$} (1.2,-3.5);
\draw[decorate,decoration={brace},thick] (2.5,-1.8) -- (2.5,-3);
\draw[-,thick] (-0.4,-3.7) -- (-0.1,-3.7);
\draw[-,thick] (-0.4,-3.7) -- (-0.4,-1.6);
\draw[->,thick] (-0.4,-1.6) -- (-0.1,-1.6);
\draw[-,dashed,gray,thin] (0,-1.6) -- (2.4,-1.6);
\draw[-,dashed,gray,thin] (0,-3.7) -- (2.4,-3.7);

\end{tikzpicture}
  \end{center}
\subcaption{Inter-Component\\ Communication}
  \label{fig:diag-ext}
\end{subfigure}
\begin{subfigure}{0.68\textwidth}
  \begin{center}
\begin{small}
    \pgfdeclarelayer{background}
\pgfsetlayers{background,main}
\usetikzlibrary{decorations.pathreplacing}

\begin{tikzpicture}[commutative diagrams/every diagram]
        \node(gH)at (3,0) {$\gG$};
        \node(HL1)at (0.5,-1) {$\dht{\hT'_1}$};
        \node(HL2)at (2.1,-1) {$\dht{\hT'_2}$};
        \node(ghdots)at (3.3,-1){$\dots$};
        \node(HLn)at (4.5,-1) {$\dht{\hT'_n}$};
	\node(H1)at (0.5,-2) {$\hTone$};
	\node(H2)at (2.1,-2) {$\hTtwo$};
        \node(hdots)at (3.3,-2){$\dots$};
	\node(Hn)at (4.5,-2) {$\hTn$};
	\node(L11)at (0,-3) {$\dlt{L_{11}}$};
        \node(l1dots)at (0.5,-3) {$\dots$};
	\node(L1k)at (1,-3) {$\dlt{L_{1k_1}}$};
        \node(L21)at (1.6,-3) {$\dlt{L_{21}}$};
        \node(l2dots)at (2.1,-3) {$\dots$};
	\node(L2k)at (2.6,-3) {$\dlt{L_{2k_2}}$};
        \node(ldots)at (3.3,-3){$\dots$};
        \node(Ln1)at (4,-3) {$\dlt{L_{n1}}$};
        \node(lndots)at (4.5,-3) {$\dots$};
	\node(Lnk)at (5,-3) {$\dlt{L_{nk_n}}$};
        \node(noup)at (5,0.2){};
        \node(nodown)at (5,-2.2){};
        \node(nol)at (5.2,-1) {};
        \node(nor)at (7.5,-0.2) {};
        \node(eq)at (5.7,-3) {${\bf \color{magenta}{=}}$};
        \node(G)at (7.65,0) {$\G$};
        \node(L11b)at (6.4,-3) {$\dlt{L_{11}}$};
        \node(l1dotsb)at (6.9,-3) {$\dots$};
	\node(L1kb)at (7.4,-3) {$\dlt{L_{1k_1}}$};
        \node(L21b)at (7.9,-3) {$\dlt{L_{21}}$};
        \node(ldotsb)at (8.5,-3) {$\dots$};
        \node(Lnkb)at (9.1,-3) {$\dlt{L_{nk_n}}$};

	\path[commutative diagrams/.cd,every arrow]
        (gH) edge node[left,yshift=1.2mm] {$\upharpoonright$} (HL1)
        (gH) edge node[right] {$\upharpoonright$} (HL2)
        (gH) edge node[right,yshift=1.2mm] {$\upharpoonright$} (HLn)
        (H1) edge node[left] {$\loc{}{}$} (HL1)
        (H2) edge node[left] {$\loc{}{}$} (HL2)
        (Hn) edge node[left] {$\loc{}{}$} (HLn)
	(H1) edge node[left,yshift=1mm] {$\upharpoonright$} (L11)
        (H1) edge node[right,yshift=1mm] {$\upharpoonright$} (L1k)
        (H2) edge node[left,yshift=1mm] {$\upharpoonright$} (L21)
        (H2) edge node[right,yshift=1mm] {$\upharpoonright$} (L2k)
        (Hn) edge node[left,yshift=1mm] {$\upharpoonright$} (Ln1)
        (Hn) edge node[right,yshift=1mm] {$\upharpoonright$} (Lnk)
	(G) edge node[left,xshift=-1mm] {$\upharpoonright$} (L11b)
        (G) edge node[left] {$\upharpoonright$} (L1kb)
        (G) edge node[right] {$\upharpoonright$} (L21b)
        (G) edge node[right,xshift=1mm] {$\upharpoonright$} (Lnkb)
        ;
        \path[commutative diagrams/.cd,every arrow, thick, magenta, font=\footnotesize]
        (nol) edge node[above,yshift=2mm,xshift=-4mm] {composition} (nor);
\draw [decorate,decoration={brace}, thick, magenta](noup) -- (nodown) node [right] {};

\end{tikzpicture}
  \end{small}
\end{center}
  \subcaption{Distributed Protocol Specification, Compatibility through Projection}
  \label{fig:overview}
\end{subfigure}
\caption{Inter-Component Communication and
Compatibility via Projection}
\end{figure*}

\section{Hybrid Types for Protocol Specification}
\label{sec:hybrid}

\subsection{Background: Preliminaries of Multiparty Session Types}

We give a short summary of \emph{multiparty session types}
\cite{HYC2016,Scalas:2019,CDPY2015,verygentle};
specifically, our theory
is based on the formulation
in \cite{DenielouYoshida2013},
extended to parallel composition of global types.
The notation for our MPST system is standard
(directly adapted from \cite{zooid}).

Atoms of our syntax are: a set of \emph{roles} (or \emph{participants}), ranged over by $\p,\q,\dots$, 
a set of \emph{(type) variables}, ranged over by $\hX,\hY,\dots$; and a set of labels, ranged over by $\lblFmt{\ell_0},\lblFmt{\ell_1},\dots,\lblFmt{\ell_i},\lblFmt{\ell_j},\dots$.
\begin{definition}[Sorts, Global Types, and Local Types]
  \label{def:MPST}%
        {\em Sorts},
        {\em global types},  and
        {\em local types}, ranged over by $\tS$, $\G$, and $\lT$
        respectively, are inductive datatypes generated by:
        \begin{small}
          \[
          \begin{array}{l}
            \tS::= \tunit \SEP \tnat \SEP \tint \SEP\tbool \SEP \tplus \tS \tS \SEP \tpair \tS \tS\qquad
\G   ::=  \gend \SEP \gX \SEP\grec \gX \G\SEP\msgi \p\q {\elllbl}{\tS} {\G} \SEP \dgt{\G_1 \, | \, \G_2} \\
\lT  ::= \lend \SEP \lX \SEP \lrec \lX \lT \SEP \lsend \q {\lblFmt{\elllbl}}{\tS} {\lT}\SEP \lrecv \p {\lblFmt{\elllbl}}{\tS} {\lT}
          \end{array}
          \]
          \end{small}
where, 
 $\p\neq \q$, $I\neq \emptyset$, and
$\lblFmt{\ell_i} \neq \lblFmt{\ell_j}$ when $i \neq j,$ for all $i,j\in I$, in $\msgi \p\q {\elllbl}{\tS} {\G}$, $\lsend \q {\lblFmt{\elllbl}}{\tS} {\lT}$, and $\lrecv \p {\lblFmt{\ell}}{\tS} {\lT}$. 
\end{definition}

The \emph{global message}
$\msgi \p\q \elllbl {\tS} {\G}$ describes a protocol where %
participant $\p$ sends to $\q$ one message %
with label $\lblFmt{\ell_i}$ and a value of sort $\dte{\tS_i}$ as payload, for
some $i \in I$; %
then, depending on which $\lblFmt{\ell_i}$ was sent by $\p$, %
the protocol continues as $\G_i$. 
The type $\gend$ represents a \emph{terminated protocol}. %
\emph{Recursive protocol} is modelled as $\dgt{\gmu\gX.\G}$, %
where recursion \emph{variable} $\gX$ is bound.
The \emph{parallel}
construct $ \dgt{\G_1 \, | \, \G_2}$ 
describes a protocol composed
by two independent ones.
The participants of $\dgt{\G_1}$ and $\dgt{\G_2}$
are required to be disjoint:
no communication happens between
$\dgt{\G_1}$ and $\dgt{\G_2}$,
but only internally in each one of them
(for a broader discussion, %
see \S\ref{sec:parallel}).
The intuition for local types $\lend$, $\lX$ and $\dlt{\mu\lX.\lT}$
is the same as for global types.
The \emph{send type}
$\lsend \q \elllbl {\tS} {\lT}$ %
says that the participant implementing the type %
must choose a labelled message to send to $\q$; %
if the participant chooses the label $\lblFmt{\ell_i}$, 
it must include in the message to $\q$ a payload value of sort
$\tS_i$, %
and continue as prescribed by $\dlt{\lT_i}$. %
The \emph{receive
  type} 
$\lrecv \p \elllbl {\tS} {\lT}$ %
requires to wait to receive %
a value of sort $\tS_i$ (for some $i \in I$) %
from the participant $\p$, via a message with label $\lblFmt{\ell_i}$; %
then the process continues as prescribed by $\dlt{\lT_i}$.

We are interested in types that are (1)
\emph{guarded}---e.g.,
$\gmu\gX.\allowbreak \msg \p\q \dgt{\elllbl (\tnat). \gX}$ is a valid
global type, whereas $\gmu\gX\dgt{.}\gX$ is not---(detail in Appendix \appref{subsec:appendix1}\inApp, Definition \appref{def:guarded-a})
and (2)
\emph{closed}, i.e., all variables are bound by $\gmu\gX$. In messages, sends and receives, the payload type can be omitted (e.g., $\msg \p \q \dgt{\elllbl.\dots}$), when only a label 
is exchanged. We assume that global and local types are always guarded,
and that, in $\dgt{\G_1 \, | \, \G_2}$,
$\dgt{\G_1}$ and $\dgt{\G_2}$ are closed.

\emph{Projection} plays a central role in MPST theories: it connects the protocol discipline, provided by the global type, with the local types that separately describe the behaviour of each participant.
%
\begin{definition}[Projection for Global Types]  \label{def:proj-old}%
  The \emph{projection of a global type onto a role $\pr$} is a partial function defined by recursion on $\G$, whenever the recursive call is defined:
  \begin{small}
    \[
    \begin{array}{l}
      \rulename{proj-end} \quad \projt{\gend} \pr = \lend \quad \quad \quad
      \rulename{proj-var} \quad \projt{\gX}{\pr}=\lX
      \\
      \rulename{proj-rec}\quad \projt{(\grec \gX \G)}{\pr}=\lrec \lX {} (\projt{\G}{\pr})\ \text{if}\ \lrec \lX {} (\projt{\G}{\pr})\ \text{is guarded, else, if}\ (\grec \gX \G)\ \text{is closed,}\ \projt{(\grec \gX \G)}{\pr}=\lend
      \\
      \rulename{proj-par}\quad\projt {\dgt{\G_1 \, | \, \G_2}} \pr = \projt {\dgt{\G_i}} \pr\ \text{if}\ \pr\ \text{is a participant of}\ \dgt{\G_i}\text{,}\
      \lend\ \text{otherwise}
      \\
      \rulename{proj-send}\quad \projt{\msgi \p\q \elllbl {\tS} {\G}} \pr=\lsnd \q \dlt{\{\lblFmt{\ell_i} {(\tS_i)}.  ({\color{black}\projt{\G_{\dgt i}} \pr})\}_{i\in I} } \ \text{if $\pr=\p$}
      \\
  \rulename{proj-recv}
  \quad\projt{\msgi \p\q \elllbl {\tS} {\G}} \pr=\lrcv \p \dlt{\{\lblFmt{\ell_i} {(\tS_i)}.  ({\color{black}\projt{\G_{\dgt i}} \pr})\}_{i\in I} }\ \text{if $\pr=\q$}
     \\
     \rulename{proj-merge}\quad
     \projt{\msgi \p\q \elllbl {\tS} {\G}} \pr = \merge_{i\in I} (\projt {\G_{\dgt i}} \pr)\ \text{if}\ \pr\neq \p\ \text{and}\ \pr \neq \q
\end{array}
\]
\end{small}
  \noindent $\projt \G \pr$ is undefined if none of the above applies. \emph{Merging} ($\mrg$) is defined as a partial operator over two local types such that: $\lT \mrg \lT = \lT$ for every type, it delves inductively inside all constructors (e.g., {\small $\lsend \q \elllbl \tS \lT \mrg \lsend \q \elllbl \tS {\lT'} = \lsnd \q \dlt{\{\lblFmt{\ell_i} {(\tS_i)}.  (\lT_i}\mrg\dlt{\lT'_i)\}_{i\in I}}$}), and
  {\small $\lrecv \p \elllbl \tS \lT \mrg \lrecvj \p \elllbl \tS {\lT'} = 
  \lrcv \p \dlt{\{\lblFmt{\ell_k} {(\tS_k)}.  \lT_k}\mrg\dlt{\lT'_k\}_{k\in I\cap J} \cup \{\lblFmt{\ell_k} {(\tS_k)}.  \lT_k\}_{k\in I\backslash J} \cup \{\lblFmt{\ell_k} {(\tS_k)}.  \lT'_k\}_{k\in J\backslash I}}$}.
\end{definition}
We describe the clauses of Definition \ref{def:proj-old}.
$\rulename{proj-end}$, $\rulename{proj-var}$, and $\rulename{proj-rec}$ are standard.
$\rulename{proj-send}$ (resp. $\rulename{proj-recv}$) %
  states that a global type starting with a message %
  from $\pr$ to $\q$ (resp. from $\p$ to $\pr$) %
  projects onto a sending (resp. receiving) local type
  $\lsnd \q \dlt{\{\lblFmt{\ell_i} {(\tS_i).} {\color{black}\projt{\G_{\dgt i}} \pr}\}}$
  (resp.~$\lrcv \p \dlt{
    \{\lblFmt{\ell_i} {(\tS_i).} {\color{black}\projt{\G_{\dgt i}}\pr}\} }$), provided that the continuations
  $\projt{\G_{\dgt{i}}}\pr$ are also projections of the corresponding
  global-type continuations $\dgt{G_i}$. $\rulename{proj-merge}$ states that, if the projected global type
  starts with an interaction between $\p$ and $\q$, and if we are
  projecting it onto a third participant $\pr$, then the projection is defined (and we can skip the message $\msg \p \q$) if all the continuations project onto \emph{mergeable} types (according to the merge operator $\mrg$ defined above).
  $\rulename{proj-par}$ states that
  projecting a parallel type $\dgt{\G_1 \, | \, \G_2}$ on $\pr$ is the same as projecting $\dgt{\G_1}$ or $\dgt{\G_2}$ onto $\pr$, depending on whether $\pr$ is a participant of one or the other type. 

By projecting a global type $\G$ onto all participants
($\projb \G\pr=\dlt{\lT_\pr}$, for the generic role $\pr$), 
we obtain a collection of local types $\{\dlt{\lT_\pr}\}$,
where $\dlt{\lT_\pr}$
is the behavioural type
for 
$\pr$.
Existing MPST theories (e.g., \cite{DenielouYoshida2013})
guarantee that a session of well-typed
implementations of the participants
inherits semantic guarantees for its communication, from $\G$.
Our development 
in \S\ref{sec:comp} is
compatible with such
theories: result of composition is a
well-formed $\G$, 
thus, implementations of 
projected
local types 
benefit from well-established semantics
results from the literature.

\subsection{Hybrid Types}
\label{sec:hybridtypes}\label{subsec:hybrid}
To allow the specification of interacting subprotocols,
we enrich the syntax of global types with local constructs.
We thus obtain ``hybrid'' types,
which use global messages for intra-protocol communication
and local sends and receives as openings
for inter-protocol communication.

\begin{definition}[Hybrid Types] \label{def:hybrid-types}
\emph{Hybrid types} are defined inductively by:\\[1mm]
\centerline{$
  \begin{array}{rcl}
    \hT  & ::= & \hend \SEP \hX \SEP\hrec \hX \hT \SEP \dht{\hT_1 | \hT_2}
    \SEP \hsend \p \q {\ell} \tS \hT
    \SEP \hrecv \p \q \ell \tS \hT \SEP
    \hmsgi \p\q \ell {S} {\hT}
  \end{array}
$}\\[1mm]
where, in $\hmsgi \p\q {\ell}{\tS} {\hT}$, $\hsend \p \q {\lblFmt{\ell}}{\tS} {\hT}$,
and 
$\hrecv \p \q {\lblFmt{\ell}}{\tS} {\hT}$, $\p\neq \q$, $I\neq \emptyset,$ and $\lblFmt{\ell_i} \neq \lblFmt{\ell_j}$ when $i \neq j,$ for all $i,j\in I$. 
We indicate the datatype of hybrid types with the notation $\HH$.
\end{definition}
The intuition behind each construct
is the same as in
Definition
\ref{def:MPST}, but
we write $\hsend \p \q \ell \tS \hT$
in place of $\lsend \q \ell \tS \lT$,
and $\hrecv \p \q \ell \tS \hT$
in place of $\lrecv \p \ell \tS \lT$.
For local types, $\lsend \q \ell \tS \lT$
(resp. $\lrecv \p \ell \tS \lT$) describes
the communication of the participant $\p$
sending a message to $\q$
(resp. $\q$ receiving a message from $\p$),
and then continuing with
interactions \emph{all involving $\p$ (resp. $\q$) as a subject}.
Therefore, such subject can be left implicit.
For hybrid types, instead,
different (internal) subjects
interact both with internal and external participants.
For instance:\\[1mm]
\centerline{
  $
    ^{(a)}\hsnd \p \q\dht{\lblFmt{\ell_1}.}\ \
    ^{(b)}\hmsg \p \pr \dht{\lblFmt{\ell_2}.}\ \
    ^{(c)}\hrcv \q \pr\dht{\lblFmt{\ell_3}.}\ \
    ^{(d)}\hend
  $}
    \begin{itemize*}
    \item[$(a)$] $\p$ sends an external message to $\q$;
    \item[$(b)$] $\p$ exchanges a message internally with $\pr$;
    \item[$(c)$] $\pr$ receives an external message from $\q$; and
    \item[$(d)$] the protocol terminates.
    \end{itemize*}


\begin{definition}[Internal and External Participants] \label{def:participants}
  We define the sets of \emph{internal participants} and \emph{external participants} of a hybrid type $\hT$ by recursion:
  \begin{small}
  \[
  \hspace{-3mm}
  \begin{array}{ll}
  \begin{array}{l}
    \hpart \hend = \emptyset \quad \hpart \hX = \emptyset \\
    \hpart {\hrec \hX \hT} = \hpart \hT\\
      \hpart {\hpar \hTone \hTtwo} = \hpart {\hTone} \cup \hpart \hTtwo \\
      \hpart {\hsend \p \q \ell \tS \hT}= \{\p\} \cup \bigcup_{i\in I} \hpart {\hTi}\\
      \hpart {\hrecv \p \q \ell \tS \hT} = \{\q\} \cup \bigcup_{i\in I} \hpart {\hTi}\\
      \hpart {\hmsgi \p \q \ell \tS \hT} = \{\p,\q\} \cup \bigcup_{i\in I} \hpart {\hTi}
  \end{array}&
  \hspace*{-4mm}\begin{array}{l}
    \hepart \hend = \emptyset \quad \hepart \hX = \emptyset\\
    \hepart {\hrec \hX \hT} = \hepart \hT\\
      \hepart {\hpar \hTone \hTtwo} = \hepart {\hTone} \cup \hepart \hTtwo \\
       \hepart {\hsend \p \q \ell \tS \hT} = \{\q\} \cup \bigcup_{i\in I} \hepart {\hTi}\\
       \hepart {\hrecv \p \q \ell \tS \hT}= \{\p\} \cup \bigcup_{i\in I} \hepart {\hTi}\\
       \hepart {\hmsgi \p \q \ell \tS \hT} = \bigcup_{i\in I} \hepart {\hTi}
  \end{array}
  \end{array}
  \]
  \end{small}
\end{definition}

We define, for hybrid types, 
\emph{guardedness} and \emph{closedness},
as for global and local types (Definition \ref{def:MPST}).
We require that all 
hybrid types in this paper
are guarded
(detail in Appendix \appref{subsec:appendix1}\inApp,
Definition \appref{def:guarded-a}) and that, for all
$\hT$, $\hpart\hT\cap\hepart\hT=\emptyset$.
Also, we require well-formedness for parallel constructs: for
$\hpar \hTone \hTtwo$,
$\hTone$ and $\hTtwo$ must be closed and
$(\hpart\hTone \cup \hepart \hTone)\allowbreak\cap(\hpart \hTtwo \cup \hepart\hTtwo)=\emptyset$.
Namely, the parallel construct
describes
communication that happens independently,
within two separate groups of participants.
We 
express 
\emph{global and local types in terms of hybrid types}, 
with 
two predicates on $\HH$:
$\isglobal\hT$ holds iff $\hT$ is formed only by global constructs
(global type syntax in Definition~\ref{def:MPST}); and
$\islocal\hT$ holds iff $\hpart\hT$
contains at most one element (hence $\hT$ contains only local constructs,
see local types in Definition~\ref{def:MPST}).

\begin{example}\label{ex:pprot} 
  We 
  use as a recurring example
  the company from \S\ref{sec:overview}.
  A designer, $D_{\sf str}$, describes the protocol $\hTstr$
  for the strategic team, as in Figure \ref{fig:co-hybrid2}:\\
  \begin{small}
    \centerline{$
      \hTstr:=\hmsg\dir\ad \dht{\product(\tnat).\hsnd \dir \sales \product(\tnat).\hsnd \dir \fone \product(\tnat).}
    \hrec \hX {\hrcv \fone \dir}
    \dht{\left\{
      \begin{array}{l}
        \ok.\hmsg\dir\ad\go.\hend,\\
        \wait.\hmsg\dir\ad\wait.\hX
      \end{array}
      \right\}}
      $}
  \end{small}
First, $\dir$ sends internally a product ID to $\ad$, then a similar external message to $\sales$, of the sales department, and to $\fone$, of the finance department. $\dir$ waits in a recursive loop for $\fone$ to give the $\ok$. When this happens, $\dir$ internally communicates to $\ad$ that they can proceed with the product advertisement.
For $\hTstr$, 
the sets of internal and external participants are
$\hpart\hTstr = \{\dir,\ad\}$ and $\hepart\hTstr=\{\sales,\fone\}$.
We observe that $D_{\sf str}$ is not concerned with
the communication that happens internally
to the sales department or the financial one,
nor with the communication between these two.
Designers $D_{\sf sales}$ and $D_{\sf fin}$
independently give protocols
for the sales and financial departments
respectively (as in Figure \ref{fig:co-hybrid2}):\\[1mm]
\begin{small}
\centerline{$
\begin{array}{l}
  \hTsales=\hrcv \dir \sales \dht{\product(\tnat).}
  \hrec\hX{\hrcv\fone\sales}
\dht{\left\{
  \begin{array}{l}
    \price(\tnat).\hmsg\sales\web\publish(\tnat).\hend,\\
    \wait.\hmsg\sales\web\wait.\hX
  \end{array}
  \right\}}\\
  \hTfin=\hrcv \dir \fone \dht{\product(\tnat).\hmsg \fone \ftwo \product(\tnat).}
  \hrec\hX{\hmsg\ftwo\fone}
  \dht{\left\{
    \begin{array}{l}
      \price(\tnat).\hsnd \fone \dir\ok.\hsnd \fone\sales \price(\tnat).\hend,\\
      \wait.\hsnd \fone \dir\wait.\hsnd \fone\sales \wait.\hX
    \end{array}
    \right\}}
\end{array}
$}\\[1mm]
\end{small}
In the sales department, once $\dir$ has communicated the product, $\sales$ waits in a loop for the decision about the price from the financial department, then gives to the website administrator $\web$ the command to publish.
We have that $\hpart\hTsales = \{\sales,\web\}$ and $\hepart\hTsales=\{\dir,\fone\}$. The decision about the price of the product is taken by $\ftwo$, and communicated internally to the financial department with 
$\hmsg \ftwo \fone \dht{\dots}$; then $\fone$ communicates the decision to the other departments, which can continue with their internal communication. We have that $\hpart\hTfin = \{\fone,\ftwo\}$ and $\hepart\hTfin=\{\dir,\sales\}$.




In \S\ref{sec:comp},
we prove
the above types
\emph{compatible}
and \emph{compose} them into a single
global type.
\end{example}


\subsection{Projection and Localiser}

We introduce \emph{projection} and \emph{localiser} for hybrid types. These operators play are fundamental for defining compatibility and, ultimately, achieving compositionality.

\begin{definition}[Projection] \label{def:proj}
  The (generalised) projection of a hybrid type on the set of participants $E$, is a partial operator, $\proj \_ {E}:\HH\rightarrow\HH$, recursively defined by the following clauses (whenever the recursive call is defined):
  \begin{small}
  \[
  \begin{array}{ll}
    \begin{array}{l}\rulename{proj-end}  \quad \proj \hend {E} = \hend \end{array} &
    \begin{array}{l}\rulename{proj-var} \quad \proj \hX {E} = \hX \end{array}\\
  \begin{array}{l}
  \rulename{proj-rec} \\
  \projb {\hrec \hX \hT} {E} =  \hrec \hX {\dht{(\projb \hT {E})}}\ \text{if}\ \hrec \hX {\dht{(\projb \hT {E})}}\
  \text{is guarded,}\\
  \text{else, if}\ (\hrec \hX \hT)\ \text{is closed,}\ \projb{(\hrec \hX \hT)}{E}=\hend\\
  \end{array} &
  \begin{array}{l}
  \rulename{proj-par}\\
  \projb  {\hpar {\hT_1} {\hT_2}} {E} = \projb {\dht{\hT_i}} {E}
  \\
  \text{if $E \cap \hpart {\dht{\hT_{\mathit j}}} = \emptyset$, with $i,j\in\{1,2\}$ and $i\neq j$}
  \end{array}\\
  \begin{array}{l}
  \rulename{proj-send}\\
    \projb{\hsend \p \q \ell \tS \hT} {E} = \\
    \begin{cases}
      \hsnd \p \q  \dht{\{\lblFmt{\ell_i} (\tS). (\projb {\hT_i} {E})\}_{i\in I}}\ \text{if $\p\in E$ and $\q\notin E$}\\
      \merge_{i\in I} \dht{(\projb {\hT_i} {E})} \ \text{if $\p,\q\notin E$}
    \end{cases}
    \\[3mm]
  \rulename{proj-recv}\\
        \projb{\hrecv \p \q \ell \tS \hT} {E}  = \\
        \begin{cases}
          \hrcv \p \q \dht{\{\lblFmt{\ell_i} (\tS). (\projb {\hT_i} {E})\}_{i\in I}} \ \text{if $\p\notin E$ and $\q\in E$}\\
          \merge_{i\in I} \dht{(\projb {\hT_i} {E})} \ \text{if $\p,\q\notin E$}
        \end{cases}
  \end{array}
  &
  \begin{array}{l}
          \rulename{proj-msg}\\
  \projb {\hmsgi \p\q \ell {\tS} {\hT}} {E} =
  \\ \begin{cases}
      \hmsg \p \q \dht{\{\lblFmt{\ell_i} (\tS). (\projb {\hT_i} {E})\}_{i\in I}} \ \text{if $\p,\q\in E$} \\
          \hsnd \p \q  \dht{\{\lblFmt{\ell_i} (\tS). (\projb {\hT_i} {E})\}_{i\in I}} \ \text{if $\p\in E$ and $\q\notin E$}\\
          \hrcv \p \q  \dht{\{\lblFmt{\ell_i} (\tS). (\projb {\hT_i} {E})\}_{i\in I}} \ \text{if $\p\notin E$ and $\q\in E$}\\
          \dht{\merge_{i\in I} (\projb {\hT_i} {E})} \ \text{if $\p,\q\notin E$}
  \end{cases}\\[10mm]
  \text{and undefined otherwise.}
  \end{array}
    \end{array}
  \]
  \end{small}
 \emph{Merging} ($\mrg$) is defined as a partial commutative operator
 over two hybrid types such that: for all $\hT$, $\hT \mrg \hT = \hT$, it delves inductively inside all constructors (e.g., {\small $\hmsgi \p \q \elllbl \tS \hT \mrg\allowbreak \hmsgi \p \q \elllbl \tS {\hTb} \allowbreak = \hmsg \p \q \dht{\{\lblFmt{\ell_i} {(\tS_i)}.  (\hTi}\mrg\dht{\hTb_i)\}_{i\in I}}$}), and
{\small $
 \hrecv \p \q \ell \tS \hT \mrg \hrecvj \p \q \ell \tS {\hT'} = \allowbreak
  \hrcv \p \q
  \allowbreak \dht{
    \{\lblFmt{\ell_k} {(\tS_k)}.  \hT_k}\mrg\dht{\hT'_k\}_{k\in I\cap J}}
  \dht{\cup \{\lblFmt{\ell_k} {(\tS_k)}.  \hT_k\}_{k\in I\backslash J}
    \cup \{\lblFmt{\ell_k} {(\tS_k)}.  \hT'_k\}_{k\in J\backslash I}
  }
  $}
\end{definition}
With respect to Definition \ref{def:proj-old}, we now allow projection onto a \emph{set} of participants, and we introduce rules for projecting send and receive constructs. We highlight the differences below:
\begin{itemize}
\item $\rulename{proj-msg}$ defines $\hmsgi \p\q \ell {\tS} {\hT}$ to be projected onto $E$, if \emph{both $\p\in E$ and $\q\in E$}; in this case the structure of the global message $\hmsg \p \q$ is maintained in the projected type;
\item $\rulename{proj-send}$ defines projection when the sender $\p$ is in $E$ and $\q$ is not, and when both $\p,\q\notin E$;
\item $\rulename{proj-recv}$ defines projection when the receiver $\q$ is in $E$ and $\p$ is not, and when both $\p,\q\notin E$.
\end{itemize}

\begin{remark}\label{rem:proj}
  Projection is defined only onto
  sets of \emph{internal} participants; e.g.,
  $\hrecv \p \q \ell \tS \hT$ can be projected
  onto $\q$, but not onto $\p$;
  also, $\hpart {\projb \hT E} \subseteq \hpart \hT$.
  If we project a hybrid type onto a singleton,
  we obtain  
  a local type:
  $\islocal {\projb \hT {\{\p\}}}$.
  Furthermore, if $\isglobal\hT$, then
  $\projb \hT {\{\p\}}$ is exactly the traditional
  MPST projection of the
  global type $\hT$ onto $\p$,
  $\projb\hT\p$ (Definition \ref{def:proj-old}).
 \end{remark}

\begin{definition}[Localiser] \label{def:loc}
  The localiser of a hybrid type 
  is a partial operator, $\loc {} {\;\_}:\HH\rightarrow\HH$, recursively defined by the following clauses (whenever the recursive call is defined):
  \begin{small}
  \[
  \begin{array}{ll}
    \rulename{loc-end}&\loc{}\hend=\hend
    \quad\quad\quad\rulename{loc-var} \quad\quad  \loc{}\hX=\hX\quad \quad \quad
 \rulename{loc-par} \quad\quad\loc{} (\hpar {\hT_1} {\hT_2}) = \hpar {\loc {}\dht{\hT_1}} {\loc{} \dht{\hT_2}}
    \\
    \rulename{loc-rec}&\loc{}\hrec \hX \hT=\hrec \hX {(\loc{}\hT)}\ \text{if}\ \hrec \hX {(\dht{\loc {} \hT})}\
  \text{is guarded, 
      else, if}\ \hrec \hX \hT\ \text{is closed,}\ \loc {} {\hrec \hX \hT}=\hend\\


        \rulename{loc-send} &\loc{}(\hsend \p \q \ell \tS \hT)=
          \hsend \p \q \ell \tS {\loc{}\hT}\\
        \rulename{loc-recv} &\loc{}(\hrecv \p \q \ell S \hT)=
          \hrecv \p \q \ell S {\loc{}\hT}\\
        \rulename{loc-msg} &\loc{}(\hmsgi \p\q \ell {S} {\hT}) =
           \lmerge_{i\in I} {\loc{}(\dht{\hT_i})} \\
           \multicolumn{2}{l}{\text{and undefined otherwise.}}
      \end{array}
  \]
  \end{small}
    \emph{Merging for the localiser} ($\lmrg$) is a partial
    commutative operator over two hybrid types such that: for all $\hT$,
    $\hT \lmrg \hT = \hT$,
    it delves inductively inside all constructors (e.g., {\small $\hsend \p \q \elllbl \tS \hT \allowbreak\lmrg \hsend \p \q \elllbl \tS {\hTb} = \hsnd \p \q \dht{\{\lblFmt{\ell_i} {(\tS_i)}.  (\hT_i}\lmrg\dht{\hTb_i)\}_{i\in I}}$}), and
    {\small $
  \hsend \p \q \ell \tS \hT \lmrg \hsendj \p \q \ell \tS {\hT'} =\\
  \quad\quad\quad\quad
  \hsnd \p \q
  \dht{\{\lblFmt{\ell_k} {(\tS_k)}.  \hT_k}\lmrg\dht{\hT'_k\}_{k\in I\cap J}}
  \dht{\cup \{\lblFmt{\ell_k} {(\tS_k)}.  \hT_k\}_{k\in I\backslash J}
    \cup \{\lblFmt{\ell_k} {(\tS_k)}.  \hT'_k\}_{k\in J\backslash I}
  }
  $}
\end{definition}
The localiser is a forgetful operator that \emph{preserves local constructs} and discards global messages. 
$\rulename{loc-end}$, $\rulename{loc-var}$, $\rulename{loc-rec}$, and $\rulename{loc-par}$ preserve the non-message structure of the type, into its localisation. 
$\rulename{loc-send}$ and $\rulename{loc-recv}$ state that send construct (an internal participant $\p$ sends to an external participant $\q$) and receive construct (an internal participant $\q$ receives from an external participant $\p$) are to be maintained and their continuations $\hTi$ localised into $\loc {} \hTi$. $\rulename{loc-msg}$ is the central rule: each global message has to be skipped and its continuations need to be merged.

\begin{remark}
  The merge operator for the localiser, $\lmrg$,
  is \emph{dual to the merge for projection}, $\mrg$.
  To build the intuition behind this,
  let us consider the following hybrid type:
  $\hT = \hmsg\p\q \dht{\{\lblFmt{\ell_1}.
    \hsnd\p\pr\lblFmt{\ell_3}.\hend,}
  \allowbreak\dht{\lblFmt{\ell_2}.
    \hsnd\p\pr\lblFmt{\ell_4}.\hend\}}$.
  First, $\p$ chooses on which branch to take,
  by \emph{internally} sending either
  $\lblFmt{\ell_1}$ or $\lblFmt{\ell_2}$ to $\q$;
  then according to the chosen branch, $\p$
  itself sends a different \emph{external} message
  to $\pr$. When we localise the above
  type we obtain
  $\loc {} \hT =
  \hsnd\p\pr\dht{\{\lblFmt{\ell_3}.
    \hend,\lblFmt{\ell_4}.\hend\}}$,
  namely we have merged send
  constructs with different labels,
  $\nhsnd\p\pr\lblFmt{\ell_3}$ and
  $\nhsnd\p\pr\lblFmt{\ell_4}$.
  From the point of view of the external receiver $\pr$,
  it makes no difference whether
  such choice has been taken by $\p$
  at the time $\p$ sends to $\pr$ (with $\hsnd\p\pr$),
  or at a precedent stage of communication,
  internal to $\hT$ (with $\nhmsg\p\q$).
  This intuition is proven correct
  by the results from the next section,
  when we define
  a compatibility notion,
  based on localiser and projection,
  and we prove compositionality. 
\end{remark}

\section{Compositionality for Distributed Specification}
\label{sec:comp}




In \S\ref{sec:hybrid}, we have set definitions in place 
to compose subprotocols.
In particular,
following the overview of Figure \ref{fig:overview}, \S\ref{sec:overview}, what we need is:
\begin{itemize}
\item 
hybrid types $\hTone,\hTtwo,\dots,\hTN$ for the multiple components of the communicating system;
\item  a compatibility hybrid type $\ghT$
(we sometimes use the notation $\gG$, when $\ghT$ is a global type, namely when $\isglobal \ghT$)
that disciplines the 
inter-component communication; 
and
\item the property that $\ghT$ projects onto the localisations of $\hTone,\hTtwo,\dots,\hTN$ (compatibility).
\end{itemize}

In this section, we present our journey to
multiparty compositional specification in three steps:
\begin{enumerate}
\item we focus on a single hybrid type $\hTi$, for which
compatibility holds: $\projb \ghT {E_i} = \loc {} \hTi$;
we build a new type $\mathcal{B}^1_{E_i}\;{(\ghT)}\;{(\hTi)}$,
whose projection on $E_i$ coincides with $\hTi$,
and which contains the information
for external communication from $\ghT$ (Theorem \ref{thm:comp1});
\item we show how Step 1 is a base case for
composing multiple compatible protocols:
from $\bbo$, we recursively define
$\mathcal{B}\;{(\ghT)}\;{([\hTone,\dots,\hTN])}$,
which projects onto $\hTi$ \emph{for all} $i=1,\dots,N$
(compositionality, Theorem \ref{thm:comp});
\item we prove that projection composes over
the subset relation (Theorem \ref{thm:proj-comp});
this guarantees the applicability and correctness of our result:
if $\ghT$ is a \emph{global type},
we obtain a well-formed global type $\G=\bba {(\ghT)} {([\hTone,\dots,\hTN])}$
for the whole system,
the projections of which, onto every participant,
are the same as the projections of the subprotocols $\hTi$ (Corollary \ref{cor:dist-spec}).
\end{enumerate}

\subsection{Step 1: Building Back a Single Subprotocol}
\label{subsec:comp1}
Our first step towards compositionality
is also the most technical of the three. In this section we
present the main design choices, both in constructions
and in proofs, that make our theory sound.
For more details, we refer to Appendix \appref{subsec:appendix1}\inApp.

We are given $\ghT$, the compatibility type disciplining communication
happening among subprotocols, and with one of these subprotocols
$\hTE$, describing the communication
from the point of view of its internal participants,
contained in the set $E$.
The local constructs of $\hTE$ are \emph{compatible} with
what prescribed by $\ghT$ for communicating externally, formally:
\begin{equation}
\projb {\ghT} {E} = \loc{}\hTE
\tag{\textbf{C}}\label{eq:compatibility}
\end{equation}
The above notion is designed for the direct
composition of multiple subprotocols:
the hybrid type $\hTE$ for one component is
checked compatible,
not against other components, but
against $\ghT$, which gives global guidance
for inter-component communication.
This design choice differentiates our theory from previous
work, where compatibility is checked by
directly matching the inputs and outputs of two
separate components
(see \S\ref{sec:related} for further discussion).
With such compatible types, we build $\hT=\bbeoa {(\ghT)} {(\hTE)}$
that retains the information about external communication of $\ghT$
and about internal communication in the component $\hTE$.
\begin{figure}
\begin{small}
\[
    \begin{array}{l}
    \bbeoa {(\hend)} {(\hTE)} = \hTE \qquad\qquad\qquad\qquad\qquad \bbeoa {(\hX)} {(\hTE)} = \hTE\\
    \bbeoa {(\hrec \hX \ghTb)} {(\hTE)} =
    \begin{cases}
      \hpar {(\hrec \hX \ghTb)} {\hTE}\ \text{if $\hpart \ghTb \cap E = \emptyset$ and both $\hrec \hX \ghTb$ and $\hTE$ are closed}\\
      \hmsg \ps \pr \dht{\{\lblFmt{\ell_i}(\tS).(\bbeoa {(\hrec \hX \ghTb)} {(\hTEi)})\}_{i\in I}}
      \ \text{if $\hpart \ghTb \cap E = \emptyset$,}
      \\ \quad\text{$\hTE=\hmsg\ps\pr
      \dht{\{\lblFmt{\ell_i}(\dte{\tS_i})\hTEi\}_{i\in I}}$, }
      \text{and one of $\hrec \hX \ghTb$ and $\hTE$ is not closed}\\
      \hrec \hX {(\bbeoa {(\ghTb)} {(\hTEb)})}\ \text{if $\hpart \ghTb \cap E \neq \emptyset$ and $\hTE = \hrec \hX \hTEb$}
    \end{cases}\\
    \bbeoa {(\hpar \ghTone \ghTtwo)} {(\hTE)} =
    \begin{cases}
      \hpar {(\bbeoa {(\ghTone)} {(\hTE)})} \ghTtwo\ \text{if $\hpart \ghTtwo \cap E = \emptyset$}\\
      \hpar \ghTone {(\bbeoa {(\ghTtwo)} {(\hTE)})}\ \text{if  $\hpart \ghTtwo \cap E \neq \emptyset$ and $\hpart \ghTone \cap E = \emptyset$}\\
      \end{cases}\\
\hspace*{-3.5mm}\begin{array}{ll}
\begin{array}{l}
         \bbeoa {(\hsnd \p \q \dht{\{\lblFmt{\ell_i} (\dte{\tS_i}). \ghTi\}_{i\in I}})} {(\hTE)} =
    \\ \quad
    \begin{cases}
      \hsnd \p \q \dht{\{\lblFmt{\ell_i} (\dte{\tS_i}).(\bbeoa {(\ghTi)} {(\hTEi)})\}_{i\in I}}\
      \\ \quad \text{if $\hTE = \hsend \p \q \ell \tS \hTE$}\\
      \hmsg \ps \pr \dht{\{\lblFmt{m_j} (\dte{\tS'_j}).(\bbeoa {({\sf UL}_j)} {(\hTEj)})\}_{j\in J}}\
      \\ \quad \text{with ${\sf UL}_j = \pi_j(\ula \ghT {[\loc{} {\hTEj}]_{j\in J}})$,}
      \\ \quad \text{if $\hTE =
      \hmsg \ps \pr
      \dht{\{\lblFmt{m_j} (\dte{\tS'_j}) \hTEj\}_{j\in J}}$ and $\p\in E$}\\
      \hsnd \p \q \dht{\{\lblFmt{\ell_i} (\dte{\tS_i}).(\bbeoa {(\ghTi)} {({\sf UP}_i)})\}_{i\in I}}
      \\ \quad \text{with ${\sf UP}_i = \pi_i(\upa \hTE {[\projb {\ghTi} E]_{i\in I}})$, if $\p\notin E$}
    \end{cases}
    \end{array}
    &\hspace*{-6.5mm} 
\begin{array}{l}
     \bbeoa {(\hrcv \p \q \dht{\{\lblFmt{\ell_i} (\dte{\tS_i}). \ghTi\}_{i\in I}})} {(\hTE)} =
    \\ \quad
    \begin{cases}
      \hrcv \p \q \dht{\{\lblFmt{\ell_i} (\dte{\tS_i}).(\bbeoa {(\ghTi)} {(\hTEi)})\}_{i\in I}}\
      \\ \quad \text{if $\hTE = \hrecv \p \q \ell \tS \hTE$}\\
      \hmsg \ps \pr \dht{\{\lblFmt{m_j} (\dte{\tS'_j}).(\bbeoa {({\sf UL}_j)} {(\hTEj)})\}_{j\in J}}\
      \\ \quad \text{with ${\sf UL}_j = \pi_j(\ula \ghT {[\loc{} {\hTEj}]_{j\in J}})$,}
      \\ \quad \text{if $\hTE =
      \hmsg \ps \pr
      \dht{\{\lblFmt{m_j} (\dte{\tS'_j}) \hTEj\}_{j\in J}}$ and $\q\in E$}\\
      \hrcv \p \q \dht{\{\lblFmt{\ell_i} (\dte{\tS_i}).(\bbeoa {(\ghTi)} {({\sf UP}_i)})\}_{i\in I}}\
      \\ \quad \text{with ${\sf UP}_i = \pi_i(\upa \hTE {[\projb {\ghTi} E]_{i\in I}})$, if $\q\notin E$}
    \end{cases}
    \end{array}
    \end{array}\\
    \bbeoa {(\hmsg \p \q \dht{\{\lblFmt{\ell_i} (\dte{\tS_i}). \ghTi\}_{i\in I}})} {(\hTE)} =
    \begin{cases}
      \hmsg \p \q \dht{\{\lblFmt{\ell_i} (\dte{\tS_i}).(\bbeoa {(\ghTi)} {(\hTEi)})\}_{i\in I}}
      \\ \quad \text{if $\hTE = \hsnd \p \q \dht{\{\lblFmt{\ell_i} (\dte{\tS_i}). \hTEi\}_{i\in I}}$ or $\hTE = \hrcv \p \q \dht{\{\lblFmt{\ell_i} (\dte{\tS_i}). \hTEi\}_{i\in I}}$}\\
      \hmsg \ps \pr \dht{\{\lblFmt{m_j} (\dte{\tS'_j}).(\bbeoa {({\sf UL}_j)} {(\hTEj)})\}_{j\in J}}\
      \\ \quad \text{with ${\sf UL}_j = \pi_j(\ula \ghT {[\loc{} {\hTEj}]_{j\in J}})$,}
      \\ \quad \text{if $\hTE = \hmsg \ps \pr
      \dht{\{\lblFmt{m_j} (\dte{\tS'_j}). \hTEj\}_{j\in J}}$ and $\p\in E$ or $\q\in E$}\\
      \hmsg \p \q \dht{\{\lblFmt{\ell_i} (\dte{\tS_i}).(\bbeoa {(\ghTi)} {({\sf UP}_i)})\}_{i\in I}}\
      \\ \quad \text{with ${\sf UP}_i = \pi_i(\upa \hTE {[\projb {\ghTi} E]_{i\in I}})$, if $\{\p,\q\}\cap E = \emptyset$}
    \end{cases}
    \end{array}
\]
\end{small}
\caption{Build-Back of a Single Component: Equations}\label{fig:bb1-equations}
\end{figure}

\begin{definition}[Build-Back of a Single Component]\label{def:bb1}
Given a set of participants $E$, we define the
\emph{build-back of a single component} as the partial recursive
function $\bbeoa {(\ghT)} {(\hTE)}$. The recursive equations are
given in Figure \ref{fig:bb1-equations}; if none of those apply or if
$\hpart \hTE\nsubseteq E$, $\bbeoa {(\ghT)} {(\hTE)}$ is undefined.
\end{definition}

The rest of this subsection is dedicated to discussing the intuition behind the function $\bbeo$ and its correctness (Theorem \ref{thm:comp1}). First, let us consider the following equations from Definition \ref{def:bb1}: 
\[
\begin{small}
\begin{array}{l}
         \bbeoa {(\hmsgi \p \q \ell \tS \ghT)} {(\hsend \p \q \ell \tS \hTE)} = \hmsg \p \q \dht{\{\lblFmt{\ell_i} (\dte{\tS_i}).(\bbeoa {(\ghTi)} {(\hTEi)}\}_{i\in I}}\\
         \bbeoa {(\hmsgi \p \q \ell \tS \ghT)} {(\hrecv \p \q \ell \tS \hTE)} = \hmsg \p \q \dht{\{\lblFmt{\ell_i} (\dte{\tS_i}).(\bbeoa {(\ghTi)} {(\hTEi)}\}_{i\in I}}\\

\end{array}
\end{small}
\]
The two equations above show how our compatibility (Equation \ref{eq:compatibility})
comes into play when building back a more general type. E.g.,
when the projection of
$\ghT=\hmsgi \p \q \ell \tS \ghT$ onto $E$ is equal to the localisation
of a send type $\hTE=\hsend \p \q \ell \tS \hTE$ (with $\p\in E$) then
their build-back is
$\hmsg \p \q \dht{\{\lblFmt{\ell_i} (\dte{\tS_i}).(\bbeoa {(\ghTi)} {(\hTEi)}\}_{i\in I}}$.
The case where $\hTE=\hrecv \p \q \ell \tS \hTE$ is analogous.


The next example shows
how the build-back
retains the information both
\begin{enumerate*}
\item internal to the component of $\hTE$ and
\item about the inter-component communication of the system, given by $\ghT$.
\end{enumerate*}

\begin{example}[Build-Back, Intuition]
We are given 
$\hTE=\dht{\hmsg\ps\pr\lblFmt{\ell_1}.\hsnd\ps\p\lblFmt{\ell_2}.\hend}$,
with internal participants $E=\{\ps,\pr\}$,
$\ghT=\dht{\hmsg\p\q\lblFmt{\ell_0}.\hmsg\ps\p\lblFmt{\ell_2}.\hend}$,
for compatibility, describing inter-component communication in the system:
compatibility \ref{eq:compatibility} holds.
Following the equations in Figure \ref{fig:bb1-equations},
we first \emph{build back} the prefix $\dht{\hmsg\p\q\lblFmt{\ell_0}}$,
then recursively the prefix $\dht{\hmsg\ps\pr\lblFmt{\ell_1}}$, and ultimately
we exploit compatibility and $\dht{\hsnd\ps\p\lblFmt{\ell_2}}$ gets absorbed into
$\dht{\hmsg\ps\p\lblFmt{\ell_2}}$. Namely, 
$\bbeoa {(\ghT)} {(\hTE)}=\dht{\hmsg\p\q\lblFmt{\ell_0}.\hmsg\ps\pr\lblFmt{\ell_1}.\hmsg\ps\p\lblFmt{\ell_2}.}\allowbreak\hend$.
We observe that 
$\bbeoa {(\ghT)} {(\hTE)}$ contains \emph{all the interactions},
both intra-component (in $\hTE$) and
inter-component (in $\ghT$); in other words,
$\bbeoa {(\ghT)} {(\hTE)}$ carries the information
both in $\hTE$ and in $\ghT$.
This property of the build-back is formalised by conclusions $(1)$ and $(2)$
of Theorem \ref{thm:comp1}.
\end{example}
Some detail 
from the previous example is 
hidden in
the auxiliary ``unmerge'' functions $\up$ and $\ul$,
for projection and localiser respectively.
They reproduce in $\bbeoa {(\ghT)} {(\hTE)}$
a branching structure that is faithful to the branching both in
$\ghT$ and in $\hTE$, where such branching may have been
merged when projecting $\ghT$ on $E=\{\ps,\pr\}$
(with $\mrg$, see Definition \ref{def:proj})
or when localising $\hTE$
(with $\lmrg$, see Definition \ref{def:loc}).
We present the unmerge mechanism
with the next example, while, for
formal details,
we refer the interested reader
to Appendix \appref{subsec:appendix1}\inApp{}.
\begin{example}[Unmerge]
We focus on the merge for the localiser $\lmrg$; the case
for projection is analogous.
We are given $\hTE=\dht{\hmsg\ps\pr}
\dht{\{\lblFmt{\ell_{11}}.\hsnd\ps\p\lblFmt{\ell_{21}}.\hend,}
\dht{\lblFmt{\ell_{12}}.\hsnd\ps\p\lblFmt{\ell_{22}}.\hend\}}$
and
$\ghT=\dht{\hmsg\p\q\lblFmt{\ell_0}.\hmsg\ps\p}\allowbreak
\dht{\{\lblFmt{\ell_{21}}.\hend,}\allowbreak\dht{\lblFmt{\ell_{22}}.\hend\}}$
In this case,
$\projb \ghT E = \loc{} \hTE = \dht{\hsnd\ps\p\{\lblFmt{\ell_{21}}.\hend,\lblFmt{\ell_{22}}.\hend\}}$.
In particular, when computing $\loc{} \hTE$,
a merge of branches happens:
$\loc{} \hTE=\dht{\hsnd\ps\p\{\lblFmt{\ell_{21}}.\hend\}}\lmrg\dht{\hsnd\ps\p\{\lblFmt{\ell_{22}}.\hend\}}$.
When building back, we need to unmerge and reproduce the original
branching from $\hTE$. In particular
\begin{small}
\[
\projb {\dht{\hmsg\ps\p\{\lblFmt{\ell_{21}}.\hend,\lblFmt{\ell_{22}}.\hend\}}} E=
(\dht{\hsnd\ps\p\{\lblFmt{\ell_{21}}.\hend\}})\lmrg(\dht{\hsnd\ps\p\{\lblFmt{\ell_{22}}.\hend\}})
\]
\end{small}
Under this hypothesis, $\ul$
returns
suitable branches for the build-back:
\begin{small}
\[
\ul\;(\dht{\hmsg\ps\p\{\lblFmt{\ell_{21}}.\hend,\lblFmt{\ell_{22}}.\hend\}})\;[\dht{\hsnd\ps\p\{\lblFmt{\ell_{21}}.\hend\}},\dht{\hsnd\ps\p\{\lblFmt{\ell_{22}}.\hend\}}]
=[\ \dht{\hmsg\ps\p\{\lblFmt{\ell_{21}}.\hend\}}\ ,\ \dht{\hmsg\ps\p\{\lblFmt{\ell_{22}}.\hend\}}\ ]
\]
\end{small}
Then, by following the build-back algorithm we obtain:
\[
\begin{small}
\bbeoa {(\ghT)} {(\hTE)}=\dht{\hmsg\p\q\lblFmt{\ell_0}.\hmsg\ps\pr}
\dht{\left\{
\dht{\lblFmt{\ell_{11}}.\hmsg\ps\p\lblFmt{\ell_{21}}.\hend,}
\dht{\lblFmt{\ell_{12}}.\hmsg\ps\p\lblFmt{\ell_{22}}.\hend}
\right\}}
\end{small}
\]
We observe that, above,
the output of $\ul$, with arguments $\ghTb=\dht{\hmsg\ps\p\{\lblFmt{\ell_{21}}.\hend,\lblFmt{\ell_{22}}.\hend\}}$
$[\hTLone,\hTLtwo]=[\dht{\hsnd\ps\p\{\lblFmt{\ell_{21}}.\hend\}},\dht{\hsnd\ps\p\{\lblFmt{\ell_{21}}.\hend\}}]$,
is a list of two types
$[\ghTone,\ghTtwo]=\allowbreak[\dht{\hmsg\ps\p\{\lblFmt{\ell_{21}}.\hend\}},\allowbreak\dht{\hmsg\ps\p\{\lblFmt{\ell_{21}}.\hend\}}]$;
for these, in particular, the following properties hold:
\begin{itemize*}
\item[$(a)$] $\projb \ghTi E = \hTLi$, for $i=1,2$
\item[$(b)$] $\projb \ghTb {E'}= \projb \ghTone {E'} \lmrg \projb \ghTtwo {E'}$, for any $E'$, set of participants such that $E\cap E'=\emptyset$.
\end{itemize*}
The nesting of branching for general $\ghT$ and $\hTE$
may be intricate and tedious; the auxiliary functions
$\up$ and $\ul$ take care of the detail
(see Appendix \appref{subsec:appendix1}\inApp),
in a way that properties $(a)$ and $(b)$ as above
hold for $\ul$, and similar ones for $\up$
(Lemmas \appref{lem:unml} and \appref{lem:unmp},
Appendix \ref{subsec:appendix1}, \inApp).
Generally, such properties ensure that both conclusions $(1)$ and $(2)$
(essential for composing multiple
subprotocols), of Theorem \ref{thm:comp1}, hold. 
\end{example}

Finally, we can state our first compositionality result,
which certifies the definition of $\bbo$.

\begin{theorem}[Building Back a Single Component]\label{thm:comp1} 
  We fix a set of participants $E$, and we are given hybrid types $\ghT$ and $\hTE$, such that:
  \begin{itemize*}
  \item[\rm (a)] $\hpart \hTE \subseteq E$,
  \item[\rm (b)] $\hepart \hTE \cap E = \emptyset$, and
  \item[\rm (c)] $\projb {\ghT} {E} = \loc{} \hTE$, (compatibility \ref{eq:compatibility}).
  \end{itemize*}
  We set $\hT = \bbeoa {(\ghT)} {(\hTE)}$ and we have:
  \begin{enumerate}
  \item $\projb {\hT} {E} = \hTE$ and
  \item for all $E'$, such that $E'\cap E = \emptyset$, $\projb \hT {E'} = \projb \ghT {E'}$.
  \end{enumerate}
   Moreover if $\isglobal{\ghT}$ then $\isglobal{\hT}$.
\end{theorem}
The proof of Theorem \ref{thm:comp1} proceeds by induction on
$\hdepth{\ghT}+\hdepth{\hTE}$.
Its inductive structure follows the
defining equations of $\bbo$ (Figure \ref{fig:bb1-equations})
and it is non-trivial; the full detail
can be found in Appendix
\appref{subsec:appendix1}\inApp{} (Theorem \appref{thm:comp1-a}),
together with the auxiliary lemmas for merging. Theorem \ref{thm:comp1}
ensures that the result of building back (backwards,
with respect to the usual direction of projection)
$\hT=\bbeoa {(\ghT)} {(\hTE)}$
contains both (1) the information for the internal communication
in $\hTE$ (i.e., $\projb \hT E = \hTE$) and (2) the information for
the external communication prescribed by $\ghT$ (i.e., for all $E'$,
such that $E'\cap E =\emptyset$, $\projb \hT {E'} = \projb \ghT {E'}$).
We describe the algorithm of $\bbo$ and the proof outline
of Theorem \ref{thm:comp1} below, via example.

\begin{example}[Definition \ref{def:bb1} and Theorem \ref{thm:comp1}] \label{ex:pprot-bb1} From Example \ref{ex:pprot}, we consider the subprotocol for the strategy department.
  \begin{small}
\[
    \hTstr=
    \hmsg\dir\ad \dht{\product(\tnat).\hsnd \dir \sales \product(\tnat).\hsnd \dir \fone \product(\tnat).}
    \hrec \hX {\hrcv \fone \dir}
    \dht{\left\{
        \ok.\hmsg\dir\ad\go.\hend,
        \wait.\hmsg\dir\ad\wait.\hX
      \right\}}
\]
  \end{small}
The following protocol, described by the chief designer of the company $D$, coordinates the communication among the three departments (and ignores their internal one). 
\begin{small}
\[
\gG = 
\hmsg \dir \sales \dht{\product(\tnat).}\hmsg \dir \fone \dht{\product(\tnat).}
\hrec \hX {\hmsg \fone \dir}
\dht{\left\{
    \ok.\hmsg\fone\sales\price(\tnat).\hend,
    \wait.\hmsg\fone\sales\wait.\hX
  \right\}}
\]
\end{small}
We observe that $\hpart\hTstr = \{\dir,\ad\}$, and that
compatibility \ref{eq:compatibility} holds:
  \begin{small}
  \[
\loc {} \hTstr  = \projb \gG {\{\dir,\ad\}} =
\hsnd \dir \sales \dht{\product(\tnat).\hsnd \dir \fone \product(\tnat).} \hrec \hX {\hrcv \fone \dir}
    \dht{\left\{
        \ok.\hend, 
        \wait.\hX
      \right\}}
\]
\end{small}
%
To obtain $\ghTone= \bboa {\{\dir,\ad\}} {(\gG)} {(\hTstr)}$,  we first build back the internal global prefix in $\hTstr$:\\
\begin{small}
\centerline{$
\ghTone = \hmsg\dir\ad \dht{\product(\tnat).\ghTone'}
$}\\
\end{small}
We then proceed by induction, namely $\dht{\ghTone'}$ is built by composing $\ghT$ and the smaller hybrid type obtained from $\hTstr$, by removing this first prefix: $\hsnd \dir \sales \dht{\product(\tnat).\hsnd \dir \fone \product(\tnat).\dots}$. We observe that the first two send constructs correspond to the projection of the two initial messages of $\ghT$ (this is guaranteed by the compatibility condition \ref{eq:compatibility}); we take:\\
\begin{small}
\centerline{$
\ghTone =
\hmsg\dir\ad
\dht{\product(\tnat).
  \hmsg \dir \sales \product(\tnat).
  \hmsg \dir \fone \product(\tnat).
  \dht{\ghTone''}}
$}\\
\end{small}
To obtain $\dht{\ghTone''}$ we observe that the compatibility condition takes care of the recursive construct $\hrec \hX \dots$ and of the first message $\hrcv \fone \dir\dht{\dots}$. After that, in each branch, we need to add first the internal message in $\hTstr$ and then the external messages given by $\gG$. We obtain 
\begin{small}
\[
\dht{\ghTone''}=\hrec \hX {\hmsg \fone \dir}
\dht{\left\{
    \ok.\hmsg\dir\ad\go.\hmsg\fone\sales\price(\tnat).\hend, 
    \wait.\hmsg\dir\ad\wait.\hmsg\fone\sales\wait.\hX
  \right\}}
\]
\end{small}
and, ultimately,\\
\begin{small}
\centerline{$
\begin{array}{ll}
\ghTone=&
\hmsg\dir\ad \dht{\product(\tnat).\hmsg \dir \sales \product(\tnat).\hmsg \dir \fone \product(\tnat).}
\\ &
\qquad\qquad\qquad\qquad\qquad\qquad
\hrec \hX {\hmsg \fone \dir}
\dht{\left\{
  \begin{array}{l}
    \ok.\hmsg\dir\ad\go.\hmsg\fone\sales\price(\tnat).\hend,\\
    \wait.\hmsg\dir\ad\wait.\hmsg\fone\sales\wait.\hX
  \end{array}
  \right\}}
\end{array}
$}
\end{small}
Indeed, $\ghTone$ contains all interactions from both $\gG$ and $\hTstr$. The recursive definition of $\bbo$ (Definition \ref{def:bb1}) and the inductive proof of Theorem \ref{thm:comp1} follow the procedure presented in this example (Theorem \ref{thm:comp1-a}, Appendix \appref{subsec:appendix1}\inApp).
\end{example}

\subsection{Step 2: Multiparty Compositionality}
\begin{figure*}
  \begin{center}
\small
    \pgfdeclarelayer{background}
\pgfsetlayers{background,main}

\begin{tikzpicture}[commutative diagrams/every diagram]

        \node(gH)at (2.25,0) {$\ghT$};
        \node(HL1)at (0,-1.5) {$\dht{\hT'_1}$};
        \node(HL2)at (1.5,-1.5) {$\dht{\hT'_2}$};
        \node(HLN)at (3.5,-1.5) {$\dht{\hT'_N}$};
	\node(H1)at (1,-1.75) {$\hTone$};
	\node(H2)at (2.5,-1.75) {$\hTtwo$};
        \node(hdots)at (3,-1.6){$\dots$};
	\node(HN)at (4.5,-1.75) {$\hTN$};
        \node(nol1)at (3.5,-1) {};
        \node(nor1)at (5.5,-1) {};
        \node(gHb)at (7,0) {$\bboa {E_1} {(\ghT)} {(\hTone)}$};
        \node(HL2b)at (6.25,-1.5) {$\dht{\hT'_2}$};
        \node(HLNb)at (8.25,-1.5) {$\dht{\hT'_N}$};
	\node(H1b)at (5.5,-1.75) {$\hTone$};
	\node(H2b)at (7.25,-1.75) {$\hTtwo$};
        \node(hdotsb)at (7.75,-1.6){$\dots$};
	\node(HNb)at (9.25,-1.75) {$\hTN$};
        \node(nol2)at (9.25,-1) {};
        \node(nor2)at (10.75,-1) {};
        \node(nol3)at (8.25,-1) {};
        \node(nor3)at (9.75,-1) {};
        \node(gHf)at (12,-0) {$\bba {(\ghT)} {([\hTone,\dots,\dht{\hT_N}])}$};
	\node(H1f)at (10.75,-1.75) {$\hTone$};
	\node(H2f)at (11.5,-1.75) {$\hTtwo$};
        \node(hdotsf)at (12.25,-1.6){$\dots$};
	\node(HNf)at (13.5,-1.75) {$\hTN$};

	\path[commutative diagrams/.cd,every arrow,font=\footnotesize]
        (gH) edge node[left,yshift=1.2mm] {$\upharpoonright$} (HL1)
        (gH) edge node[left] {$\upharpoonright$} (HL2)
        (gH) edge node[right,yshift=1.2mm] {$\upharpoonright$} (HLN)
        (H1) edge node[above,xshift=3mm] {$\loc{}{}$} (HL1)
        (H2) edge node[above,xshift=1mm] {$\loc{}{}$} (HL2)
        (HN) edge node[above,xshift=1mm] {$\loc{}{}$} (HLN)
        (gHb) edge node[right] {$\upharpoonright$} (HL2b)
        (gHb) edge node[right,yshift=1.2mm] {$\upharpoonright$} (HLNb)
        (H2b) edge node[above,xshift=1mm] {$\loc{}{}$} (HL2b)
        (HNb) edge node[above,xshift=1mm] {$\loc{}{}$} (HLNb)
        ;

        \path[commutative diagrams/.cd,every arrow, magenta]
        (gHb) edge node[black, left] {$\upharpoonright$} (H1b)
        (gHf) edge node[black, left,yshift=1.2mm] {$\upharpoonright$} (H1f)
        (gHf) edge node[black, left] {$\upharpoonright$} (H2f)
        (gHf) edge node[black, right,yshift=1.2mm] {$\upharpoonright$} (HNf)
        ;

        \path[commutative diagrams/.cd,every arrow, magenta, thick, font=\scriptsize]
        (nol1) edge node[above,yshift=3mm] {Theorem \ref{thm:comp1}} (nor1)
        (nol2) edge node[dashed,above,xshift=-4mm,yshift=3mm] {Theorem \ref{thm:comp}} (nor2)
        ;
        \path[commutative diagrams/.cd, dashed, magenta, thick, font=\scriptsize]
        (nol3) edge node[above,yshift=3mm] {} (nor3);
\end{tikzpicture}
  \end{center}
  \caption{Composing Multiple Protocols: Theorem
      \ref{thm:comp1} as the base case, Theorem
      \ref{thm:comp} as the inductive step.}
  \label{fig:comp}
\end{figure*}

$\bbeoa {(\ghT)} {(\hTE)}$ composes the subprotocol $\hTE$
with the compatibility type $\ghT$.
Here, we 
iterate this process
for an arbitrary number of subprotocols $\hTone,\dots,\hTN$,
whenever
compatible with respect to
$\ghT$ 
(Equation \ref{eq:compatibility}):
we achieve full
\emph{multiparty compositionality} of subprotocols.
The overview is given in Figure~\ref{fig:comp}.

\begin{definition}[Build-Back]\label{def:bb}
Given a list of (disjoint) sets of participants $L=[E_1,\dots,E_N]$, we define the partial recursive function $\bbafor L {(\ghT)} {([\hTone,\dots,\hTN])}$ as follows:
\[
\begin{small}
\begin{array}{l}
\bbafor {[E_1]} {(\ghT)} {([\hTone])} = \bboa {E_1} {(\ghT)} {(\hTone)}\\
\bbafor {[E_1,\dots,E_{n+1}]} {(\ghT)}  {([\hTone,\dots,\dht{\hT_{n+1}}])} =
\bbafor {[E_2,\dots,E_{n+1}]} {(\bboa {E_1}\ghT{\hTone})} {([\hTtwo,\dots,\dht{\hT_{n+1}}])}
\end{array}
\end{small}
\]
\end{definition}

From now on, we 
leave implicit the first list argument (of sets of participants): we write $\bba {(\ghT)}  {([\hTone,\dots,\hTN])}$ for $\bbafor {L} {(\ghT)}  {([\hTone,\dots,\hTN])}$.

\begin{theorem}[Compositionality for Multiple Protocols]\label{thm:comp} 
  We are given $E_1,\dots,E_N$ sets of roles, and the hybrid types $\hTone,\hTtwo,\dots,\hTN$, and $\ghT$, such that:
  \begin{itemize*}
  \item[(a)] $E_i\cap E_j=\emptyset$ for all $i\neq j$,
  \item[(b)] $\hpart \hTi \subseteq E_i$ for all $i$,
  \item[(c)] $\hepart \hTi \cap E_i = \emptyset$ for all $i$, and
  \item[(d)] $\projb {\ghT} {E_i} = \loc{} \hTi$ (compatibility \ref{eq:compatibility}).
  \end{itemize*}
  We set $\hT=\bba {(\ghT)} {([\hTone,\dots,\hTN])}$ and we have that, for all $i$, $\projb {\hT} {E_i} = \hTi$. Moreover if $\isglobal{\ghT}$ then $\isglobal{\hT}$.
\end{theorem}

\begin{proof}
  By induction on $N$. We add to the thesis: $(\textit{dsj})$ for all $E'$, $E'\cap\bigcup_{i=1,\dots,N}E_i=\emptyset$, $\projb\hT {E'} = \projb{\ghT}{E'}$, since we 
  need $(\textit{dsj})$ within the induction hypothesis.
  The case $N=1$ is Theorem \ref{thm:comp1}. For $N=n+1$, we set
  $\ghTb = \bboa {E_1}{(\ghT)}{(\hTone)}$ and we apply the induction hypothesis
  to $\hT=\bbafor {[E_2,\dots,E_{n+1}]} {(\ghTb)} {([\hTtwo,\dots,\dht{\hT_{n+1}}])}$: we obtain that for $i=2,\dots,n$, $\projb \hTb {E_i} = \hTi$. For $E_1$, since $E_1\cap\bigcup_{i=2,\dots,n+1}E_i=\emptyset$, thanks to $(\textit{dsj})$, we have $\projb \hT {E_1} = \projb \ghTb {E_1} = \hTone$ (by Theorem \ref{thm:comp1}). For $E'$, $E'\cap\bigcup_{i=1,\dots,n+1}E_i=\emptyset$, we have that $E'\cap\bigcup_{i=2,\dots,n+1}E_i=\emptyset$, and hence for $(\textit{dsj})$, $\projb \hT {E'}=\projb \hTb {E'}$. We conclude by observing that $E'\cap E_1=\emptyset$ and thus, by Theorem \ref{thm:comp1}, $\projb\hTb {E'} = \projb\ghT {E'}$.
\end{proof}


\begin{example}[Theorem \ref{thm:comp}]\label{ex:pprot-bb} 
In Example \ref{ex:pprot-bb1}, we have seen how to build back $\gG$ and $\hTstr$ into $\ghTone$, a new hybrid type containing the information both for the inter-protocol communication (from $\gG$) and for the communication inside the strategy department (from $\hTstr$). We observe that compatibility \ref{eq:compatibility} holds not only for $\hTstr$, but also for $\hTsales$ and $\hTfin$, namely:
  \begin{small}
  \[
    \begin{array}{l}
      \loc {} \hTstr  = \projb \gG {\{\dir,\ad\}} =
      \hsnd \dir \sales \dht{\product(\tnat).\hsnd \dir \fone \product(\tnat).} \hrec \hX {\hrcv \fone \dir}
      \dht{\left\{
          \ok.\hend,
          \wait.\hX
        \right\}}\\
      \loc {} \hTsales  = \projb \gG {\{\sales,\web\}}=
\hrcv \dir \sales \dht{\product(\tnat).} \hrec \hX {\hrcv \fone \sales}
    \dht{\left\{
        \price(\tnat).\hend,
        \wait.\hX
      \right\}}\\
\loc {} \hTfin  = \projb \gG {\{\fone,\ftwo\}}=
\hrcv \dir \fone \dht{\product(\tnat).} \hrec \hX {\hsnd \fone \dir}
    \dht{\left\{
        \ok.\hsnd\fone\sales\price(\tnat).\hend,
        \wait.\hsnd\fone\sales\wait.\hX
      \right\}}\\
    \end{array}
  \]
\end{small}
The hypothesis of Theorem \ref{thm:comp} holds and thus we can build $\G=\bba{(\gG)}{([\hTstr,\hTsales,\hTfin])}$, such that $\projb \G {\{\dir,\ad\}} = \hTstr$, $\projb \G {\{\sales,\web\}} = \hTsales$, and $\projb \G {\{\fone,\ftwo\}} = \hTfin$. To make the construction of $\G$ explicit, we follow the inductive proof structure of Theorem \ref{thm:comp}. The base case is taken care of in Example \ref{ex:pprot-bb1}, where we apply Theorem \ref{thm:comp1} and build $\ghTone$, by composition of $\gG$ and $\hTstr$, we obtain :\\
  \begin{small}
\centerline{$
\begin{array}{l}
\ghTone=
\hmsg\dir\ad \dht{\product(\tnat).\hmsg \dir \sales \product(\tnat).\hmsg \dir \fone \product(\tnat).}
\\ \qquad\qquad \qquad\qquad \qquad\qquad
\hrec \hX {\hmsg \fone \dir}
\dht{\left\{
  \begin{array}{l}
    \ok.\hmsg\dir\ad\go.\hmsg\fone\sales\price(\tnat).\hend,\\
    \wait.\hmsg\dir\ad\wait.\hmsg\fone\sales\wait.\hX
  \end{array}
  \right\}}
\end{array}
$}\\
\end{small}
We observe that $\projb \ghTone {\{\sales,\web\}} = \projb \gG {\{\sales,\web\}}$, hence, since \ref{eq:compatibility} still holds, we can apply again the build-back procedure and obtain:\\
  \begin{small}
\centerline{$
\begin{array}{l}
\ghTtwo=
\hmsg\dir\ad \dht{\product(\tnat).\hmsg \dir \sales \product(\tnat).}
\dht{\hmsg \dir \fone \product(\tnat).}
\\ \qquad \qquad \qquad
\hrec \hX {\hmsg \fone \dir}
\dht{\left\{
  \begin{array}{l}
    \ok.\hmsg\dir\ad\go.\hmsg\fone\sales\price(\tnat).\hmsg\sales\web\publish(\tnat).\hend,\\
    \wait.\hmsg\dir\ad\wait.\hmsg\fone\sales\wait.\hmsg\sales\web\wait.\hX
  \end{array}
  \right\}}
\end{array}
$}
  \end{small}
$\ghTtwo$ collects
all the interactions from $\gG$, $\hTstr$, and $\hTsales$.
To obtain a type $\G$ that also includes
the internal interactions of $\hTfin$,
we perform one more
induction step,
building back from $\ghTtwo$ and $\hTfin$. 
\begin{small}
\[
\begin{array}{l}
\G=
\hmsg\dir\ad \dht{\product(\tnat).\hmsg \dir \sales \product(\tnat).\hmsg \dir \fone \product(\tnat).}\dht{\hmsg \fone\ftwo \product(\tnat).}
\\ \qquad 
\hrec \hX {\hmsg \ftwo\fone}
\dht{\left\{
  \begin{array}{l}
    \begin{array}{l}
      \price(\tnat).\hmsg \fone \dir\ok.\hmsg\dir\ad\go.
      \hmsg\fone\sales\price(\tnat).\hmsg\sales\web\publish(\tnat).\hend,
      \end{array}\\
    \wait.\hmsg \fone \dir\wait.\hmsg\dir\ad\wait.\hmsg\fone\sales\wait.\hmsg\sales\web\wait.\hX
  \end{array}
  \right\}}
\end{array}
\]
\end{small}
A graphical representation of $\G$ can be found in Appendix \ref{appendix-fig}\inApp.
\end{example}

Theorem \ref{thm:comp} gives
a technique for composing
multiple subprotocols into a more general one.
The next, conclusive step
proves
that compositionality
well-behaves with respect to
MPST projection.

\subsection{Step 3: Compositionality through Projection}\label{sec:comp-proj}

With Definition \ref{def:proj}, we have generalised the MPST projection to hybrid types. We prove that generalised projection well-behaves with respect to set inclusion.

\begin{theorem}[Projection Composes over Set Inclusion]
  \label{thm:proj-comp}
  Given $\hT$, $E_1$ and $E_2$, $E_2\subseteq E_1$, if $\projb \hT {E_1}$ is defined, then $\projb {(\projb \hT {E_1})} {E_2} = \projb \hT {E_2}\ $.
\end{theorem}
\begin{proof}
By structural induction on $\hT$ (see Appendix \appref{subsec:appendix2}\inApp).
\end{proof}

Theorem \ref{thm:proj-comp} is the last fundamental
ingredient to achieve
\emph{distributed protocol specification}.

\begin{corollary}[Distributed Protocol Specification]\label{cor:dist-spec}
  Given $E_1,\dots,E_N$ disjoint sets of participants, a global type $\gG$, and $\hTone,\dots,\hTN$ hybrid types, such that:
  \begin{itemize*}
  \item[\rm (a)] $\hpart \hTi = E_i$ for all $i$,
  \item[\rm (b)] $\hepart \hTi \cap E_i = \emptyset$ for all $i$,
  \item[\rm (c)] $\hpart \gG \subseteq \bigcup_{i=1,\dots,N} E_i$, and
  \item[\rm (d)] $\projb {\gG} {E_i} = \loc{} \hTi$ (compatibility \ref{eq:compatibility});
  \end{itemize*}
  there exists $\G$ such that, for all $i$, for all $\pr\in E_i$, $\projb {\G} {\{\pr\}} = \projb {\hTi} {\{\pr\}}$.
\end{corollary}

\begin{proof} We set $\G=\bba {(\gG)} {([\hTone,\dots,\hTN])}$; by Theorem \ref{thm:comp}, $\G$ is global and such that $\projb \G {E_i} = \hTi$. Then, we apply Theorem \ref{thm:proj-comp} and we obtain, for $\pr \in E_i$, $\projb \G {\{\pr\}} = \projb {\projb \G {E_i}} {\{\pr\}} =  \projb {\hTi} {\{\pr\}}$. We observe that $\projb \G {\{\pr\}}$ is a local type, for all $\pr$, (Remark \ref{rem:proj}).
\end{proof}

\begin{example} \label{ex:dist-ver}
In Examples \ref{ex:pprot-bb1} and \ref{ex:pprot-bb},
the protocol designer for each department
has given their hybrid type, $\hTstr$, $\hTsales$, and $\hTfin$,
disciplining internal communication (with messages $\nhmsg \p \q$)
and specifying the 
communication with
other departments
(with sends/receives, $\nhsnd \p \q$/$\nhrcv \p \q$).
Compatibility (Equation \ref{eq:compatibility})
has been verified against 
$\gG$, as described by the chief designer $D$.
In (Example \ref{ex:pprot-bb}) we build back
$\G$ (we observe that is $\isglobal\gG$ implies $\isglobal\G$).
Corollary \ref{cor:dist-spec} holds and projections of
components $\hTstr$, $\hTsales$, and $\hTfin$
onto single participants are local types,
also projections of $\G$ (Theorem \ref{thm:proj-comp}),
the global type disciplining
the whole system.
For instance, if we
wanted to get the local type for 
$\dir$, traditionally,
we would do so by projecting
$\G$. With our distributed protocol specification, it
is enough to project $\hTstr$ onto $\dir$.
\begin{small}
\[
\begin{array}{l}
  \projb \G {\{\dir\}} = \projb {(\projb \G {\{\dir,\ad\}})}{\{\dir\}} = \projb \hTstr {\{\dir\}} =
  \\[1mm] \qquad
    \hsnd\dir\ad \dht{\product(\tnat).\hsnd \dir \sales \product(\tnat).\hsnd \dir \fone \product(\tnat).}
    \hrec \hX {\hrcv \fone \dir}
    \dht{\left\{
        \ok.\hsnd\dir\ad\go.\hend,
        \wait.\hsnd\dir\ad\wait.\hX
      \right\}}
\end{array}
\]
\end{small}
If instead 
want to implement 
processes for $\fone$ and $\ftwo$,
we can obtain the local types 
from $\hTfin$.
  \begin{small}
\[
\begin{array}{l}
  \projb{\G}{\{\fone\}}=\projb\hTfin{\{\fone\}}=
  \hrcv \dir \fone \dht{\product(\tnat).\hsnd \fone \ftwo \product(\tnat).}
  \hrec\hX{\hrcv\ftwo\fone}
  \dht{\left\{\hspace*{-1mm}
    \begin{array}{l}
      \price(\tnat).\hsnd \fone \dir\ok.\hsnd \fone\sales \price(\tnat).\hend,\\
      \wait.\hsnd \fone \dir\wait.\hsnd \fone\sales \wait.\hX
    \end{array}
    \hspace*{-1mm}\right\}}\\
  \projb{\G}{\{\ftwo\}}\hspace*{-1mm}=\projb\hTfin{\{\ftwo\}}\hspace*{-1mm}=
  \dht{\hrcv \fone \ftwo \product(\tnat).}
  \hrec\hX{\hsnd\ftwo\fone}
  \dht{\left\{
      \price(\tnat).\hend,
      \wait.\hX
    \right\}}
\end{array}
\]
\end{small}
We can proceed analogously
for all participants.
%
%
\end{example}

\begin{remark}[Applicability and Preservation of Semantics]
\label{rem:app}
  Example \ref{ex:dist-ver} displays the essence of
  of our 
  theory,
  formally captured by Corollary \ref{cor:dist-spec}:
  for a 
  system specified in a distributed way,
  with components $\hTone,\dots,\hTN$ and
  compatibility type $\gG$,
  there is
  \emph{no need for an explicit description of $\G$}. 
 After compatibility checks (Equation \ref{eq:compatibility}),
  our theory builds back 
  a well-formed global type $\G$ for the whole system, and
  the session of local types,
  projections of $\G$, can be obtained,
  in a distributed fashion, by directly
  projecting subprotocols,
  since, for $\p\in\hpart\hTi$,
  $\projb\hTi{\{\p\}}=\projb\G {\{\p\}}$.

At the design stage,
each designer $D_i$
gives the subprotocol $\hTi$ and,
with the simple equality check \ref{eq:compatibility},
they make sure that their protocol is compatible
with $\gG$ (described by the chief designer $D$).
At the type-checking/implementation stage,
the designer $D_i$ \emph{independently}
obtains local types for
well-behaved implementations
directly from their specification $\hTi$.
$D_i$ is never concerned
with the communication happening
internally to, or among, other components.
What guarantees global well-behaviour
is the existence of $\G$, proved
once and for all by our theory;
no designer or programmer needs
its explicit description.
We have achieved
\emph{distributed protocol specification}.

Desirable MPST semantic guarantees,
such as liveness and deadlock freedom,
are preserved by our theory, thanks
to its \emph{semantics preservation}. 
Our compositionality-through-projection
technique can explicitly build back
a protocol as a global type that is
traditionally well-formed
(see, e.g., \cite{HYC2016,DenielouYoshida2013}).
Thus, our theory brings modularity to
the protocol design phase,
but, after such distributed specification,
the result is a
traditional MPST system,
with a \emph{single} global type that
projects on local types for
\emph{all} participants,
which benefits from existing
semantics results
from the MPST literature.


\end{remark}

\begin{remark}[On Hybrid Types]
\label{rem:on-hybrid}
Central to this work
is Definition \ref{def:hybrid-types}
of hybrid types.
Our theory
shows how
``open'' subprotocols, interacting
with other components of the system,
can be specified as \emph{hybrid types},
safely composed into a \emph{global type},
and projected onto \emph{local types}.

The syntax of hybrid types is
simply the combination of the
syntaxes of global and local types
and, through the predicates
$\isglobal{\hT}$ and $\islocal{\hT}$,
we can 
isolate
global and local types
respectively,
from the rest of hybrid types.
This choice
makes our development compatible
with existing MPST systems and
is key to semantics preservation:
in our compositionality
theory, well-formed global types
guarantee semantics properties
and local types are used for
participant implementation,
exactly as in traditional MPST
(see Corollary \ref{cor:dist-spec}
and Remark \ref{rem:app}).

At the same time, dealing with
a single syntax (hybrid types)
simplifies our theory
significantly. This paper proposes
an approach to protocol compositionality
that heavily relies on projection.
Traditional MPST projection
operates on global types and returns
local types for implementation.
Our generalised projection, instead,
takes a hybrid type 
and returns a hybrid type,
but, since global and local types
are hybrid types, our projection
maintains and extends
the functionalities of the traditional
operator. The main gain is flexibility:
a function with the same domain as its
codomain can be composed with itself
and this is central in our proofs
(see Theorem \ref{thm:proj-comp}
and its role in the proof of Corollary \ref{cor:dist-spec}).
In other words, instead of working with multiple
operators (which would have very similar definitions)
and proving them compatible, we
rely on a single one: generalised projection.
In particular, we apply projection
in the following key steps of our development:
\begin{itemize}
\item projecting the component $\hTE$ of a system
onto one of its internal participants $\p$, gives---as
it is customary in the literature---the local type
$\dlt{\lT_{\p}}$ for implementing $\p$;
\item projecting the type $\gG$ is necessary for compatibility \ref{eq:compatibility}:
$\projb {\gG} {E} = \loc{}\hTE$; and
\item the preservation of projection (Theorem \ref{thm:proj-comp})
ensures
applicability and correctness:
local types for implementation, obtained
by separately projecting the distributed components, are
also projections
of a single, well-formed global type for the whole system
(equation $\projb {\G} {\{\pr\}} = \projb {\hTi} {\{\pr\}}$,
Corollary \ref{cor:dist-spec}).
\end{itemize}

\end{remark}

\section{Case Studies}
\label{sec:eval}

In this section, we evaluate our development
with case studies. 
In \S\ref{sec:parallel}, we discuss the role
of the \emph{parallel construct}
$\hpar \hTone \hTtwo$
in our protocol composition.
Then, 
we consider the
industry-standard protocol for authorisation,
OAuth 2.0, \cite{oauth2,horne:2020}:
in \S\ref{subsec:oauth}, we observe the \emph{modularity}
benefits of our theory, and, in \S\ref{subsec:oauth2},
we reach an \emph{optimisation} for it.
In \S\ref{subsec:del}, we show
how hybrid types
can be smoothly extended to 
feature
\emph{delegation and explicit connections}. 

\subsection{On the Parallel Construct and Compositionality}\label{sec:parallel}


Our work enriches the type system
of MPST with compositionality
at the protocol-description level,
while retaining the traditional syntax
of local types. Consequently,
the targeted process language is
standard \cite{zooid,GHH2021}.
In particular, local types and
processes are \emph{single-threaded}.
On the other hand,
to allow the description of
protocols where two different components
execute independently, without
exchanging messages with each other,
we have added the
parallel construct
to the syntax of hybrid
(and global) types.

The
\emph{parallel construct} (or \emph{parallel composition})
for global types
appears in
the first presentation of MPST
\cite{Honda2008Multiparty},
but dismissed in subsequent literature
\cite{Bettini2008Global,CDPY2015,zooid}.
For achieving compositionality,
without requiring further well-formedness restrictions on global types, %
we need to explicitly add parallel composition
$\hpar\hTone\hTtwo$.
Let us consider the following distributed specification: 
\\[1mm]
\begin{small}
\centerline{$
\ghTpar = \hmsg \p \pr \dht{\elllbl.\hend}\qquad
  \hTparone = \hsnd\p\pr\dht{\elllbl.\hrec \hX {\hmsg\p\q\{\lblFmt{\ell_1}.\hX,\lblFmt{\ell_2}.\hend\}}}\qquad
  \hTpartwo =\hrcv\p\pr\dht{\elllbl.\hrec \hY {\hmsg\pr\ps\{\lblFmt{\ell_3}.\hY,\lblFmt{\ell_4}.\hend\}}}
$}
\end{small}
Compatibility holds:
$\projb \ghTpar {\{\p,\q\}} = \hsnd \p \pr \dht{\elllbl.\hend} \allowbreak= \loc {} \hTparone$ and
$\projb \ghTpar {\{\pr,\ps\}} \allowbreak= \hrcv \p \pr \dht{\elllbl.\hend} = \loc {} \hTparone$.
Our theory guarantees the existence of a well-formed global type
for the whole system, which entails a deadlock-free session:
$\Gpar=\hmsg\p\pr\dht{\elllbl.}
\dht{\hpar{\ \big(\hrec \hX {\hmsg\p\q\{\lblFmt{\ell_1}.\hX,\lblFmt{\ell_2}.\hend\}}\big)\ } {\ \big(\hrec \hY {\hmsg\pr\ps\{\lblFmt{\ell_3}.\hY,\lblFmt{\ell_4}.\hend\}}\big)\ } }$.
Without 
parallel composition,
it is not clear how
to compose
$\hTparone$ and $\hTpartwo$,
even if they are compatible
with respect to
$\ghTpar$.
\footnote{For more details on the role of parallel composition in the inductive proof of Theorem \ref{thm:comp1}, see the full proof in Appendix \appref{subsec:appendix1}\inApp, Theorem \appref{thm:comp1-a}} %
Indeed, recent work \cite{GHH2021} has
drawn attention to the role of the parallel
construct in traditional MPST: by exploiting
a similar example to the $\Gpar$ above, the authors
show that, if the syntax does
not include 
parallel composition,
there are non-deadlocked sessions
that do not have a global type.

\subsection{Distributed Specification for OAuth 2.0}
\label{subsec:oauth}

We consider the industry-standard protocol for authorisation, OAuth 2.0, \cite{oauth2,horne:2020}. In such protocol, the owner of a resource gives approval (through the OAuth server) for an external application to access some resource; the OAuth server ensures, by means of tokens, that the sensitive data of the owner are not shared with the external application.


\begin{figure}
\begin{small}
\[\begin{array}{l}
  \hToaone = \hrcv \ua \oa \dht{\{\init (\appid,\scope).}\hmsg \oa \ow \dht{\{\login(\appid,\scope).}\hmsg \ow \oa\\
  \dht{\left\{
  \begin{array}{l}
    \deny. \hsnd \oa \ua \dht{\close.\hsnd \oa \res \{\release.\hend\}}, \\
    \auth (\name,\pwd).
    \hsnd \oa \ua 
\dht{\left\{
    \begin{array}{l}
      \dht{\close.\hsnd \oa \res \{\release.\hend \}}, \\
      \dht{\authcode(\cod).}
      \hrcv \ua \oa 
      \dht{\{\exchange(\appid,\secret,\cod). \hsnd \oa \ua}\\
      \dht{\left\{
      \dht{\close.\hsnd \oa \res \{\release.\hend \}}, 
      \dht{\accesstoken(\tkn).\hsnd \oa \res \{\pass.\hend\} }
    \right\}}
  \dht{\}}
    \end{array}
    \right\}}
  \end{array}
  \right\}}
  \dht{\}\}}\\
  \hToatwo = \hsnd \ua \oa \dht{\{\init(\appid,\scope).\hrcv \oa \ua}
  \dht{\left\{
  \begin{array}{l}
      \dht{\close.\hrcv \oa \res \{\release.\hend \}}, \\
      \dht{\authcode(\cod).}\hsnd \ua \oa \dht{\{\exchange(\appid,\secret,\cod). \hrcv \oa \ua}\\
  \dht{\left\{
    \begin{array}{l}
      \dht{\close.\hrcv \oa \res \{\release.\hend \}}, \\
      \dht{\accesstoken(\tkn).\hrcv \oa \res \{\pass.\hresa\} }
    \end{array}
    \right\}}
  \dht{\}}
    \end{array}
    \right\}}
  \dht{\}}\\
  \hresa = \hrec \hX
       {\hmsg \ua \res \dht{ \{\request(\tkn).}}
       {\hmsg \res \ua
         \dht{\left\{
             \revoke.\hend,
             \response(\data).\hX
         \right\}}
         \dht{\}}}\\
         \ghToa = \hmsg \ua \oa \dht{\{\init(\appid,\scope).}\hmsg \oa \ua\\
  \dht{\left\{
    \begin{array}{l}
      \dht{\close.\hmsg \oa \res \{\release.\hend \}},\\
      \dht{\authcode(\cod).\hmsg \ua \oa \{\exchange(\appid,\secret,\cod). \hmsg \oa \ua}\\
      \dht{\left\{
      \dht{\close.\hmsg \oa \res \{\release.\hend \}}, 
      \dht{\accesstoken(\tkn).\hmsg \oa \res \{\pass.\hend\} }
    \right\}}\dht{\}}
    \end{array}
  \right\}}
  \dht{\}}
    \end{array}
\]
\end{small}
      \caption{Distributed Specification for OAuth 2.0}\label{fig:oauth}
\end{figure}

We present a specification for the OAuth 2.0 protocol in two components (Figure \ref{fig:oauth}): a first designer $D_{\sf auth}$ gives the hybrid type $\hToaone$ for the \emph{OAuth server} $\oa$ and the resource \emph{owner} $\ow$, while a second $D_{\sf res}$ is in charge of the interactions involving an \emph{untrusted app} $\ua$ and the \emph{resource service} $\res$, $\hToatwo$. Separately,
a chief designer
describes the compatibility type $\ghToa$. Compatibility holds,
$\projb \ghToa {\{\oa,\ow\}} = \loc {} \hToaone$ and $\projb \ghToa {\{\ua,\res\}} = \loc {} \hToatwo$, and Corollary \ref{cor:dist-spec} guarantees the existence of a well-formed global type for the whole protocol.

Let us focus on the component $\hToatwo$,
and in particular on its subcomponent $\hresa$, which 
contains exclusively internal interactions:
once authorisation is granted, the
untrusted app communicates
directly with the resource service.
We observe then that $\hresa$
can be specified \emph{modularly}:
the designer $D_{\sf res}$ could, at a later stage,
\emph{specify a different protocol}, $\hresab$,
for the interaction between $\ua$ and $\res$,
\emph{without affecting compatibility}. 
We recognise one extra benefit of our theory
in \emph{modular specification for intra-component communication}:
a designer ($D_{\sf res}$) can modify the specification
of their protocol ($\hToatwo$) over time,
in its internal interactions ($\hresa, \hresab$),
as long as its external communication 
($\loc {} \hToatwo$)---and hence compatibility
with respect to the prescription of the chief designer
($\ghToa$)---is preserved.

\subsection{Optimisation of OAuth 2.0 Specification}
\label{subsec:oauth2}

By inspecting the types in Figure \ref{fig:oauth},
we notice that all inter-component interactions 
(between $\hToaone$ and $\hToatwo$) 
go through the participant $\oa$
and, hence, they are all
documented by $\ghToa$ , but also in $\hToaone$.
Here, describing explicitly $\ghToa$ 
looks redundant 
and we ask ourselves whether a more
efficient specification is possible.
It turns out
that we can omit the specification of $\ghToa$ for
the OAuth 2.0 protocol and
\emph{optimise distributed specification} (Corollary \ref{cor:dist-spec})
in general. 

\begin{corollary}[Distributed MPST Specification, Optimisation]\label{cor:dist-spec2}
  Given $E, E_1,\dots,E_N$ disjoint sets of participants, a global type $\gG$,
  and $\hTone,\dots,\hTN$ hybrid types, such that:
  \begin{itemize*}
  \item[(a)] $\hpart \hTi\allowbreak = E_i$, for all $i=1,\dots,N$,
  \item[(b)] $\hepart \hTi \cap E_i\allowbreak = \emptyset$, for all $i$,
  \item[(c)] $\hpart \gG \allowbreak\subseteq \allowbreak(\bigcup_{i=1,\dots,N} E_i)\cup E$, and
  \item[(d)] $\projb {\gG} {E_i} = \loc{} \hTi$;
  \end{itemize*}
  there exists $\G$ such that, for all $i=1,\dots,N$, for all $\pr\in E_i$, $\projb {\G} {\pr} = \projb {\hTi} {\{\pr\}}$, and for all $\pr\in E$,  $\projb {\G} {\pr} = \projb {\gG} {\{\pr\}}$.
\end{corollary}

\begin{figure}
\begin{small}
$
\begin{array}{l}
  \gGoa = \hmsg \ua \oa \dht{\{\init(\appid,\scope).}\hmsg \oa \ow \dht{\{\login(\appid,\scope).}\hmsg \ow \oa\\
  \dht{\left\{
  \begin{array}{l}
    \deny. \hmsg \oa \ua \dht{\close.\hmsg \oa \res \{\release.\hend\}}, \\
    \auth (\name,\pwd).
    \hsnd \oa \ua 
    \dht{\left\{
    \begin{array}{l}
      \dht{\close.\hmsg \oa \res \{\release.\hend \}}, \\
      \dht{\authcode(\cod).}
      \hmsg \ua \oa 
      \dht{\{\exchange(\appid,\secret,\code). \hmsg \oa \ua}\\
      \dht{\left\{
      \dht{\close.\hmsg \oa \res \{\release.\hend \}}, 
      \dht{\accesstoken(\tkn).\hmsg \oa \res \{\pass.\hend\} }
    \right\}}
  \dht{\}}
    \end{array}
    \right\}}
  \end{array}
  \right\}}
  \dht{\}\}}\\
  \text{$\hToatwo$ (and $\hresa$) the same as in Figure \ref{fig:oauth}}
\end{array}
$
\end{small}
\caption{Distributed for OAuth 2.0, Optimised}
\label{fig:oauth-2}
\end{figure}

The proof of the above corollary is analogous to the proof of Corollary \ref{cor:dist-spec} and entails the same semantic guarantees, since it leads to the existence of a global type for the whole system. 
However, in Corollary \ref{cor:dist-spec2}, $\gG$ plays a twofold role:
\begin{itemize*}
\item[(a)] it is the hybrid type, communication subprotocol for the component
with set of participants $E$, and
\item[(b)] it is the compatibility protocol, carrying all the inter-component
interactions of the system.
\end{itemize*}
For OAuth 2.0,
by exploiting Corollary \ref{cor:dist-spec2},
we obtain a more efficient distributed
specification 
(Figure \ref{fig:oauth-2}).
Only two hybrid types, one for each component,
are specified, provided that
one contains also the compatibility information.
We observe, with notations from Figures \ref{fig:oauth}
and \ref{fig:oauth-2} that
$\projb \hToaone {\{\p\}} = \projb \gGoa {\{\p\}}$, for $\p\in \{\oa, \ow\}$. 
Hence, for all participants,
in both specifications we obtain the same local
types for implementation. However, 
in the optimised case of
Figure \ref{fig:oauth-2}, for compatibility, we only need to check one equality: $\projb \gGoa {\{\ua,\res\}} = \loc {} \hToatwo$.

\begin{remark}[Local Types from Partial Protocols]
  We observe that Corollaries \ref{cor:dist-spec}
  and \ref{cor:dist-spec2} prescribe differently
  how to obtain local types for implementation.
In Corollary \ref{cor:dist-spec}, we
project each $\hTi$ onto its internal participants (in $E_i$)
to get the right local types. 
  $\gG$ has the sole role of disciplining
  inter-component communication for compatibility,
  hence all its projections onto single
  participants (even if well-defined) are \emph{not} meant
  for implementation. In Corollary \ref{cor:dist-spec2},
  instead, the global type $\gG$ not
  only provides for compatibility, but also
  describes the internal communication
  of 
  the component that is concerned
  with the participants in $E$.
  Thus, if $\pr\in E_i$, the respective
  local type is obtained by $\projb \hTi {\{\pr\}}$,
  while, if $\pr\in E$, the local type for implementing
  $\pr$ is $\projb \gG {\{\pr\}}$.
\end{remark}



\subsection{An Extension to Delegation and Explicit Connections}
\label{subsec:del}

The MPST literature offers a
variety of formalisms that enrich
the type system with expressive
features
\cite{10.1145/3485501,BYY2014,ZFHNY2020,10.1145/3428223,lagaillardie_et_al:LIPIcs.ECOOP.2022.4}, while
maintaining the central mechanism
of projecting global types onto
local types for distributed
implementation. This
suggests that our
compositional methodology
is general enough to capture
more sophisticated formalisms
than core MPST.
In this section,
we support this intuition
with a case study:
we extend
hybrid types to include
\emph{delegation}
and \emph{explicit connections}.
At the end of the section
(Example \ref{ex:del}),
we show that this suitability
for extensions is a
prerogative of
our compositionality-through-projection,
differently from
other compositionality approaches
based on input/output matching
\cite{BLTDezani,smdg:facs21}.

Both delegation and explicit connections
are relevant and practical features,
which have been extensively studied by
the literature on concurrency.
Delegation---the mechanism
in which a participant appoints a different
participant to act on their behalf---first
appears in the context
of object-oriented concurrency
\cite{self,del-book,10.1145/83880.84528}
and, naturally, it has been implemented in
mainstream object-oriented languages
\cite{imai_et_al:LIPIcs:2020:13166,HU07TYPE-SAFE,scalas_et_al:LIPIcs:2016:6115}.
Over the years,
the session-type literature has investigated the
verification of concurrency
in the presence of delegation
\cite{Honda:1998,Honda2008Multiparty,Bettini2008Global,SDHY2017}.
In particular, we take the approach
of \citet{CASTELLANI2020128},
who treat delegation as internal
to the session, in contrast with
channel-passing delegation,
which requires the interleaving of sessions.
Thus, the authors can model (internal)
delegation simply by adding specific
constructs to the syntax of global types.
Moreover, global types from this paper
benefit from the flexibility of
explicit connections:
some participant may or may not
take part in the communication,
depending on the
choice made by some other participant
at a previous stage of execution.
Explicit connections are common
in the design of real-world protocols
\cite{oauth2,kerberos} and
have been 
significantly addressed by
the literature
\cite{
HY2017,DBLP:journals/acta/CastellaniDG19,
harvey_et_al:LIPIcs.ECOOP.2021.10,gheri_et_al:LIPIcs.ECOOP.2022.8}.
Specifically, our
approach 
is immediately
compatible with
the type system
and the semantics in
\cite{CASTELLANI2020128}
(from which we take
most notation):
we can compose subprotocols into
a well-formed and
\emph{well-delegated} global type,
projecting on
local session types
for the whole communicating system.
Thus, 
no new semantics is needed,
but MPST semantics guarantees
(subject reduction, session fidelity, and progress)
hold.
%

For this case study, we focus on
our novel notion of
compatibility through projection
(Equation \ref{eq:compatibility}).
Developing the full theory
goes beyond the scope of this paper:
the structure of the proofs
would be exactly the same as
in our core theory
(Section \ref{sec:comp}).
Instead, we
\begin{itemize*}
\item[$(a)$] extend hybrid types to delegation and explicit connections,
\item[$(b)$] generalise projection,
\item[$(c)$] define the localiser, and 
\item[$(d)$] state compatibility.
\end{itemize*}



\begin{definition}[Hybrid Types with Delegation]\label{def:ht-del}
We define a set of \emph{prefixes} for global and local messages;
when $\p$ establishes an \emph{explicit connection} to $\q$,
we use the superscript $^e$.
\[
\begin{array}{llll}
  \da ::= \dm \p \q \SEP \dme \p \q &
  \dpo ::= \ds \p \q \SEP \dse \p \q &
  \dpi ::= \dr \p \q \SEP \dre \p \q
  &\ms::=\da\SEP\dpo
\end{array}
\]
\emph{Hybrid types with delegation} are defined inductively by:
\[
    \begin{array}{ll}
      \hT ::= & \hend \SEP \hX \SEP\hrec \hX \hT \SEP \dht{\hT_1 | \hT_2}
      \SEP \hchoicep \ell \hT \SEP
      \hint \ell \hT \SEP\\
      & \hfdel \p \q \hT \SEP \hbdel \q \p \hT \SEP
      \hafdel \p \q \hT \SEP \hpbdel \q \p \hT \SEP
      \hpfdel \p \q \hT \SEP \habdel \q \p \hT
\end{array}
\]

\end{definition}
Without loss of generality, we omit payload types
(sorts) from the syntax above: formally,
participants only exchange labels $\elllbl$.
For degenerate branchings (with a single branch, where no actual choice happens)
we omit the operators $\dht{\boxplus}$ and $\dht{\bigwedge}$
and we simply write the message (starting with a prefix $\da,\dpo,\dpi$).

Hybrid types with delegation (Definition \ref{def:ht-del})
endow with local send/receive constructs
the global types from \cite{CASTELLANI2020128}.
In particular, \emph{global choices}
and \emph{union types} from \cite{CASTELLANI2020128}
are particular cases of the above syntax.
The global choice, $\hgchoice\ell\hT$, is obtained
by asking that, in $\hchoicep \ell \hT$,
all $\dht{\ms^{\p}_i}$ are global messages,
$\dht{\da^{\p}_i}$.
Local union (send) types, 
$\hun \ell \hT$, are now written as $\hchoicep\ell\hT$,
where all $\dht{\ms^{\p}_i}$ are send constructs $\dpoi$.
In what follows we discuss how Definition \ref{def:ht-del}
is a direct extension of Definition \ref{def:hybrid-types}.

We observe that branching is more permissive than in
Definition \ref{def:hybrid-types}.
\begin{itemize}
\item 
In a single choice $\hchoicep \ell \hT$ the sender $\p$ is unique, but
receivers may be different, internal or external:
in choices we allow the mixing of global messages $\da$ and send constructs $\dpo$.
\item Intersection types $\hint\ell\hT$ combine receive constructs $\dpi$, possibly
with different senders (and receivers).
\end{itemize}

\begin{remark}[Generalising Hybrid Types to Delegation]
The syntax of hybrid types with delegation
(Definition \ref{def:ht-del}) is a generalisation
of the core syntax
of hybrid types (Definition \ref{def:hybrid-types}).
\footnote{In order to make the notation lighter,
we ignore sorts in both syntaxes:
only labels $\elllbl$ are sent.}
For each construct in
Definition \ref{def:hybrid-types},
we show how it can be expressed in
the formalism of \ref{def:ht-del}.
\begin{itemize}
\item $\hend$, $\hX$, $\hrec \hX \hT$, and $\hpar \hTone \hTtwo$
are preserved.
\item $\dht{{\p}!{\q};\{\lblFmt{\ell_i}. \hT_i \}_{i \in I}}$
can be written as $\hchoicep \ell \hT$, with
$\dht{\ms^{\p}_i}=\ds \p \q$ for all $i\in I$.
\item $\dht{{\p}?{\q};\{\lblFmt{\ell_i}. \hT_i \}_{i \in I}}$
can be written as $\hint \ell \hT$, with
$\dpii=\dr \p \q$ for all $i\in I$.
\item $\dht{{\p}\to{\q};\{\lblFmt{\ell_i}. \hT_i \}_{i \in I}}$
can be written as $\hchoicep \ell \hT$, with
$\dht{\ms^{\p}_i}=\dm \p \q$ for all $i\in I$.
\end{itemize}
\end{remark}

The global construct $\hfdel\p\q\hT$ models a \emph{forward delegation}
where $\p$ delegates their behaviour to $\q$; then,
such behaviour is given back with the global construct
for \emph{backward delegation}, $\hbdel\q\p\hT$.
The notation for local constructs is analogous:
\emph{active forward delegation} $\hafdel\p\q\hT$,
\emph{passive backward delegation} $\hpbdel\q\p\hT$,
\emph{passive forward delegation} $\hpfdel\p\q\hT$, and
\emph{active backward delegation} $\habdel\q\p\hT$.
As 
in our core theory
(see Definition \ref{def:hybrid-types}),
local constructs carry
both internal participants
and external ones (in square brackets):
e.g., in $\hafdel\p\q\hT$,
$\p$ (internal) delegates their behaviour to
the participant $\q$ of a different component (external).

\begin{example}[Global Types with Delegation and Explicit Connections]
\label{ex:d1}
Here, we consider two simple examples
that display in isolation the features
of delegation and explicit connections.
Let us consider the following global protocols,
written with the syntax of
hybrid types with delegation (Definition \ref{def:ht-del}).
\[
\begin{small}
\Gd=\dht{
\dm \buyer \seller (\ok);
\hfdel\seller \bank {\dm \buyer \seller} (\card);
\hbdel\bank \seller \hend
}
\end{small}
\]
In $\Gd$,
after the $\buyer$ has given the $\ok$,
the $\seller$ \emph{delegates} to the $\bank$ their role
in the communication
($\seller \fdel\bank$).
Namely, when the $\buyer$ sends their $\card$ number to
the $\seller$, they are in fact sharing that information
with the $\bank$. With the construct $\bank\bdel \seller$,
the $\bank$ delegates back their role to the $\seller$.
\[
\begin{small}
\Gec=
\dht{
\dm \user\website (\location);
\boxplus
\left\{
\begin{array}{l}
\dme \website\shopeu (\open);
\dm \shopeu\user (\fronteu);\hend
\\
\dme\website\shopuk (\open);
\dm \shopuk\user (\frontuk);\hend
\end{array}
\right\}
}
\end{small}
\]
In $\Gec$, after the $\website$ receives
the $\user$'s location, it decides whether
to \emph{(explicitly) connect} them to the $\shopeu$ or to the
$\shopuk$. Then the $\shopeu$
(resp. the $\shopuk$)
sends the $\fronteu$
(resp. the $\frontuk$)
to the $\user$. Here in the first
branch (resp. the second) of the choice,
there is an explicit connection $\dme \website\shopeu$
(resp. $\dme \website\shopuk$) to
the participant
$\shopeu$ (resp. $\shopuk$),
which does not appear in the other branch.

\end{example}


We adopt the definition of \emph{well-delegated}
type from \cite{CASTELLANI2020128}
(Definition $4.11$); in particular
a forward delegation (e.g., $\hfdel \p \q {}$)
is always followed by a corresponding
backwards one (e.g., $\hbdel \p \q {}$);
also, choices must not appear between
corresponding forward and backward delegation.




We define \emph{projection} and \emph{localiser}
for hybrid types with delegation, by extending
Definitions \ref{def:proj} and \ref{def:loc},
\S\ref{subsec:hybrid}.
The main differences are listed below.
\begin{itemize}
\item Projection behaves on \emph{delegation constructs}
following the same intuition as for messages. E.g.,
$\projb {(\hfdel \p \q \hT)} {E} = \hfdel \p \q {(\projb\hT E)}$
if $\{\p,\q\}\subseteq E$, or
$\projb {(\hfdel \p \q \hT)} {E} = \hafdel \p \q {(\projd \hT {\{(\p,\q)\}} E)}$
if $\p\in E$ and $\q\notin E$,
where projection needs to keep track of $\p$ delegating to $\q$,
as indicated by the notation
$\projd \hT {\{(\p,\q)\}} E$.
\item Branches with explicit connections of a participant $\ps$ can be merged with branches where $\ps$ does not appear; e.g., if $\dht{\hT= 
\boxplus
\{
\dm \p \q(\lblFmt{\ell});\hend,}\allowbreak
\dht{\dme \p \ps(\lblFmt{\ell^e});\dm\ps\q(\lblFmt{\ell'});\hend
\}}
$, then
$\projb{\hT} {\{\ps\}}= \dht{\dre \p\ps(\lblFmt{\ell^e})
;\ds\ps\q(\lblFmt{\ell'});\hend}$.
\item We allow merging of receive constructs with different senders; e.g., with $\hT$ as in the
bullet point above, $\projb{\hT} {\{\q\}}= \dht{\bigwedge\{
\dre  \p \q(\lblFmt{\ell});\hend, \dre \ps\q(\lblFmt{\ell'});\hend\}}$
\item As in our core theory \S\ref{subsec:hybrid}, the localiser skips global constructs and retains local ones; e.g., $\loc{}{\hfdel\p\q\hT}=\loc{}\hT$ and $\loc{}{\hafdel\p\q\hT}=\hafdel\p\q{(\loc{}\hT)}$.
\end{itemize}
The partial function
$\projd \_ \de E$
generalises the delegation projection functions
$\projdone \_ \p \q$ and $\projdtwo\_\p\q$ from \cite{CASTELLANI2020128}
(Figure $6$); this function
is of a sequential nature: it is not defined
for branching, recursion, and parallel constructs.
Full definitions of projection $\projb{}{}$,
the auxiliary delegation projection $\upharpoonright^{d}$,
and localiser $\loc{}{}\!\!$ 
are in Appendix \appref{sec:del-proj-a}\inApp\
and compatible with 
\cite{CASTELLANI2020128}.

\begin{example}[Projection with Delegation and Explicit Connections]
\label{ex:d2}
Let us consider $\Gd$ and $\Gec$ from Example \ref{ex:d1}
and, in particular, their projections onto single participants

The intuition behind the delegation constructs in $\Gd$
is well displayed by its projections.
\[
\begin{small}
\begin{array}{l}
\projb \Gd {\{\seller\}} =
\dht{\ds \seller\buyer(\ok);
\hafdel\seller\bank
{\hpbdel\bank\seller\hend}}\\
\projb \Gd {\{\buyer\}} =
\dht{\dr \seller\buyer(\ok);\ds \buyer\seller(\card);\hend}\\
\projb \Gd {\{\bank\}} =
\dht{
\hpfdel\seller\bank{\dr\buyer\bank}(\card);
\habdel\bank\seller\hend}
\end{array}
\end{small}
\]
After \emph{actively}
delegating their role to the $\bank$
($\dht{\seller\fdel[\bank]}$)
and before \emph{passively}
being delegated their role back
($\dht{[\bank]\bdel\seller}$),
the $\seller$ is not involved
in any communication ($\projb \Gd {\seller}$).
At the same time, the $\bank$
plays the opposite role: is passive
in receiving the delegation from
$\seller$ ($\dht{[\seller]\fdel \bank}$)
and active in delegating the role back
($\bank\bdel[\seller]$).
Furthermore, we observe that,
while the $\bank$
knows that they are receiving
the $\card$ number from the $\buyer$,
the $\buyer$ acts as if they are
sending it directly to the
$\seller$:
as expected, the $\buyer$
is not involved in
the delegation process
and hence ignores it.

When projecting $\Gec$ onto its participants,
we obtain the following types.
\[
\begin{array}{l}
\projb \Gec {\{\user\}} =
\dht{
\ds \user\website (\location);
\bigwedge
\left\{
\begin{array}{l}
\dr\shopeu\user (\fronteu);\hend
\\
\dr \shopuk\user (\frontuk);\hend
\end{array}
\right\}
}\\
\projb \Gec {\{\website\}}=
\dht{
\dr \user\website (\location);
\boxplus
\left\{
\begin{array}{l}
\dse \website\shopeu (\open);
\hend
\\
\dse\website\shopuk (\open);
\hend
\end{array}
\right\}
}
\\
\projb \Gec {\{\shopeu\}}=
\dht{
\dre \website\shopeu (\open);
\ds \shopeu\user (\fronteu);\hend
}
\\
\projb \Gec {\{\shopuk\}}=
\dht{
\dre \website\shopuk (\open);
\ds \shopuk\user (\frontuk);\hend
}
\end{array}
\]
We observe that
the participants $\shopeu$
and $\shopuk$ are concerned
only with the interactions
that happen after their
explicit connection
($\dre \website\shopeu$ and
$\dre \website\shopuk$ respectively)
and not with the communication
in the other branch, where
they are not connected.
\end{example}


We now define \emph{compatibility}.
In the usual multi-component scenario,
we give hybrid types
(with delegation and explicit connections)
$\hTone,\hTtwo,\dots,\hTN$, for each component,
and a compatibility type
$\gG$. 
As in \S\ref{sec:comp},
we consider the generic $\hTE\in\{\hTone,\hTtwo,\dots,\hTN\}$
(with $E$ being the set of its internal participants)
in isolation
and we state compatibility: 
\begin{equation}
\projb {\gG} {E} = \loc{}\hTE
\tag{\textbf{D}}\label{eq:compatibility-d}
\end{equation}
Equation \ref{eq:compatibility-d} is exactly the same
as Equation \ref{eq:compatibility}. This captures the 
\emph{generality} of our technique:
given an existing top-down MPST system, 
first, we extend its syntax to hybrid types and
generalise projection to sets of participants;
then we isolate the inter-component communication of each protocol
with the localiser; and, finally, we can state compatibility,
prove compositionality, and achieve distributed protocol
specification.
Such general design 
gives a clear advantage to our theory, with respect
to the dual approach from 
previous work \cite{BLTDezani,smdg:facs21}.

\begin{example}[Compatibility Through Projection VS Input/Output Matching]
\label{ex:comp-io}\label{ex:del}
We consider the following three-component system $\hTone,\hTtwo,\hTthree$,
with $\gG$ for compatibility.
\begin{small}
\[
\hspace*{-1mm}\begin{array}{ll}
\gG= \dht{\boxplus}
\dht{\left\{
\dm \p \q(\lblFmt{\ell});\hend,
\dme \p \ps(\lblFmt{\ell^e});\dm\ps\q(\lblFmt{\ell'});\hend
\right\}}
&\hspace*{-1mm}
\hTtwo=\dht{\dm\q\qone(\lblFmt{\ell_1});\bigwedge
\left\{
\dr \p \q(\lblFmt{\ell});\hend,
\dr \ps \q(\lblFmt{\ell'});\hend
\right\}}
\\
\hTone=\dht{\boxplus}
\dht{\left\{
\dm \p \pzero (\lblFmt{\ell_0}); \ds \p \q(\lblFmt{\ell});\hend,
\dse \p \ps(\lblFmt{\ell^e});\dm \p \pzero (\lblFmt{\ell_0});\hend
\right\}}
&\hspace*{-1mm}
\hTthree=\dre \p \ps(\lblFmt{\ell^e});\dme \ps\pr(\lblFmt{\ell^e_0});
\ds\ps\q(\lblFmt{\ell'});\hend
\end{array}
\]
\end{small}
In $\hTone$, $\p$ makes a choice: on the first branch, first it sends
an internal message and then an external one to $\q$ in $\hTtwo$;
on the other branch $\p$ first makes an explicit connection to $\ps$
in $\hTthree$ and then sends an internal message.
The component
$\hTtwo$ (after an internal interaction) is waiting to
receive either from $\p$ in $\hTone$ or from $\ps$ in $\hTtwo$.
The third component $\hTthree$ is concerned only with
the second branch, in case it is chosen by $\p$.
All types
are compatible with respect to $\gG$:
$\projb \gG {E_i}=\loc{}\hTi$, with $E_1=\{\p,\pzero\}$, $E_2=\{\q,\qone\}$,
and $E_3=\{\ps,\pr\}$. Indeed, we can build a well-formed
global type for the whole system:
\begin{small}
\[
\G=\dm\q\qone(\lblFmt{\ell_1});\boxplus
\dht{\left\{
\dm \p \pzero (\lblFmt{\ell_0});\dm \p \q(\lblFmt{\ell});\hend,
\dme \p \ps(\lblFmt{\ell^e});\dme \ps\pr(\lblFmt{\ell^e_0});\dm\ps\q(\lblFmt{\ell'});\dm \p \pzero (\lblFmt{\ell_0});\hend
\right\}}
\]
\end{small}
We focus on $\hTtwo$ and we observe that,
in the intersection of inputs, $\q$
is waiting to receive from $\p$ in $\hTone$,
on the first branch, and from $\ps$ in $\hTthree$,
on the second.
Therefore, a match for inputs in $\hTtwo$ cannot
happen solely with outputs of $\hTone$, nor
solely with outputs of $\hTthree$.
In this case, dual compatibility relations
would fail, while, with our approach based
on projection, compatibility
can be simply stated with respect to the
global guidance of $\gG$, through
Equation \ref{eq:compatibility-d}.
\end{example}


\section{Related Work}
\label{sec:related}
Our work achieves \emph{protocol compositionality}
in MPST top-down systems, namely those
systems \emph{where the communication protocol
is explicitly described} as a global type
and, subsequently, from the projection of it,
local types are obtained for implementation.
We organise this section as a progressive discussion
of related work, with respect to our paper,
from more distant to closer.


\myparagraph{MPST Alternatives to the Top-Down Approach.}
Since their first appearance \cite{Honda2008Multiparty},
MPST have evolved into a variety of
frameworks
for the specification and
the verification of concurrency,
often beyond the original
top-down approach.
The works of
\citet{10.1007/978-3-642-32940-1_17}
and \citet{DenielouYoshida2013}
explicitly give algorithms
for the synthesis
of global types from 
communicating finite-state machines, while
\citet{LTY2015} propose
a similar method to build
graphical choreographies,
expressed as global graphs.
\citet{SY2019} 
develop a framework
where global types are not necessary,
relying instead on
model- and type-checking techniques
for verifying
safety properties of
collections of local types.
The advantage of such approaches 
is that they offer analysis
for pre-existent systems.
However, before the need to
hierarchically design
a new communicating system,
the top-down approach
enables a
high-level specification of the system
that guarantees safe interactions
of distributed implementations.
The top-down approach
has seen a variety of tools and implementations, e.g.,
\cite{10.1007/978-3-642-19056-8_4,YHNN13,FeatherweightScribble,nuscr,gheri_et_al:LIPIcs.ECOOP.2022.8,zooid,CFMPITP:2021,10.1145/2429069.2429101}
and it has been investigated beyond MPST
\cite{BLT-CA,chor-prog}.
Recent research
(e.g., \cite{GHH2021,lagaillardie_et_al:LIPIcs.ECOOP.2022.4,cledou_et_al:LIPIcs.ECOOP.2022.27,10.1145/3547638})
has kept exploring the
possibilities offered
by an explicit design of systems
through protocol specification.
Our work adds
compositionality
to top-down protocol
specification.


\myparagraph{Protocol Flexibility and Modularity beyond MPST.}
\citet{10.1007/978-3-642-00590-9_21,10.1007/978-3-662-43376-8_10}
do not
address protocol compositionality,
but, in defining conversation types,
the authors combine global and local constructs
in the same syntax, for
a flexible specification: 
messages can be scheduled,
while participants can be at first
left unspecified, thus
allowing for interleaving of sessions.
In the context of reactive programming,
\citet{CMTV:2018,SGTV:2020}
propose a technique for modular design:
a communicating system is specified in terms of components,
each (\emph{composite}) component contains a
choreography (or, \emph{protocol}) and for each role
in the protocol a new implementation (component) is specified.
Intuitively speaking, each component comes with
an input/output interface allowed by the protocol,
so to keep track of data-flow dependencies.
\citet{compchor} develop a compositional
technique for choreographic
programming: the specification of partial choreographies
is allowed, i.e., the implementation of some of the
roles 
can be left unspecified.
These roles can be implemented by a different choreography
at a later time; 
\emph{compatibility} is achieved through an additional typing
relation on choreographies,
which relies both on global and local types.
All the above work modifies the essence of
the protocol structure. Our theory differs
from it, first, because
our compatibility condition
(Equation \ref{eq:compatibility})
relies on projection,
instead of dual input/output matching.
Then, compositionality based on
hybrid types retains the
simplicity of traditional MPST
protocol structure:
the result of composition
is a well-formed global type
and no new semantics need to be developed.
Thanks to such semantics preservation 
and our novel compatibility,
our techniques are general enough
to be applied to traditional MPST
(\S\ref{sec:hybrid} and \S\ref{sec:comp}),
as well as to 
more expressive specifications (\S\ref{subsec:del}).

\myparagraph{Modular Global Types through Nesting.}
Global types are choreographic objects
and, hence, originally \cite{Honda2008Multiparty}
intended as 
standalone entities.
Given their monolithic nature, techniques
for making them more flexible have been
proposed very early on; e.g., \citet{nested} describe 
a methodology for nesting global types,
via calls to a subprotocol from a parent protocol.
\citet{aspects} propose an alternative
approach to nesting protocols,
by extending the syntax of
global and local types with
\emph{aspects} \cite{AOP}. 
\citet{DBLP:journals/fmsd/DemangeonHHNY15}
also explore nesting techniques for global types,
and apply these to extend the Scribble
protocol description language
\cite{10.1007/978-3-642-19056-8_4,FeatherweightScribble,nuscr}
with interruptible interactions.
Our work does not establish a parent/offspring
relation, 
but it explores 
direct composition of
subprotocols 
treated as peers.

\myparagraph{State of The Art of Direct Composition of MPST Protocols.}
To the best of
our knowledge, only \citet{BLTDezani} and \citet{smdg:facs21}
investigate the direct composition 
of protocols
for MPST, so that
\emph{the result of composition is a well-formed global type}.
Our development differs from this work, in primis, because of its
semantics preservation: 
in our theory, global types are exactly as in
traditional MPST theories (e.g., \cite{DenielouYoshida2013});
thus, our compositionality is
immediately compatible with those
and benefits from their semantics results.
Furthermore, our theory
can compose \emph{more than two} protocols
(missing in \cite{BLTDezani}) and captures the full
expressiveness of MPST, including \emph{parallel composition
and recursion} (missing in \cite{smdg:facs21}).
Ultimately, our compatibility 
(Equation \ref{eq:compatibility})
\emph{overcomes the limitations of dual input/output matching},
on which both \citet{BLTDezani} and \citet{smdg:facs21} rely:
our theory is general enough
to be applied to more sophisticated
MPST systems
(\S\ref{subsec:del}, Example \ref{ex:comp-io}).
Below we expand on the significance of our contribution,
with respect to this related work, with detailed examples.

In \cite{BLTDezani}, the composition of
\emph{two} global types
is achieved with
\emph{gateways}
\cite{gateways1,gateways2,BLT-isola2020}: two participants,
one for each subprotocol (global type), are selected as
forwarders (gateways) for communicating with the other;
if the subprotocols
are compatible, with respect to the
gateway choice, they can be composed
into a more general global type.
The central difference between our design and \cite{BLTDezani}
is the ``interface'' through which subprotocols
communicate. 
%
Instead of gateways, hybrid types use local
constructs for inter-component interactions and,
instead of a dual input/output matching, they
rely on 
projection for compatibility.
Thus, we
can safely combine two or more protocols at once.
%
Concretely, we consider the
following example involving three subprotocols,
described
as global types with gateways.
\\
\begin{small}
\centerline{$
\begin{array}{c}
  \dgt{\G_1} = \msg {\p} {\phtwo} \dgt{\elllbl. \msg {\p} \q \elllbl. \msg {\phthree}\q\elllbl.}\gend
  \qquad
   \dgt{\G_2} = \msg {\pkone} {\pr} \dgt{\elllbl. \msg {\ps} \pr \elllbl. \msg {\pr} {\pkthree}\elllbl.}\gend\\
   \dgt{\G_3} = \msg {\pt} {\pu} \dgt{\elllbl. \msg {\pu} {\plone} \elllbl. \msg {\pltwo}\pu\elllbl.}\gend
\end{array}
$}\\
\end{small}
We identify $(\phtwo,\pkone)$, $(\phthree,\plone)$,
and $(\pkthree,\pltwo)$ as gateways;
we choose the first one and compose.\\
\begin{small}
\centerline{$
\begin{array}{l}
 \dgt{\G_1}^{\phtwo\leftrightarrow\pkone} \dgt{\G_2} = \msg {\p} {\phtwo} \dgt{
  \elllbl.\msg{\phtwo}{\pkone}\elllbl.\msg {\pkone} {\pr}\elllbl.
\msg {\p} \q \elllbl. \msg {\phthree}\q \elllbl.\msg {\ps} \pr \elllbl. \msg {\pr} {\pkthree}\elllbl.\gend}
\end{array}
$}
\end{small}
The next steps of composition should happen between $\phthree$ and $\plone$,
and between $\pkthree$ and $\pltwo$.
Applying again gateway-composition does not work:
after a first step (e.g., using $\phthree$ and $\plone$) we would
be left with a single global type, i.e.,
no two types left to compose
(e.g., using $\pkthree$ and $\pltwo$). While we
could try 
simultaneous composition for more
than one pair of gateways at once, the simple
example in \cite{BLTDezani}, Section $11$, shows how this extension is
unsound and could lead to deadlocked systems.
Another 
na\"ive attempt to improve on gateway
composition is the following:
we force $\dgt{\G_1}$ and $\dgt{\G_2}$
to have the same participant for communicating with $\dgt{\G_3}$
(
relaxing the condition from \cite{BLTDezani},
where participants in two composing global types must be distinct):\\
\begin{small}
\centerline{$
\begin{array}{c}
   \dgt{\G'_1} = \msg {\p} {\ph} \dgt{\elllbl. \msg {\pl} \p\elllbl.}\gend
   \qquad 
   \dgt{\G'_2} = \msg {\pk} {\q} \dgt{\elllbl. \msg {\q} \pl \elllbl.}\gend\\
   \dgt{\G'_3} = \msg {\pt} {\pr} \dgt{\elllbl. \msg {\pr} \pt \elllbl.}\gend
\end{array}
$}\\
\end{small}
The communication is supposed to happen between $\ph$ and $\pk$,
and between $\pl$ and $\pt$. If we allow non-distinct participants
in distinct global types, however, it is not immediate
to achieve a well-defined gateways compositionality.
For example, the order in which we compose plays a key role:
while $(\dgt{\G_2}^{\pk\leftrightarrow\ph}\dgt{\G_1})^{\pl\leftrightarrow\pt}\dgt{\G_3}$
is well defined,
$(\dgt{\G_1}^{\ph\leftrightarrow\pk}\dgt{\G_2})^{\pl\leftrightarrow\pt}\dgt{\G_3}$
is not. Also, associativity would not hold:
we cannot compose $\dgt{\G_3}$ with any
of $\dgt{\G_1}$ and $\dgt{\G_2}$ and then finish
composing with the remaining one.

We specify the system above in terms of hybrid types:
in place of gateways,
we use local constructs for inter-component
communication (highlighted),
and we give the type $\gG$ for compatibility.
\\
\begin{small}
\centerline{$
\begin{array}{c}
  \hTone = \hla{\hsnd {\p} {\pr}} \dht{\elllbl. \hmsg {\p} \q \elllbl. \hlb{\hrcv {\pu} \q}}\dht{\elllbl.\hend}
  \qquad 
  \hTtwo = \hla{\hrcv {\p} {\pr}} \dht{\elllbl. \hmsg {\ps} \pr \elllbl. \hlc{\hsnd {\pr} {\pu}}}\dht{\elllbl.\hend}\\
  \hTthree = \hmsg {\pt} {\pu} \dht{\elllbl. \hlb{\hsnd {\pu} {\q}} \elllbl. \hlc{\hrcv \pr {\pu}}}\dht{\elllbl.\hend}
  \qquad \gG = \hla{\hmsg {\p} {\pr}} \dht{\elllbl. \hlb{\hmsg {\pu} \q} \elllbl. \hlc{\hmsg {\pr} \pu}}\dht{\elllbl.\hend}
\end{array}
$}
\end{small}
For all components, compatibility \ref{eq:compatibility} holds
and our theory can compose the three protocols 
into a well-formed global type
(Corollary \ref{cor:dist-spec}),  
thus overcoming the \emph{binary} limitation of gateways.


\citet{smdg:facs21}, develop
a binary compatibility
relation that
partially improves on
\cite{BLTDezani}: 
they can compose
more than two 
types,
but put severe restrictions on the syntax.
Global types in \cite{smdg:facs21}
are inductive, but
with \emph{no recursion and no parallel construct},
thus making our work (and traditional MPST in general)
strictly more expressive.
The key role of parallel composition
for an inductive syntax
is discussed in \cite{GHH2021}
(Example 13 and following paragraph)
and in \S\ref{sec:parallel}.
Recursion is paramount for expressiveness in MPST
\cite{Honda2008Multiparty,HYC2016,CDPY2015,verygentle}:
it allows typing processes with loops
and it is
omnipresent in real-world protocols
(e.g., see \S\ref{subsec:oauth}
and \S\ref{subsec:oauth2}).
As for semantics,
recursion allows for infinite executions and is
a key element in the correspondence between local
types and communicating finite-state machines
\cite{DenielouYoshida2013}, on which
practical implementations
of the popular Scribble protocol language
are based
\cite{10.1007/978-3-642-19056-8_4,YHNN13,
HY2016,FeatherweightScribble,nuscr}.
Also, recursion (in combination with branching)
is among the most delicate aspects of MPST
(see, e.g., Observation 3 in \cite{GHH2021}
or Definition \ref{def:guarded-a}
and the following well-formedness lemmas
in Appendix \ref{subsec:appendix1}\inApp).
E.g., the recursive types
$\hTone=\dht{\mu\hX.}\hmsg\p\q\dht{\lblFmt{\ell_1}.\hsnd\p\pr\lblFmt{\ell}.\hX}$ and
$\hTtwo=\dht{\mu\hX.}\hrcv\p\pr\dht{\lblFmt{\ell}.\hmsg\ps\pr\lblFmt{\ell_2}.\hX}$,
are composable both with gateways \cite{BLTDezani}
and with our theory
($\gG=\hrec \hX {\hmsg\p\pr\lblFmt{\ell}.\hX}$
for compatibility), but they cannot be expressed 
in \cite{smdg:facs21}.

In summary, our theory allows composing \emph{two or more}
subprotocols 
into a well-formed global type, while
retaining \emph{full expressiveness} of MPST, thus improving
on the state of the art \cite{BLTDezani, smdg:facs21}.
Furthermore, previous work
focuses on dual input/output matching,
which makes it inapplicable to more expressive MPST systems,
where instead, to the best of our knowledge,
our novel compatibility \ref{eq:compatibility} is the first notion
to succeed (see \S\ref{subsec:del}, Example \ref{ex:comp-io}).
%
We observe that, as it happens
for the extension of
binary session types to multiparty session types,
when relying on projection instead of duality,
we commit to the specification of one more global
object (what we called $\gG$), but, in return,
we obtain full multiparty compatibility.




\section{Future Work}
\label{sec:future}

The first envisioned application for
our compositionality theory
is its integration with
practical protocol design
languages, such as
Zooid \cite{zooid} or Scribble
\cite{10.1007/978-3-642-19056-8_4,YHNN13,nuscr}.
In what follows we
briefly describe
how this integration
can be realised.

In the past ten years,
Scribble has been employed as
the protocol language
of multiple toolchains,
supporting different
programming languages and
integrating a variety of
expressive features
\cite{10.1007/978-3-642-19056-8_4,YHNN13,FeatherweightScribble,nuscr,HY2017,SDHY2017,NHYA2018,CHJNY19,MFYZ21,ZFHNY2020,gheri_et_al:LIPIcs.ECOOP.2022.8}.
Independently of the
specific implementation,
the Scribble toolchain is generally
designed as follows.
\begin{enumerate}
\item The designer specifies the communication
protocol (a global type $\G$) in Scribble.
\item $\G$ is projected onto local types
$\dlt{\lT_1},\dots,\dlt{\lT_n}$---or,
equivalently, their representation
as CFSMs \cite{DenielouYoshida2013}.
\item From local types, APIs for
the distributed implementation
of all participants
are generated
(following the approach of \citet{HY2016}).
\end{enumerate}
Through API implementation,
the communication
for the multiparty
system is MPST certified
and does not get stuck
(semantic guarantees hold).
Thanks to its semantics-preserving
features and its backwards compatibility with
existing MPST systems,
our compositionality framework can be
integrated with Scribble toolchains,
with minimal effort. First,
we will extend
the Scribble protocol design language
to the syntax of hybrid types
(Definition \ref{def:hybrid-types})---of which
global types are 
a particular case---with
constructs $\hsnd\p\q$ and $\hrcv\p\q$.
Then, we will implement the localiser
(Definition \ref{def:loc}) and extend
projection to Definition \ref{def:proj};
here compatibility with previous
implementations is guaranteed by the
considerations in Remark \ref{rem:proj}:
Definition \ref{def:proj} generalises
traditional MPST projection
(Definition  \ref{def:proj-old}).
Finally, we will implement checks
for $\isglobal\gG$ and compatibility
\ref{eq:compatibility}
($\projb {\gG} {E} = \loc{}\hTE$).
With these simple changes,
we 
will
endow Scribble
with compositionality, thus enabling
distributed protocol specification.
``Compositional Scribble''
will look as follows.
\begin{itemize}
\item[$(1^{\text{c}})$] Multiple designers give $\hTone,\dots,\hTN$
in a distributed fashion and a chief designer specifies
$\gG$, as \emph{hybrid types}.
\item[$(2^{\text{c}})$] Internal checks are performed for
$\isglobal\gG$ and compatibility \ref{eq:compatibility}.
\item[$(3^{\text{c}})$] Each $\hTi$ is projected onto local types
$\dlt{\lT^i_1},\dots,\dlt{\lT^i_{n_i}}$---or,
equivalently, their representation
as CFSMs \cite{DenielouYoshida2013}.
\item[$(4^{\text{c}})$] From local types, APIs for
the distributed implementation
of all participants
are generated.
\end{itemize}
We observe that, thanks to our theory, the
checks in $(2^{\text{c}})$ are
enough to guarantee the existence of a well-formed global
type for the whole system,
without the need for explicitly building back such type
(see also Corollary \ref{cor:dist-spec}
and Remark \ref{rem:app}).
Moreover, since projection is preserved
(Theorem \ref{thm:proj-comp} and Corollary \ref{cor:dist-spec}),
local types are the same as in traditional MPST:
not only do semantic guarantees still hold,
but so does the correspondence between local
types and CFSMs in \cite{DenielouYoshida2013}.
%
%

With respect to traditional Scribble,
nothing changes for the user, apart from
the added functionality of distributed
protocol specification: the safe distributed
implementation of all participants
is still enabled, but, now, also protocols
can be specified in a distributed fashion,
in terms of their
components (as hybrid types $\hTone,\dots,\hTN$).


Beyond its integration with existing
protocol design languages, we
plan to build on this work in
different directions.
In Section \ref{subsec:del},
we have shown the potential
of our approach to compositionality, by adapting it
to an MPST formalism that extends
the traditional syntax of global types
to include the advanced features
of delegation and explicit connections.
Similarly, beyond our
core compositionality for
global types,
we envision
future applications
to the wide
variety of MPST systems
(e.g., featuring fault tolerance \cite{10.1145/3485501},
timed specification \cite{BYY2014},
refinements \cite{ZFHNY2020},
cost awareness \cite{10.1145/3428223}, or
exception handling \cite{lagaillardie_et_al:LIPIcs.ECOOP.2022.4})
and, orthogonally, future extensions that 
add flexibility to compositionality
itself (e.g., by factoring in
renaming mechanisms for participants \cite{10.1145/3547638}).
Another promising
perspective is the application of our 
techniques beyond MPST, to other
protocol-design formalisms
based on projection, e.g.,
\emph{choreography automata} \cite{BLT-CA,gheri_et_al:LIPIcs.ECOOP.2022.8} or
\emph{choreographic programming} \cite{chor-prog}.



\section{Conclusion}
\label{sec:conclusion}


We have developed a
``compositionality-through-projection''
technique
that allows the
distributed specification
of MPST protocols,
in terms of \emph{hybrid types}
(Definition \ref{def:hybrid-types}).
Our work neatly
improves on the state of the art,
by allowing for composition
of \emph{more than two protocols},
while retaining the \emph{full
expressiveness} of global types.
Our results
(Corollaries \ref{cor:dist-spec} and \ref{cor:dist-spec2})
guarantee correctness and make
our theory compatible
with existing MPST systems (\emph{semantics preservation}). 
Our \emph{novel compatibility} relation
(Equation \ref{eq:compatibility}), based on
generalised projection
and localiser
(Definitions \ref{def:proj} and \ref{def:loc}),
overcomes the limitations of dual input/output matching
and it is general enough
to capture extensions beyond traditional MPST
(e.g., to delegation and explicit connections
\S\ref{subsec:del}).




\begin{acks}
  We thank
  Franco Barbanera, Mariangiola Dezani-Ciancaglini,
  Francisco Ferreira, and Franco Raimondi for the
  in-depth conversations and useful comments on
  the preliminary versions of the paper.
  This work is
  supported by UKRI/EPSRC, references:
  EP/T006544/2, EP/K011715/1,
  EP/K034413/1, EP/L00058X/1, EP/N027833/2,
  EP/N028201/1, EP/T014709/2, EP/V000462/1,
  EP/X015955/1, and NCSC/EPSRC VeTSS,
  and
  EU HORIZON EUROPE
  Research and Innovation Programme,
  grant agreement 101093006 (TaRDIS).
\end{acks}

\bibliography{fullversion}

@InProceedings{Bettini2008Global,
  author    = {Bettini, Lorenzo
               and Coppo, Mario
               and D'Antoni, Loris
               and De Luca, Marco
               and Dezani-Ciancaglini, Mariangiola
               and Yoshida, Nobuko},
  editor    = {van Breugel, Franck and Chechik, Marsha},
  title     = {Global Progress in Dynamically Interleaved Multiparty Sessions},
  booktitle = {CONCUR 2008 - Concurrency Theory},
  year      = {2008},
  publisher = {Springer},
  address   = {Berlin, Heidelberg},
  pages     = {418--433},
  doi = {https://doi.org/10.1007/978-3-540-85361-9_33},
  isbn      = {978-3-540-85361-9}
}

@inproceedings{CDPY2015,
  volume    = 9104,
  pages     = "146--178",
  publisher = "Springer",
  year      = 2015,
  author    = "Mario Coppo and Mariangiola Dezani-Ciancaglini and Luca Padovani and Nobuko Yoshida",
  title     = {{A Gentle Introduction to Multiparty Asynchronous Session Types}},
  booktitle = "15th International School on Formal Methods for the Design of Computer, Communication and Software Systems: Multicore Programming",
  series    = "LNCS",
  doi = {10.1007/978-3-319-18941-3_4}
}

@InProceedings{Honda:1998,
  author={Kohei Honda and Vasco T. Vasconcelos and Makoto Kubo},
  editor={Hankin, Chris},
  title={Language primitives and type discipline for structured communication-based programming},
  booktitle={Programming Languages and Systems},
  year={1998},
  doi = {10.1007/BFb0053567},
  publisher={Springer Berlin Heidelberg},
  address={Berlin, Heidelberg},
  pages={122--138},
}

@inproceedings{Honda2008Multiparty,
  author    = {Honda, Kohei and Yoshida, Nobuko and Carbone, Marco},
  title     = {Multiparty Asynchronous Session Types},
  booktitle = {Proc. of 35th Symp. on Princ. of Prog. Lang.},
  series    = {POPL '08},
  year      = {2008},
  location  = {San Francisco, California, USA},
  pages     = {273--284},
  numpages  = {12},
  doi = {10.1145/1328897.1328472},
  publisher = {ACM},
  address   = {New York, NY, USA}
}

@string{acm = {ACM Press}}

@inproceedings{YHNN13,
  volume    = 8358,
  series    = "Lecture Notes in Computer Science",
  author    = "Nobuko Yoshida and
               Raymond Hu and
               Rumyana Neykova and
               Nicholas Ng",
  publisher = "Springer",
  editor    = {Mart{\'{\i}}n Abadi and
               Alberto Lluch{-}Lafuente},
  booktitle = {Trustworthy Global Computing - 8th International Symposium, {TGC}
               2013, Buenos Aires, Argentina},
  pages     = "22--41",
  title     = "The Scribble Protocol Language",
  year      = "2013",
  doi = "https://doi.org/10.1007/978-3-319-05119-2_3",
}

@article{HYC2016,
  url       = "http://doi.acm.org/10.1145/2827695",
  doi       = "10.1145/2827695",
  volume    = 63,
  pages     = "9:1--9:67",
  author    = "Kohei Honda and
               Nobuko Yoshida and
               Marco Carbone",
  title     = "Multiparty Asynchronous Session Types",
  journal   = {J. {ACM}},
  year      = 2016,
  bibsource = "dblp computer science bibliography, http://dblp.org",
  number    = 1,
}

@inproceedings{Scalas:2019,
 author = {Alceste Scalas and Nobuko Yoshida and Elias Benussi},
 title = {Verifying Message-passing Programs with Dependent Behavioural Types},
 booktitle = {Proceedings of the 40th ACM SIGPLAN Conference on Programming Language Design and Implementation},
 series = {PLDI 2019},
 year = {2019},
 isbn = {978-1-4503-6712-7},
 location = {Phoenix, AZ, USA},
 pages = {502--516},
 numpages = {15},
 doi = {10.1145/3314221.3322484},
 acmid = {3322484},
 publisher = {ACM},
 address = {New York, NY, USA},
}

@InProceedings{HU07TYPE-SAFE,
author="Hu, Raymond
and Yoshida, Nobuko
and Honda, Kohei",
editor="Vitek, Jan",
title="Session-Based Distributed Programming in Java",
booktitle="ECOOP 2008 -- Object-Oriented Programming",
year="2008",
publisher="Springer Berlin Heidelberg",
address="Berlin, Heidelberg",
doi = "https://doi.org/10.1007/978-3-540-70592-5_22",
pages="516--541"
}

@inproceedings{HY2017,
  author = {Raymond Hu and Nobuko Yoshida},
  title = {Explicit Connection Actions in Multiparty Session Types},
  booktitle = {FASE},
  series = {LNCS},
  volume = {10202},
  pages = {116--133},
  doi = "https://doi.org/10.1007/978-3-662-54494-5_7",
  year = 2017
}

@inproceedings{SDHY2017,
  author = {Alceste Scalas and Ornela Dardha and Raymond Hu and Nobuko Yoshida},
  title = {{A Linear Decomposition of Multiparty Sessions for Safe Distributed Programming}},
  booktitle = {{ECOOP}},
  optseries = {LIPIcs},
  optvolume = {74},
  optpages = {24:1--24:31},
  optpublisher = {Schloss Dagstuhl},
  doi = "10.4230/LIPIcs.ECOOP.2017.24",
  year = 2017
}

@inproceedings{HY2016,
  title     = {Hybrid Session Verification Through Endpoint {API} Generation},
  year      = 2016,
  volume    = 9633,
  editor    = "Perdita Stevens and
               Andrzej Wasowski",
  author    = "Raymond Hu and
               Nobuko Yoshida",
  booktitle = {Fundamental Approaches to Software Engineering - 19th International
               Conference, {FASE} 2016,Eindhoven, The Netherlands},
  series    = "Lecture Notes in Computer Science",
  pages     = "401--418",
  publisher = "Springer",
  doi       = "10.1007/978-3-662-49665-7_24",
}

@article{DBLP:journals/fmsd/DemangeonHHNY15,
  author    = {Romain Demangeon and
               Kohei Honda and
               Raymond Hu and
               Rumyana Neykova and
               Nobuko Yoshida},
  title     = {Practical interruptible conversations: distributed dynamic verification with multiparty session types and {P}ython},
  journal   = {FMSD},
  volume    = {46},
  number    = {3},
  pages     = {197--225},
  year      = {2015},
  doi       = {10.1007/s10703-014-0218-8},
  url       = {http://dx.doi.org/10.1007/s10703-014-0218-8},
  bibsource = {dblp computer science bibliography, http://dblp.org}
}

@inproceedings{BYY2014,
  author = {Laura Bocchi and Weizhen Yang and Nobuko Yoshida},
  title = {{Timed Multiparty Session Types}},
  booktitle = {25th International Conference on Concurrency Theory},
  series = {LNCS},
  volume = {8704},
  pages = {419--434},
  doi = {10.1007/978-3-662-44584-6_29},
  publisher = {Springer},
  year = 2014
}

@article{CHJNY19,
 author = {Castro, David and Hu, Raymond and Jongmans, Sung-Shik and Ng, Nicholas and Yoshida, Nobuko},
 title = {Distributed Programming Using Role-parametric Session Types in Go: Statically-typed Endpoint APIs for Dynamically-instantiated Communication Structures},
 journal = {Proc. ACM Program. Lang.},
 issue_date = {January 2019},
 volume = {3},
 number = {POPL},
 month = jan,
 year = {2019},
 issn = {2475-1421},
 pages = {29:1--29:30},
 articleno = {29},
 numpages = {30},
 url = {http://doi.acm.org/10.1145/3290342},
 doi = {10.1145/3290342},
 acmid = {3290342},
 publisher = {ACM},
 address = {New York, NY, USA},
}

@inproceedings{NHYA2018,
  author = {Rumyana Neykova and Raymond Hu and Nobuko Yoshida and Fahd Abdeljallal},
  title = {{A Session Type Provider: Compile-time API Generation for Distributed Protocols with Interaction Refinements in F\#}},
  booktitle = {27th International Conference on Compiler Construction},
  pages = {128--138},
  publisher = {ACM},
  doi = "10.1145/3178372.3179495",
  year = 2018
}

@InProceedings{DenielouYoshida2013,
author="Deni{\'e}lou, Pierre-Malo
and Yoshida, Nobuko",
editor="Fomin, Fedor V.
and Freivalds, R{\={u}}si{\c{n}}{\v{s}}
and Kwiatkowska, Marta
and Peleg, David",
title="Multiparty Compatibility in Communicating Automata: Characterisation and Synthesis of Global Session Types",
booktitle="Automata, Languages, and Programming",
year="2013",
doi = {10.1007/978-3-642-39212-2_18},
publisher="Springer Berlin Heidelberg",
address="Berlin, Heidelberg",
pages="174--186",
isbn="978-3-642-39212-2"
}

@inproceedings{SY2019,
  author = {Alceste Scalas and Nobuko Yoshida},
  title = {{Less Is More: Multiparty Session Types Revisited}},
  booktitle = {46th ACM SIGPLAN Symposium on Principles of Programming Languages},
  volume = {3},
  pages = {30:1--30:29},
  doi = {10.1145/3290343},
  publisher = {ACM},
  year = 2019
}

@Inbook{FeatherweightScribble,
  author="Neykova, Rumyana
  and Yoshida, Nobuko",
  editor="Boreale, Michele
  and Corradini, Flavio
  and Loreti, Michele
  and Pugliese, Rosario",
  title="Featherweight Scribble",
  bookTitle="Models, Languages, and Tools for Concurrent and Distributed Programming: Essays Dedicated to Rocco De Nicola on the Occasion of His 65th Birthday",
  series = {LNCS},
  volume = {11665},
  year="2019",
  publisher="Springer",
  address="Cham",
  pages="236--259",
  isbn="978-3-030-21485-2",
  doi="10.1007/978-3-030-21485-2_14",
}

@inproceedings{LTY2015,
  author = {Julien Lange and Emilio Tuosto and Nobuko Yoshida},
  title = {{From communicating machines to graphical choreographies}},
  booktitle = {42nd ACM SIGPLAN-SIGACT Symposium on Principles of Programming Languages},
  pages = {221--232},
  doi = {10.1145/2676726.2676964},
  publisher = {ACM},
  year = 2015
}

@InProceedings{10.1007/978-3-642-19056-8_4,
author="Honda, Kohei
and Mukhamedov, Aybek
and Brown, Gary
and Chen, Tzu-Chun
and Yoshida, Nobuko",
editor="Natarajan, Raja
and Ojo, Adegboyega",
title="Scribbling Interactions with a Formal Foundation",
booktitle="Distributed Computing and Internet Technology",
year="2011",
doi = {10.1007/978-3-642-19056-8_4},
publisher="Springer Berlin Heidelberg",
address="Berlin, Heidelberg",
pages="55--75",
isbn="978-3-642-19056-8"
}

@article{ZFHNY2020,
author = {Zhou, Fangyi and Ferreira, Francisco and Hu, Raymond and Neykova, Rumyana and Yoshida, Nobuko},
title = {Statically Verified Refinements for Multiparty Protocols},
year = {2020},
issue_date = {November 2020},
publisher = {Association for Computing Machinery},
address = {New York, NY, USA},
volume = {4},
number = {OOPSLA},
url = {https://doi.org/10.1145/3428216},
doi = {10.1145/3428216},
journal = {Proc. ACM Program. Lang.},
month = {nov},
articleno = {148},
numpages = {30}
}

@InProceedings{nested,
author="Demangeon, Romain
and Honda, Kohei",
editor="Koutny, Maciej
and Ulidowski, Irek",
title="Nested Protocols in Session Types",
booktitle="CONCUR 2012 -- Concurrency Theory",
year="2012",
publisher="Springer Berlin Heidelberg",
address="Berlin, Heidelberg",
pages="272--286",
doi = "https://doi.org/10.1007/978-3-642-32940-1_20",
isbn="978-3-642-32940-1"
}

@inproceedings{aspects,
author = {Tabareau, Nicolas and S\"{u}dholt, Mario and Tanter, \'{E}ric},
title = {Aspectual Session Types},
year = {2014},
isbn = {9781450327725},
publisher = {Association for Computing Machinery},
address = {New York, NY, USA},
url = {https://doi.org/10.1145/2577080.2577085},
doi = {10.1145/2577080.2577085},
booktitle = {Proceedings of the 13th International Conference on Modularity},
pages = {193–204},
numpages = {12},
location = {Lugano, Switzerland},
series = {MODULARITY '14}
}

@InProceedings{BLT-isola2020,
author="Barbanera, Franco
and Lanese, Ivan
and Tuosto, Emilio",
editor="Margaria, Tiziana
and Steffen, Bernhard",
title="Composing Communicating Systems, Synchronously",
booktitle="Leveraging Applications of Formal Methods, Verification and Validation: Verification Principles",
year="2020",
publisher="Springer International Publishing",
address="Cham",
pages="39--59",
doi = "https://doi.org/10.1007/978-3-030-61362-4_3",
isbn="978-3-030-61362-4"
}

@InProceedings{BLT-CA,
author="Barbanera, Franco
and Lanese, Ivan
and Tuosto, Emilio",
editor="Bliudze, Simon
and Bocchi, Laura",
title="Choreography Automata",
booktitle="Coordination Models and Languages",
year="2020",
publisher="Springer International Publishing",
address="Cham",
pages="86--106",
doi = "https://doi.org/10.1007/978-3-030-50029-0_6",
isbn="978-3-030-50029-0"
}

@article{BLTDezani,
title = {Composition and decomposition of multiparty sessions},
journal = {Journal of Logical and Algebraic Methods in Programming},
volume = {119},
pages = {100620},
year = {2021},
issn = {2352-2208},
doi = {https://doi.org/10.1016/j.jlamp.2020.100620},
url = {https://www.sciencedirect.com/science/article/pii/S235222082030105X},
author = {Franco Barbanera and Mariangiola Dezani-Ciancaglini and Ivan Lanese and Emilio Tuosto}
}

@inproceedings{gateways1,
  author    = {Franco Barbanera and
               Ugo de'Liguoro and
               Rolf Hennicker},
  editor    = {Massimo Bartoletti and
               Sophia Knight},
  title     = {Global Types for Open Systems},
  booktitle = {Proceedings 11th Interaction and Concurrency Experience, {ICE} 2018,
               Madrid, Spain, June 20-21, 2018},
  series    = {{EPTCS}},
  volume    = {279},
  pages     = {4--20},
  year      = {2018},
  url       = {https://doi.org/10.4204/EPTCS.279.4},
  doi       = {10.4204/EPTCS.279.4},
  bibsource = {dblp computer science bibliography, https://dblp.org}
}

@article{gateways2,
title = {Connecting open systems of communicating finite state machines},
journal = {Journal of Logical and Algebraic Methods in Programming},
volume = {109},
pages = {100476},
year = {2019},
issn = {2352-2208},
doi = {https://doi.org/10.1016/j.jlamp.2019.07.004},
url = {https://www.sciencedirect.com/science/article/pii/S2352220818301627},
author = {Franco Barbanera and Ugo de'Liguoro and Rolf Hennicker}
}

@inproceedings{zooid,
author = {Castro-Perez, David and Ferreira, Francisco and Gheri, Lorenzo and Yoshida, Nobuko},
title = {Zooid: A DSL for Certified Multiparty Computation: From Mechanised Metatheory to Certified Multiparty Processes},
year = {2021},
isbn = {9781450383912},
publisher = {Association for Computing Machinery},
address = {New York, NY, USA},
url = {https://doi.org/10.1145/3453483.3454041},
doi = {10.1145/3453483.3454041},
booktitle = {Proceedings of the 42nd ACM SIGPLAN International Conference on Programming Language Design and Implementation},
pages = {237–251},
numpages = {15},
location = {Virtual, Canada},
series = {PLDI 2021}
}

@InProceedings{CFMPITP:2021,
  author =	{Cruz-Filipe, Lu{\'\i}s and Montesi, Fabrizio and Peressotti, Marco},
  title =	{{Formalising a Turing-Complete Choreographic Language in Coq}},
  booktitle =	{12th International Conference on Interactive Theorem Proving (ITP 2021)},
  pages =	{15:1--15:18},
  series =	{Leibniz International Proceedings in Informatics (LIPIcs)},
  ISBN =	{978-3-95977-188-7},
  ISSN =	{1868-8969},
  year =	{2021},
  volume =	{193},
  editor =	{Cohen, Liron and Kaliszyk, Cezary},
  publisher =	{Schloss Dagstuhl -- Leibniz-Zentrum f{\"u}r Informatik},
  address =	{Dagstuhl, Germany},
  URL =		{https://drops.dagstuhl.de/opus/volltexte/2021/13910},
  URN =		{urn:nbn:de:0030-drops-139109},
  doi =		{10.4230/LIPIcs.ITP.2021.15}
}

@InProceedings{AOP,
author="Kiczales, Gregor
and Lamping, John
and Mendhekar, Anurag
and Maeda, Chris
and Lopes, Cristina
and Loingtier, Jean-Marc
and Irwin, John",
editor="Ak{\c{s}}it, Mehmet
and Matsuoka, Satoshi",
title="Aspect-oriented programming",
booktitle="ECOOP'97 --- Object-Oriented Programming",
year="1997",
publisher="Springer Berlin Heidelberg",
address="Berlin, Heidelberg",
pages="220--242",
doi = "https://doi.org/10.1007/BFb0053381",
isbn="978-3-540-69127-3"
}

@misc{CMTV:2018,
  doi = {10.48550/ARXIV.1801.08107},
  url = {https://arxiv.org/abs/1801.08107},
  author = {Carbone, Marco and Montesi, Fabrizio and Vieira, Hugo Torres},
  title = {Choreographies for Reactive Programming},
  publisher = {arXiv},
  year = {2018},
  copyright = {arXiv.org perpetual, non-exclusive license}
}

@inproceedings{SGTV:2020,
  author    = {Zorica Savanovic and
               Letterio Galletta and
               Hugo Torres Vieira},
  editor    = {Julien Lange and
               Anastasia Mavridou and
               Larisa Safina and
               Alceste Scalas},
  title     = {A type language for message passing component-based systems},
  booktitle = {Proceedings 13th Interaction and Concurrency Experience, {ICE} 2020,
               Online, 19 June 2020},
  series    = {{EPTCS}},
  volume    = {324},
  pages     = {3--24},
  year      = {2020},
  url       = {https://doi.org/10.4204/EPTCS.324.3},
  doi       = {10.4204/EPTCS.324.3},
  bibsource = {dblp computer science bibliography, https://dblp.org}
}

@InProceedings{compchor,
author="Montesi, Fabrizio
and Yoshida, Nobuko",
editor="D'Argenio, Pedro R.
and Melgratti, Hern{\'a}n",
title="Compositional Choreographies",
booktitle="CONCUR 2013 -- Concurrency Theory",
year="2013",
publisher="Springer Berlin Heidelberg",
address="Berlin, Heidelberg",
pages="425--439",
doi="https://doi.org/10.1007/978-3-642-40184-8_30",
isbn="978-3-642-40184-8"
}

@misc{oauth2,
	series =	{Request for Comments},
	number =	6749,
	howpublished =	{RFC 6749},
	publisher =	{RFC Editor},
	doi =		{10.17487/RFC6749},
	url =		{https://rfc-editor.org/rfc/rfc6749.txt},
        author =	{Dick Hardt},
	title =		{{The OAuth 2.0 Authorization Framework}},
	pagetotal =	76,
	year =		2012,
	month =		oct,
}

@InProceedings{horne:2020,
  author =	{Ross Horne},
  title =	{{Session Subtyping and Multiparty Compatibility Using Circular Sequents}},
  booktitle =	{31st International Conference on Concurrency Theory (CONCUR 2020)},
  pages =	{12:1--12:22},
  series =	{Leibniz International Proceedings in Informatics (LIPIcs)},
  ISBN =	{978-3-95977-160-3},
  ISSN =	{1868-8969},
  year =	{2020},
  volume =	{171},
  editor =	{Igor Konnov and Laura Kov{\'a}cs},
  publisher =	{Schloss Dagstuhl--Leibniz-Zentrum f{\"u}r Informatik},
  address =	{Dagstuhl, Germany},
  URL =		{https://drops.dagstuhl.de/opus/volltexte/2020/12824},
  URN =		{urn:nbn:de:0030-drops-128245},
  doi =		{10.4230/LIPIcs.CONCUR.2020.12}
}

@InProceedings{verygentle,
author="Yoshida, Nobuko
and Gheri, Lorenzo",
editor="Hung, Dang Van
and D{\textasciiacute}Souza, Meenakshi",
title="A Very Gentle Introduction to Multiparty Session Types",
booktitle="Distributed Computing and Internet Technology",
year="2020",
publisher="Springer International Publishing",
address="Cham",
pages="73--93",
doi="https://doi.org/10.1007/978-3-030-36987-3_5",
isbn="978-3-030-36987-3"
}

@InProceedings{nuscr,
author="Yoshida, Nobuko
and Zhou, Fangyi
and Ferreira, Francisco",
editor="Bampis, Evripidis
and Pagourtzis, Aris",
title="Communicating Finite State Machines and an Extensible Toolchain for Multiparty Session Types",
booktitle="Fundamentals of Computation Theory",
year="2021",
publisher="Springer International Publishing",
address="Cham",
pages="18--35",
doi="https://doi.org/10.1007/978-3-030-86593-1_2",
isbn="978-3-030-86593-1"
}

@phdthesis{chor-prog,
title = "Choreographic Programming",
author = "Fabrizio Montesi",
year = "2013",
language = "English",
url = "https://www.fabriziomontesi.com/files/choreographic_programming.pdf",
isbn = "978-87-7949-299-8"
}

@InProceedings{LangeYoshidaCAV19,
author="Lange, Julien
and Yoshida, Nobuko",
editor="Dillig, Isil
and Tasiran, Serdar",
title="Verifying Asynchronous Interactions via Communicating Session Automata",
booktitle="Computer Aided Verification",
year="2019",
publisher="Springer International Publishing",
address="Cham",
pages="97--117",
doi = "https://doi.org/10.1007/978-3-030-25540-4_6",
isbn="978-3-030-25540-4"
}

@INPROCEEDINGS{GHH2021,
  author={Glabbeek, Rob van and Höfner, Peter and Horne, Ross},
  booktitle={2021 36th Annual ACM/IEEE Symposium on Logic in Computer Science (LICS)},
  title={Assuming Just Enough Fairness to make Session Types Complete for Lock-freedom},
  year={2021},
  volume={},
  number={},
  pages={1-13},
  doi={10.1109/LICS52264.2021.9470531}}

@inproceedings{smdg:facs21,
Author = {Claude Stolze and Marino Miculan and Di Gianantonio, Pietro},
Booktitle = {FACS 2021 Conference Proceedings},
Publisher = {Springer},
Series = LNCS,
Title = {Composable Partial Multiparty Session Types},
Volume = 13077,
doi = "10.1007/978-3-030-90636-8_3",
Year = 2021
}

@article{zooidfull,
  author    = {David Castro{-}Perez and
               Francisco Ferreira and
               Lorenzo Gheri and
               Nobuko Yoshida},
  title     = {Zooid: a {DSL} for Certified Multiparty Computation},
  journal   = {CoRR},
  volume    = {abs/2103.10269},
  year      = {2021},
  url       = {https://arxiv.org/abs/2103.10269},
  eprinttype = {arXiv},
  eprint    = {2103.10269},
  bibsource = {dblp computer science bibliography, https://dblp.org}
}

@article{CASTELLANI2020128,
title = {Global types with internal delegation},
journal = {Theoretical Computer Science},
volume = {807},
pages = {128-153},
year = {2020},
note = {In memory of Maurice Nivat, a founding father of Theoretical Computer Science - Part II},
issn = {0304-3975},
doi = {https://doi.org/10.1016/j.tcs.2019.09.027},
url = {https://www.sciencedirect.com/science/article/pii/S030439751930578X},
author = {Ilaria Castellani and Mariangiola Dezani-Ciancaglini and Paola Giannini and Ross Horne}
}

@book{del-book,
title = {Object-oriented concurrent programming},
author = {Yonezawa, A and Tokoro, M},
publisher={The {MIT} {Press}},
place = {United States},
year = {1986},
month = {1},
isbn = "978-0262240260"
}

@article{10.1145/83880.84528,
author = {Agha, Gul},
title = {Concurrent Object-Oriented Programming},
year = {1990},
issue_date = {Sept. 1990},
publisher = {Association for Computing Machinery},
address = {New York, NY, USA},
volume = {33},
number = {9},
issn = {0001-0782},
url = {https://doi.org/10.1145/83880.84528},
doi = {10.1145/83880.84528},
journal = {Commun. ACM},
month = {sep},
pages = {125–141},
numpages = {17}
}

@InProceedings{imai_et_al:LIPIcs:2020:13166,
  author =	{Keigo Imai and Rumyana Neykova and Nobuko Yoshida and Shoji Yuen},
  title =	{{Multiparty Session Programming With Global Protocol Combinators}},
  booktitle =	{34th European Conference on Object-Oriented Programming (ECOOP 2020)},
  pages =	{9:1--9:30},
  series =	{Leibniz International Proceedings in Informatics (LIPIcs)},
  ISBN =	{978-3-95977-154-2},
  ISSN =	{1868-8969},
  year =	{2020},
  volume =	{166},
  editor =	{Robert Hirschfeld and Tobias Pape},
  publisher =	{Schloss Dagstuhl--Leibniz-Zentrum f{\"u}r Informatik},
  address =	{Dagstuhl, Germany},
  URL =		{https://drops.dagstuhl.de/opus/volltexte/2020/13166},
  URN =		{urn:nbn:de:0030-drops-131662},
  doi =		{10.4230/LIPIcs.ECOOP.2020.9}
}

@InProceedings{scalas_et_al:LIPIcs:2016:6115,
  author =	{Alceste Scalas and Nobuko Yoshida},
  title =	{{Lightweight Session Programming in Scala}},
  booktitle =	{30th European Conference on Object-Oriented Programming (ECOOP 2016)},
  pages =	{21:1--21:28},
  series =	{Leibniz International Proceedings in Informatics (LIPIcs)},
  ISBN =	{978-3-95977-014-9},
  ISSN =	{1868-8969},
  year =	{2016},
  volume =	{56},
  editor =	{Shriram Krishnamurthi and Benjamin S. Lerner},
  publisher =	{Schloss Dagstuhl--Leibniz-Zentrum fuer Informatik},
  address =	{Dagstuhl, Germany},
  URL =		{http://drops.dagstuhl.de/opus/volltexte/2016/6115},
  URN =		{urn:nbn:de:0030-drops-61156},
  doi =		{10.4230/LIPIcs.ECOOP.2016.21}
}

@article{DBLP:journals/acta/CastellaniDG19,
  author    = {Ilaria Castellani and
               Mariangiola Dezani{-}Ciancaglini and
               Paola Giannini},
  title     = {Reversible sessions with flexible choices},
  journal   = {Acta Informatica},
  volume    = {56},
  number    = {7-8},
  pages     = {553--583},
  year      = {2019},
  url       = {https://doi.org/10.1007/s00236-019-00332-y},
  doi       = {10.1007/s00236-019-00332-y},
  bibsource = {dblp computer science bibliography, https://dblp.org}
}

@InProceedings{harvey_et_al:LIPIcs.ECOOP.2021.10,
  author =	{Harvey, Paul and Fowler, Simon and Dardha, Ornela and Gay, Simon J.},
  title =	{{Multiparty Session Types for Safe Runtime Adaptation in an Actor Language}},
  booktitle =	{35th European Conference on Object-Oriented Programming (ECOOP 2021)},
  pages =	{10:1--10:30},
  series =	{Leibniz International Proceedings in Informatics (LIPIcs)},
  ISBN =	{978-3-95977-190-0},
  ISSN =	{1868-8969},
  year =	{2021},
  volume =	{194},
  editor =	{M{\o}ller, Anders and Sridharan, Manu},
  publisher =	{Schloss Dagstuhl -- Leibniz-Zentrum f{\"u}r Informatik},
  address =	{Dagstuhl, Germany},
  URL =		{https://drops.dagstuhl.de/opus/volltexte/2021/14053},
  URN =		{urn:nbn:de:0030-drops-140539},
  doi =		{10.4230/LIPIcs.ECOOP.2021.10}
}

@InProceedings{gheri_et_al:LIPIcs.ECOOP.2022.8,
  author =	{Gheri, Lorenzo and Lanese, Ivan and Sayers, Neil and Tuosto, Emilio and Yoshida, Nobuko},
  title =	{{Design-By-Contract for Flexible Multiparty Session Protocols}},
  booktitle =	{36th European Conference on Object-Oriented Programming (ECOOP 2022)},
  pages =	{8:1--8:28},
  series =	{Leibniz International Proceedings in Informatics (LIPIcs)},
  ISBN =	{978-3-95977-225-9},
  ISSN =	{1868-8969},
  year =	{2022},
  volume =	{222},
  editor =	{Ali, Karim and Vitek, Jan},
  publisher =	{Schloss Dagstuhl -- Leibniz-Zentrum f{\"u}r Informatik},
  address =	{Dagstuhl, Germany},
  URL =		{https://drops.dagstuhl.de/opus/volltexte/2022/16236},
  URN =		{urn:nbn:de:0030-drops-162367},
  doi =		{10.4230/LIPIcs.ECOOP.2022.8}
}

@misc{kerberos,
  author = {MIT},
  title = {Kerberos: The Network Authentication Protocol},
  year ={2022},
  howpublished = {\url{https://web.mit.edu/kerberos/}},
  note = {Accessed: 20/10/2022}
}

@misc{ethereum,
  author = {Ethereum},
  title = {Introduction to Smart Contracts},
  year ={2022},
  howpublished = {\url{https://ethereum.org/en/developers/docs/smart-contracts/}},
  note = {Accessed: 20/10/2022}
}

@misc{bpmn,
  author = {OMG},
  title = {Business Process Model and Notation},
  year ={2022},
  howpublished = {\url{https://www.bpmn.org/}},
  note = {Accessed: 20/10/2022}
}

@InProceedings{cledou_et_al:LIPIcs.ECOOP.2022.27,
  author =	{Cledou, Guillermina and Edixhoven, Luc and Jongmans, Sung-Shik and Proen\c{c}a, Jos\'{e}},
  title =	{{API Generation for Multiparty Session Types, Revisited and Revised Using Scala 3}},
  booktitle =	{36th European Conference on Object-Oriented Programming (ECOOP 2022)},
  pages =	{27:1--27:28},
  series =	{Leibniz International Proceedings in Informatics (LIPIcs)},
  ISBN =	{978-3-95977-225-9},
  ISSN =	{1868-8969},
  year =	{2022},
  volume =	{222},
  editor =	{Ali, Karim and Vitek, Jan},
  publisher =	{Schloss Dagstuhl -- Leibniz-Zentrum f{\"u}r Informatik},
  address =	{Dagstuhl, Germany},
  URL =		{https://drops.dagstuhl.de/opus/volltexte/2022/16255},
  URN =		{urn:nbn:de:0030-drops-162559},
  doi =		{10.4230/LIPIcs.ECOOP.2022.27}
}

@article{10.1145/3485501,
author = {Viering, Malte and Hu, Raymond and Eugster, Patrick and Ziarek, Lukasz},
title = {A Multiparty Session Typing Discipline for Fault-Tolerant Event-Driven Distributed Programming},
year = {2021},
issue_date = {October 2021},
publisher = {Association for Computing Machinery},
address = {New York, NY, USA},
volume = {5},
number = {OOPSLA},
url = {https://doi.org/10.1145/3485501},
doi = {10.1145/3485501},
journal = {Proc. ACM Program. Lang.},
month = {oct},
articleno = {124},
numpages = {30}
}

@article{10.1145/3428223,
author = {Castro-Perez, David and Yoshida, Nobuko},
title = {CAMP: Cost-Aware Multiparty Session Protocols},
year = {2020},
issue_date = {November 2020},
publisher = {Association for Computing Machinery},
address = {New York, NY, USA},
volume = {4},
number = {OOPSLA},
url = {https://doi.org/10.1145/3428223},
doi = {10.1145/3428223},
journal = {Proc. ACM Program. Lang.},
month = {nov},
articleno = {155},
numpages = {30}
}

@InProceedings{lagaillardie_et_al:LIPIcs.ECOOP.2022.4,
  author =	{Lagaillardie, Nicolas and Neykova, Rumyana and Yoshida, Nobuko},
  title =	{{Stay Safe Under Panic: Affine Rust Programming with Multiparty Session Types}},
  booktitle =	{36th European Conference on Object-Oriented Programming (ECOOP 2022)},
  pages =	{4:1--4:29},
  series =	{Leibniz International Proceedings in Informatics (LIPIcs)},
  ISBN =	{978-3-95977-225-9},
  ISSN =	{1868-8969},
  year =	{2022},
  volume =	{222},
  editor =	{Ali, Karim and Vitek, Jan},
  publisher =	{Schloss Dagstuhl -- Leibniz-Zentrum f{\"u}r Informatik},
  address =	{Dagstuhl, Germany},
  URL =		{https://drops.dagstuhl.de/opus/volltexte/2022/16232},
  URN =		{urn:nbn:de:0030-drops-162324},
  doi =		{10.4230/LIPIcs.ECOOP.2022.4}
}

@InProceedings{10.1007/978-3-642-00590-9_21,
author="Caires, Lu{\'i}s
and Vieira, Hugo Torres",
editor="Castagna, Giuseppe",
title="Conversation Types",
booktitle="Programming Languages and Systems",
year="2009",
publisher="Springer Berlin Heidelberg",
address="Berlin, Heidelberg",
pages="285--300",
doi = "https://doi.org/10.1007/978-3-642-00590-9_21",
isbn="978-3-642-00590-9"
}

@InProceedings{10.1007/978-3-662-43376-8_10,
author="Padovani, Luca
and Vasconcelos, Vasco Thudichum
and Vieira, Hugo Torres",
editor="K{\"u}hn, Eva
and Pugliese, Rosario",
title="Typing Liveness in Multiparty Communicating Systems",
booktitle="Coordination Models and Languages",
year="2014",
publisher="Springer Berlin Heidelberg",
address="Berlin, Heidelberg",
pages="147--162",
doi="https://doi.org/10.1007/978-3-662-43376-8_10",
isbn="978-3-662-43376-8"
}

@InProceedings{10.1007/978-3-642-32940-1_17,
author="Lange, Julien
and Tuosto, Emilio",
editor="Koutny, Maciej
and Ulidowski, Irek",
title="Synthesising Choreographies from Local Session Types",
booktitle="CONCUR 2012 -- Concurrency Theory",
year="2012",
publisher="Springer Berlin Heidelberg",
address="Berlin, Heidelberg",
pages="225--239",
doi = "https://doi.org/10.1007/978-3-642-32940-1_17",
isbn="978-3-642-32940-1"
}

@article{10.1145/3547638,
author = {Jacobs, Jules and Balzer, Stephanie and Krebbers, Robbert},
title = {Multiparty GV: Functional Multiparty Session Types with Certified Deadlock Freedom},
year = {2022},
issue_date = {August 2022},
publisher = {Association for Computing Machinery},
address = {New York, NY, USA},
volume = {6},
number = {ICFP},
url = {https://doi.org/10.1145/3547638},
doi = {10.1145/3547638},
journal = {Proc. ACM Program. Lang.},
month = {aug},
articleno = {107},
numpages = {30}
}

@inproceedings{10.1145/2429069.2429101,
author = {Carbone, Marco and Montesi, Fabrizio},
title = {Deadlock-Freedom-by-Design: Multiparty Asynchronous Global Programming},
year = {2013},
isbn = {9781450318327},
publisher = {Association for Computing Machinery},
address = {New York, NY, USA},
url = {https://doi.org/10.1145/2429069.2429101},
doi = {10.1145/2429069.2429101},
booktitle = {Proceedings of the 40th Annual ACM SIGPLAN-SIGACT Symposium on Principles of Programming Languages},
pages = {263–274},
numpages = {12},
location = {Rome, Italy},
series = {POPL '13}
}

@inproceedings{MFYZ21,
author = {Miu, Anson and Ferreira, Francisco and Yoshida, Nobuko and Zhou, Fangyi},
title = {Communication-Safe Web Programming in TypeScript with Routed Multiparty Session Types},
year = {2021},
isbn = {9781450383257},
publisher = {Association for Computing Machinery},
address = {New York, NY, USA},
url = {https://doi.org/10.1145/3446804.3446854},
doi = {10.1145/3446804.3446854},
booktitle = {Proceedings of the 30th ACM SIGPLAN International Conference on Compiler Construction},
pages = {94–106},
numpages = {13},
location = {Virtual, Republic of Korea},
series = {CC 2021}
}

@inproceedings{self,
    author = {Ungar, David and Smith, Randall B.},
    title = {Self: The Power of Simplicity},
    year = {1987},
    isbn = {0897912470},
    publisher = {Association for Computing Machinery},
    address = {New York, NY, USA},
    url = {https://doi.org/10.1145/38765.38828},
    doi = {10.1145/38765.38828},
    pages = {227–242},
    numpages = {16},
    location = {Orlando, Florida, USA},
    series = {OOPSLA '87}
    }

@misc{fullversion,
  doi = {10.48550/ARXIV.2302.01979},
  url = {https://arxiv.org/abs/2302.01979},
  author = {Gheri, Lorenzo and Yoshida, Nobuko},
  title = {Hybrid Multiparty Session Types -- Full Version},
  publisher = {arXiv},
  year = {2023},
  copyright = {Creative Commons Attribution Share Alike 4.0 International}
}

\appendix
\onecolumn
\newpage
\section{Communication Protocol for the Company, from Overview}\label{appendix-fig}

In Figure \ref{fig:diag-whole} the communication protocol for the whole company from \S\ref{sec:overview}.

\label{sec:appendix-fig}
\begin{figure}
  \begin{center}
    \begin{tikzpicture}[font=\footnotesize]

\node[draw,
  minimum width=0.8cm,
  minimum height=0.4cm,
  rounded corners] (d) at (0,0)
     {$\dir$};

\node[draw,
  minimum width=0.8cm,
  minimum height=0.4cm,
  rounded corners] (ad) at (1.5,0)
     {$\ad$};

\node[draw,
  minimum width=0.8cm,
  minimum height=0.4cm,
  rounded corners] (s) at (3,0)
     {$\sales$};

\node[draw,
  minimum width=0.8cm,
  minimum height=0.4cm,
  rounded corners] (w) at (4.5,0)
     {$\web$};

\node[draw,
  minimum width=0.8cm,
  minimum height=0.4cm,
  rounded corners] (f1) at (6,0)
     {$\fone$};

\node[draw,
  minimum width=0.8cm,
  minimum height=0.4cm,
  rounded corners] (f2) at (7.5,0)
     {$\ftwo$};

\draw[-](d) edge (0,-5.3);
\draw[-](ad) edge (1.5,-5.3);
\draw[-](s) edge (3,-5.3);
\draw[-](w) edge (4.5,-5.3);
\draw[-](f1) edge (6,-5.3);
\draw[-](f2) edge (7.5,-5.3);
\draw[-](0,-5.7) edge (0,-8.4);
\draw[-](1.5,-5.7) edge (1.5,-8.4);
\draw[-](3,-5.7) edge (3,-8.4);
\draw[-](4.5,-5.7) edge (4.5,-8.4);
\draw[-](6,-5.7) edge (6,-8.4);
\draw[-](7.5,-5.7) edge (7.5,-8.4);
\draw[-{Stealth}] (0,-1) --node[above]{$\product$} (1.5,-1);
\draw[-{Triangle[open]},thick,cbred] (0,-1.5) --node[above]{$\product$} (3,-1.5);
\draw[-{Triangle[open]},thick,cbred] (0,-2) --node[above]{$\product$} (6,-2);
\draw[-{Stealth}] (6,-2.5) --node[above]{$\product$} (7.5,-2.5);
\draw[-{Stealth}] (7.5,-3) --node[above]{$\price$} (6,-3);
\draw[-{Triangle[open]},thick,cbred] (6,-3.5) --node[above]{$\ok$} (0,-3.5);
\draw[-{Stealth}] (0,-4) --node[above]{$\go$} (1.5,-4);
\draw[-{Triangle[open]},thick,cbred] (6,-4.5) --node[above]{$\price$} (3,-4.5);
\draw[-{Stealth}] (3,-5) --node[above]{$\publish$} (4.5,-5);
\draw[-{Stealth}] (7.5,-6) --node[above]{$\wait$} (6,-6);
\draw[-{Triangle[open]},thick,cbred] (6,-6.5) --node[above]{$\wait$} (0,-6.5);
\draw[-{Stealth}] (0,-7) --node[above]{$\wait$} (1.5,-7);
\draw[-{Triangle[open]},thick,cbred] (6,-7.5) --node[above]{$\wait$} (3,-7.5);
\draw[-{Stealth}] (3,-8) --node[above]{$\wait$} (4.5,-8);

\draw[decorate,decoration={brace},thick] (7.6,-3) -- (7.6,-6);

\draw[-,thick] (-0.4,-8.2) -- (-0.1,-8.2);
\draw[-,thick] (-0.4,-8.2) -- (-0.4,-2.8);
\draw[->,thick] (-0.4,-2.8) -- (-0.1,-2.8);
\draw[-,dashed,gray,thin] (0,-2.8) -- (7.5,-2.8);
\draw[-,dashed,gray,thin] (0,-8.2) -- (7.5,-8.2);

\end{tikzpicture}
  \end{center}
  \caption{Whole Communication in the Company}
  \label{fig:diag-whole}
 \end{figure}

\section{Definitions and Proofs}\label{sec:appendix}

\subsection{Compositionality: Building Back a Subprotocol into a Global One}\label{subsec:appendix1}

\begin{definition}[Guardedness for Hybrid Types]
  \label{def:guarded-a}%
  We define \emph{guardedness} as a predicate on hybrid types, with the following inductive rules:
  \[
    \begin{array}{c}
      \dfrac{}{\guarded\ \hend}\quad
      \dfrac{}{\guarded\ \hX}\quad
      \dfrac{\code{not\_pure\_rec}\;\hX\;\hT \quad \guarded\ \hT}{\guarded\ (\hrec \hX \hT)}\quad  \dfrac{\guarded\;\hTone \quad \guarded\;\hTtwo}{\guarded\ (\hpar\hTone\hTtwo)}
      \\[3mm]
      \dfrac{\forall i\in I. \guarded\ {\dht{\hT_i}}}{\guarded\ (\hsend \p\q \ell {\tS} {\hT})}\quad
      \dfrac{\forall i\in I. \guarded\ {\dht{\hT_i}}}{\guarded\ (\hrecv \p\q \ell {\tS} {\hT})}\quad
      \dfrac{\forall i\in I. \guarded\ {\dht{\hT_i}}}{\guarded\ (\hmsgi \p\q \ell {\tS} {\hT})}
    \end{array}\vspace*{3mm}
  \]
  where $\code{not\_pure\_rec}\;\hX\;\hT$ means that $\hT$ is different
  from $\hrec {\dht{Y_1}} {\dots\hrec {\dht{Y_n}} \hX}$ and also
  $\hT\neq\hX$.
\end{definition}

The above definition is adapted from the Coq-mechanised MPST metatheory in \cite{zooid} (Definitions A.2 and A.10 in the appendix of the extended version of the paper \cite{zooidfull}).

The next six lemmas are preliminary results, adapted from classic
multiparty session type theories. They are
needed for the base cases of induction
in the build-back theorem (Theorem \ref{thm:comp1-a}). For all the
lemmas below, the proofs go by structural induction.

\begin{lemma}[Projection and Closed Types]\label{lem:proj-clo-a}
  $\projb \hT E$ is closed iff $\hT$ is closed.
\end{lemma}

\begin{lemma}[Localiser and Closed Types]\label{lem:loc-clo-a}
  $\loc{} \hT$ is closed iff $\hT$ is closed.
\end{lemma}

\begin{lemma}[Projection onto Internal Participants]\label{lem:proj-fix-a}
If $\hpart \hT \subseteq E$ then $\projb \hT E = \hT$.
\end{lemma}

\begin{lemma}[Projection onto Disjoint Participants $(1)$]\label{lem:proj-end-a}
If $\hpart \hT \cap E'=\emptyset$ and $\hT$ is closed, then $\projb \hT {E'} = \hend$.
\end{lemma}

\begin{lemma}[Projection onto Disjoint Participants $(2)$]\label{lem:proj-var-a}
If $\hpart \hT \cap E'=\emptyset$ and $\hT$ is \emph{not} closed, then $\projb \hT {E'} = \dht{\mu Y_1.\dots.\mu Y_n.}\hX$ (if such projection exists).
\end{lemma}

\begin{lemma}[Projection is Guarded]\label{lem:proj-guar-a}
If $\projb \hT {E}$ is guarded (if it is defined and $\hT$ is guarded).
\end{lemma}

The following definitions and lemmas introduce objects and results
needed for the main algorithm and ultimately
for proving Theorem \ref{thm:comp1-a}.
They guarantee the good behaviour of our compositionality algorithm,
with respect to full merge (\cite{DenielouYoshida2013}
and Definition \ref{def:proj} in this paper) for projection
and its ``dual'' for the localiser (Definition \ref{def:loc} in this paper).

The \emph{binary unmerge for localiser} $\ulb$ takes a type $\ghT$ (first argument) and ``unmerges'' it, according to the structure of two other types $\hTLone$ and $\hTLtwo$ (second and third arguments); thus $\ulb$ returns a pair of types, the unmerged components of $\ghT$. The non-binary version of the \emph{unmerge for localiser} $\ul$ is defined by iterating the binary one, thus allowing for unmerging into $n$ types. The role of this function is clarified by Lemma \ref{lem:loc-bb}.

To correctly define such function we need an auxiliary function that handles the case in which a projection merge is equal to a localiser one: $\merge_{i\in I} \ghTi = \lmerge_{j\in J} \dht{\hTL_j}$. In this case we want to obtain, for all $i\in I,j\in J$, $\dht{\hTL^i_j}$ such that $\merge_i \dht{\hTL^i_j} =  \dht{\hTL_j}$ and $\lmerge_j \dht{\hTL^i_j} = \ghTi$.

\begin{definition}[Localiser Unmerge for Projection Merge]\label{def:unmlp}
  The partial function \emph{binary localiser-unmerge for projection merge} $\ulpba \ghTone \ghTtwo \hTLone \hTLtwo$ is recursively defined by the following equations, whenever the recursive call is defined (we use $\pi_1$ and $\pi_2$ to indicate standard pair projections):
  \begin{small}
  \[
  \begin{array}{l}
    \ulpba \hend \hend \hend \hend = ((\hend, \hend), (\hend,\hend))\\
    \ulpba \hX \hX \hX \hX = ((\hX, \hX), (\hX,\hX))\\
    \ulpba {\hrec \hX \ghTone}  {\hrec \hX \ghTtwo} {\hrec \hX \hTLone} {\hrec \hX \hTLtwo} = ((\hrec\hX\hToneone,\hrec\hX\hTonetwo),(\hrec\hX\hTtwotwo,\hrec\hX\hTtwotwo))\\
    \qquad \text{where}\ \dht{\hT^n_m}=\pi_n(\pi_m(\ulpba {\ghTone}  {\ghTtwo} {\hTLone} {\hTLtwo})))\ \text{for $n,m=1,2$}\\
    \ulpba {(\hsend \p \q \ell \tS \ghTone)}
           {(\hsend \p \q \ell \tS \ghTtwo)}
           {(\hsendj \p \q \ell \tS \hTLone)}
           {(\hsendk \p \q \ell \tS \hTLtwo)} = \\
    \qquad ((\hsendj \p \q \ell \tS \hToneone,
            \hsendk \p \q \ell \tS \hTonetwo),
            (\hsendj \p \q \ell \tS \hTtwoone,
            \hsendk \p \q \ell \tS \hTtwotwo))\\
    \qquad \text{where $J\cup K = I$ and, for $n,m=1,2$,}\ \dht{{\hT^n_m}_i}=
    \begin{cases}
      \pi_n(\pi_m(\ulpba {\dht{{\ghTone}_i}} {\dht{{\ghTtwo}_i}} {\dht{{\hTLone}_i}} {\dht{{\hTLtwo}_i}}))\ \text{if $i\in J\cup K$}\\
       \pi_n(\pi_m(\ulpba {\dht{{\ghTone}_i}} {\dht{{\ghTtwo}_i}} {\dht{{\hTLone}_i}} {\dht{{\hTLone}_i}}))\ \text{if $i\in J\backslash K$}\\
       \pi_n(\pi_m(\ulpba {\dht{{\ghTone}_i}} {\dht{{\ghTtwo}_i}} {\dht{{\hTLtwo}_i}} {\dht{{\hTLtwo}_i}}))\ \text{if $i\in K\backslash J$}
    \end{cases}\\
        \ulpba {(\hrecvj \p \q \ell \tS \ghTone)}
           {(\hrecvk \p \q \ell \tS \ghTtwo)}
           {(\hrecv \p \q \ell \tS \hTLone)}
           {(\hrecv \p \q \ell \tS \hTLtwo)} = \\
    \qquad ((\hrecvj \p \q \ell \tS \hToneone,
            \hrecvj \p \q \ell \tS \hTonetwo),
            (\hrecvk \p \q \ell \tS \hTtwoone,
            \hrecvk \p \q \ell \tS \hTtwotwo))\\
    \qquad \text{where $J\cup K = I$ and, for $n,m=1,2$,}\ \dht{{\hT^n_m}_i}=
    \begin{cases}
      \pi_n(\pi_m(\ulpba {\dht{{\ghTone}_i}} {\dht{{\ghTtwo}_i}} {\dht{{\hTLone}_i}} {\dht{{\hTLtwo}_i}}))\ \text{if $i\in J\cup K$}\\
       \pi_n(\pi_m(\ulpba {\dht{{\ghTone}_i}} {\dht{{\ghTone}_i}} {\dht{{\hTLone}_i}} {\dht{{\hTLtwo}_i}}))\ \text{if $i\in J\backslash K$}\\
       \pi_n(\pi_m(\ulpba {\dht{{\ghTtwo}_i}} {\dht{{\ghTtwo}_i}} {\dht{{\hTLone}_i}} {\dht{{\hTLtwo}_i}}))\ \text{if $i\in K\backslash J$}
    \end{cases}\\
    \ulpba {(\hmsgi \p \q \ell \tS \ghTone)}
           {(\hmsgi \p \q \ell \tS \ghTtwo)}
           {(\hmsgi \p \q \ell \tS \hTLone)}
           {(\hmsgi \p \q \ell \tS \hTLtwo)}
           \\=     ((\hmsgi \p \q \ell \tS \hToneone,
            \hmsgi \p \q \ell \tS \hTonetwo),
            (\hmsgi \p \q \ell \tS \hTtwoone,
            \hmsgi \p \q \ell \tS \hTtwotwo))\\
    \qquad \text{where, for $n,m=1,2$,}\ \dht{{\hT^n_m}_i}=
    \pi_n(\pi_m(\ulpba {\dht{{\ghTone}_i}} {\dht{{\ghTtwo}_i}} {\dht{{\hTLone}_i}} {\dht{{\hTLtwo}_i}}))\\[2mm]
      \text{undefined otherwise}
  \end{array}
  \]
  \end{small}
  The partial function \emph{localiser-unmerge for projection merge} $\ulpa {L_P} {L_L}$ takes as input two lists of types instead of two pairs. It returns a tuple of lists: if $L_P= [\ghTone,\dots,\dht{\ghT_n}]$ and $L_L=[\hTLone,\dots,\dht{\hTL_m}]$ the result list is of the kind $([\dht{\hT^1_1},\dots,\dht{\hT^1_m}],\dots,[\dht{\hT^n_1},\dots,\dht{\hT^n_m}])$ ($\pi_i$ indicate the standard projection of the $i$-th component of the tuple). $\ulp$ is defined by recursion (on the sum of the lengths of $L_P$ and $L_L$) by the following equations.
  \begin{small}
  \[
  \begin{array}{l}
    \ulpa {[]} {L_L} = ([])\\
    \ulpa {[\ghT]} {L_L} = ([L_L])\\
    \ulpa  {[\ghTone, \ghTtwo]} {L_L} =
    \begin{cases}
      ([],[])\ \text{if $L_L=[]$}\\
      ([\pi_1(\pi_1(\ulpba \ghTone \ghTtwo \hTL \hTL)),\pi_1(\pi_2(\ulpba \ghTone \ghTtwo \hTL \hTL))])\ \text{if $L_L=[\hTL]$}\\
      ([\pi_1(\pi_1(\ulpba \ghTone \ghTtwo \hTLone \hTLtwo)),
        \pi_2(\pi_1(\ulpba \ghTone \ghTtwo \hTLone \hTLtwo))],\\
      \ \
       [\pi_1(\pi_2(\ulpba \ghTone \ghTtwo \hTLone \hTLtwo)),
         \pi_2(\pi_2(\ulpba \ghTone \ghTtwo \hTLone \hTLtwo))])\
      \text{if $L_L=[\hTLone,\hTLtwo]$}\\
      (\pi_1(\pi_1(U))\#L^1,\pi_1(\pi_2(U))\#L^2)\\
         \qquad\text{if $L_L=[\hTLone,\dots,\dht{\hTL_m}]$, where, for $i=1,2$}\\
         \qquad L^i=\pi_i(\ulpa {[(\pi_2(\pi_1(U))),(\pi_2(\pi_2(U)))]} {[\hTLtwo,\dots,\dht{\hTL_m}]})\ \text{and}\\
         \qquad U= \ulpba \ghTone \ghTtwo \hTLone {(\lmerge_{k=2\dots m}\dht{\hTL_k})}
    \end{cases}\\
    \ulpa  {[\ghTone, \ghTtwo,\dots,\dht{\ghT_n}]} {L_L}
    = (L^1,\pi_1(\ulpa {[\ghTtwo,\dots,\dht{\ghT_n}]} {L^2}),\dots,\pi_{n-1}(\ulpa {[\ghTtwo,\dots,\dht{\ghT_n}]} {L^2}))\\
    \qquad \text{where, for $i=1,2$,}\ L^i = \ulpa {[\ghTone,(\merge_{k=2,\dots,n}\dht{\ghT_k})]} {L_L}
  \end{array}
  \]
  \end{small}
\end{definition}

\begin{lemma}\label{lem:unmlp}
  We are given hybrid types $\ghTone,\dots,\dht{\ghT_n}$ and $\hTLone,\dots,\dht{\hTL_m}$, such that $\merge_{i=1,\dots,n} \ghTi = \lmerge_{j=1,\dots,m} \dht{\hTL_j}$. \\Then $\ \ulpa [\ghTone,\dots,\dht{\ghT_n}] [\hTLone,\dots,\dht{\hTL_m}]\ $ is defined and, if we set
  \[
  \ulpa [\ghTone,\dots,\dht{\ghT_n}] [\hTLone,\dots,\dht{\hTL_m}] = ([\dht{\hTL^1_1},\dots,\dht{\hTL^1_m}],\dots,[\dht{\hTL^n_1},\dots,\dht{\hTL^n_m}])\text{,}
  \]
  we have that:
  \begin{itemize}
  \item $\lmerge_{j=1,\dots,m} ,\dht{\hTL^i_j} = \ghTi$, for all $i=1,\dots,n$, and
  \item $\merge_{i=1,\dots,n} ,\dht{\hTL^i_j} = \dht{\hTL_j}$, for all $j=1,\dots,n$.
  \end{itemize}
\end{lemma}

\begin{proof}\ \\
  \begin{itemize}
  \item We first prove the base cases, in particular the case in which $n=2$ and $m=2$ (formally we first prove the result for $\ulpb$). This is done by structural induction on $\ghTone$, exploiting the hypothesis and the inductive nature both of $\mrg$ and $\lmrg$.
  \item We conclude by induction on $n+m$, by following the definition of $\ulp$.
  \end{itemize}
\end{proof}

We now define the unmerge for the localiser (Definition \ref{def:unml}) and ensure that it has the desired behaviour (Lemma \ref{lem:unml}).

\begin{definition}[Localiser Unmerge]\label{def:unml}
  The partial function \emph{binary unmerge for localiser} $\ulbEa \ghT \hTLone \hTLtwo$ is defined if $\hpart\hTLone\cup\\hpart\hTLtwo\subseteq E$ by the following equations, by recursion on $\hdepth\ghT$, whenever the recursive call is defined (we use $\pi_1$ and $\pi_2$ to indicate standard pair projections; we also make implicit the dependency on $E$ using the notation $\ulb$ instead of $ulbE$):
  \begin{small}
  \[
  \begin{array}{l}
    \ulba \hend \hTLone \hTLtwo = (\hend,\hend)\\
    \ulba \hX \hTLone \hTLtwo = (\hX,\hX)\\
    \ulba {(\hrec \hX \ghT)} {\hTLone} {\hTLtwo} =
    \begin{cases}
      (\hrec \hX {\pi_1(\ulba \ghT \hTLoneb \hTLtwob)}, \hrec \hX {\pi_2(\ulba \ghT \hTLoneb \hTLtwob)})
      \\ \qquad \text{if $\hTLone=\hrec \hX \hTLoneb$ and $\hTLtwo=\hrec \hX \hTLtwob$}\\
      (\hrec \hX \ghT,\hrec \hX \ghT)\ \text{if $\hTLone=\hTLtwo=\hend$}
      \end{cases}\\
    \ulba {(\hpar \ghTone \ghTtwo)} \hTLone \hTLtwo =
    \begin{cases}
      (\hpar {(\pi_1(\ulba \ghTone \hTLone \hTLtwo))} \ghTtwo, \hpar {(\pi_2(\ulba \ghTone \hTLone \hTLtwo))} \ghTtwo)\ \\
      \qquad \text{if $E\cap\hpart\ghTtwo=\emptyset$}\\
      (\hpar {\ghTone} {(\pi_1(\ulba \ghTtwo \hTLone \hTLtwo))}, \hpar \ghTone {(\pi_2(\ulba \ghTtwo \hTLone \hTLtwo))})\
      \\ \qquad \text{if $E\cap\hpart\ghTone=\emptyset$}\\
      \text{undefined otherwise}
      \end{cases}
    \\
    \ulba {(\hsend \p \q \ell \tS \ghT)} {\hTLone} {\hTLtwo} =
    \\ \qquad
    \begin{cases}
      \hspace*{-1.7mm}\begin{array}{l}
        (\hsendj \p \q \ell \tS {\ghT^1},\hsendk \p \q \ell \tS {\ghT^2})\\
        \qquad \text{with $\dht{\ghT^1_j}=\ghTj$ for $j\in J\backslash K$, $\dht{\ghT^1_j}= \pi_1(\ulba \ghTj \hTLone \hTLtwo)$ for $j\in J\cap K$,}
    \\
        \qquad \text{$\dht{\ghT^2_k}=\ghTk$ for $k\in K\backslash J$, $\dht{\ghT^2_k}= \pi_2(\ulba \ghTk \hTLone \hTLtwo)$ for $k\in J\cap K$}
        \\ \qquad \text{if $\hTLone = \hsendj \p \q {\ell} \tS \hTLone$, $\hTLtwo= \hsendk \p \q {\ell} \tS \hTLtwo$ and $I=J\cup K$}
      \end{array}\\
      \hspace*{-1.7mm}\begin{array}{l}
        (\hsnd \p \q \dht{\{\lblFmt{\ell_i} (\tS_i). \pi_1(\ulba \ghTi {\dht{\hTLone_i}} {\dht{\hTLtwo_i}})\}_{i \in I}},\hsnd \p \q \dht{\{\lblFmt{\ell_i} (\tS_i). \pi_2(\ulba \ghTi {\dht{\hTLone_i}} {\dht{\hTLtwo_i}})\}_{i \in I}})
        \\ \qquad \text{with $\ulpa {[\projb{\dht{\ghT_i}}E]_{i\in I}} {[\hTLone,\hTLtwo]} = ([\dht{\hTLone_i},\dht{\hTLtwo_i}])_{i\in I}$, in any other case }
      \end{array}\\
    \end{cases}\\
    \ulba {(\hrecv \p \q \ell \tS \ghT)} {\hTLone} {\hTLtwo} =
    \\ \qquad
    \begin{cases}
       \hspace*{-1.7mm}\begin{array}{l}
        (\hrecv \p \q \ell \tS {\ghT^1},\hrecv \p \q \ell \tS {\ghT^2})\\
        \qquad \text{with $\dht{\ghT^1_i}= \pi_1(\ulba \ghTi \hTLone \hTLtwo)$ and $\dht{\ghT^2_i}= \pi_2(\ulba \ghTk \hTLone \hTLtwo)$ for $i\in I$,}
        \\ \qquad \text{if $\hTLone = \hrecv \p \q {\ell} \tS \hTLone$, $\hTLtwo= \hrecv \p \q {\ell} \tS \hTLtwo$}
      \end{array}\\
      \hspace*{-1.7mm}\begin{array}{l}
        (\hrcv \p \q \dht{\{\lblFmt{\ell_i} (\tS_i). \pi_1(\ulba \ghTi {\dht{\hTLone_i}} {\dht{\hTLtwo_i}})\}_{i \in I}},\hrcv \p \q \dht{\{\lblFmt{\ell_i} (\tS_i). \pi_2(\ulba \ghTi {\dht{\hTLone_i}} {\dht{\hTLtwo_i}})\}_{i \in I}})
        \\ \qquad \text{with $\ulpa {[\projb{\dht{\ghT_i}}E]_{i\in I}} {[\hTLone,\hTLtwo]} = ([\dht{\hTLone_i},\dht{\hTLtwo_i}])_{i\in I}$, in any other case }
      \end{array}
    \end{cases}\\
    \ulba {(\hmsgi \p \q \ell \tS \ghT)} {\hTLone} {\hTLtwo} =\\
    \qquad\begin{cases}
        \hspace*{-1.7mm}\begin{array}{l}
        (\hmsgj \p \q \ell \tS {\ghT^1},\hmsgk \p \q \ell \tS {\ghT^2})\\
        \qquad \text{with $\dht{\ghT^1_j}=\ghTj$ for $j\in J\backslash K$, $\dht{\ghT^1_j}= \pi_1(\ulba \ghTj \hTLone \hTLtwo)$ for $j\in J\cap K$,}
    \\
        \qquad \text{$\dht{\ghT^2_k}=\ghTk$ for $k\in K\backslash J$, $\dht{\ghT^2_k}= \pi_2(\ulba \ghTk \hTLone \hTLtwo)$ for $k\in J\cap K$}
        \\ \qquad \text{if $\hTLone = \hsendj \p \q {\ell} \tS \hTLone$, $\hTLtwo= \hsendk \p \q {\ell} \tS \hTLtwo$, and $I=J\cup K$,}
        \\ \qquad \text{or if $\hTLone = \hrecvj \p \q {\ell} \tS \hTLone$, $\hTLtwo= \hrecvk \p \q {\ell} \tS \hTLtwo$, and $I=J=K$,}
        \\ \qquad \text{or if $\hTLone = \hmsgj \p \q {\ell} \tS \hTLone$, $\hTLtwo= \hmsgk \p \q {\ell} \tS \hTLtwo$,}
        \\ \qquad \qquad \qquad \qquad \qquad \qquad  \qquad \qquad \qquad \qquad \qquad  \qquad \qquad \text{and $I=J=K$}
        \end{array}\\
      \hspace*{-1.7mm}\begin{array}{l}
        (\hmsg \p \q \dht{\{\lblFmt{\ell_i} (\tS_i). \pi_1(\ulba \ghTi {\dht{\hTLone_i}} {\dht{\hTLtwo_i}})\}_{i \in I}},
        \\ \qquad \hmsg \p \q \dht{\{\lblFmt{\ell_i} (\tS_i). \pi_2(\ulba \ghTi {\dht{\hTLone_i}} {\dht{\hTLtwo_i}})\}_{i \in I}})
        \\ \qquad \text{with $\ulpa {[\projb{\dht{\ghT_i}}E]_{i\in I}} {[\hTLone,\hTLtwo]} = ([\dht{\hTLone_i},\dht{\hTLtwo_i}])_{i\in I}$, in any other case }
      \end{array}
    \end{cases}
  \end{array}
  \]
  \end{small}
  Given a set of participants $E$, the partial function \emph{unmerge for localiser} $\ula {\ghT} {\dht{L}}$ (again we drop the dependency on $E$ in the notation) is defined if $\bigcup_{\hTL\in \sf{set}(L)}\hpart\hTL$ by the following equations, by recursion on the list of hybrid types $\dht{L}$, whenever the recursive call is defined:
  \begin{small}
  \[
  \begin{array}{l}
    \ula {\ghT} {[]} = \text{undefined}\\
    \ula {\ghT} {[\hTL]} = [\ghT]\\
    \ula {\ghT} {[\hTLone,\hTLtwo]} = [\pi_1(\ulba {\ghT} \hTLone \hTLtwo),\pi_2(\ulba {\ghT} \hTLone \hTLtwo)]\\
    \ula {\ghT} {(\hTLone\#(\hTLtwo\#\dht{L}))} =  \\
    \qquad\pi_1(\ulba {\ghT} \hTLone {(\lmerge({\sf set} (\hTLtwo\#\dht{L}))}) \#
    \ula {\pi_2(\ulba {\ghT} \hTLone {(\lmerge({\sf set} (\hTLtwo\#\dht{L}))})} {(\hTLtwo\#\dht{L})}
  \end{array}
  \]
  \end{small}
  We observe that ${\sf {length}}\;(\ula \ghT {\dht{L}})={\sf{length}}\;\dht{L}$ and that $\hdepth{\dht{\ghT_t}}=\hdepth{\ghT}$, where $\dht{\ghT_t}$ is the $t$-th element of the list $\ula \ghT {[\hTLone,\hTLtwo,\dots,\dht{\hTL_n}]}$ (when defined).
\end{definition}

\begin{lemma}[Unmerge for Localiser]\label{lem:loc-bb}\label{lem:unml}
  We fix a set of participants $E$, and we are given hybrid types $\ghT$ and $\dht{\hTL_t}$, $t=1,\dots,n$, such that:
  \begin{itemize}
  \item[(a)] $\hpart {\dht{\hTL_t}} \subseteq E$, for all $t=n$,
  \item[(b)] $\hepart  {\dht{\hTL_t}} \cap E = \emptyset$,
  \item[(c)] $\projb {\ghT} {E} = \lmerge_{t=1,\dots,n} {\dht{\hTL_t}}$;
  \end{itemize}
  then, for all $t=1,\dots,n$, we call $\dht{\ghT_t}$ the $t$-th element of the list $\ula \ghT {[\hTLone,\hTLtwo,\dots,\dht{\hTL_n}]}$ and we have
  \begin{enumerate}
  \item $\projb {\dht{\ghT_t}} E = {\dht{\hTL_t}}$ and
  \item $\merge_{t=1,\dots,n} (\projb {\dht{\ghT_t}} {E'}) = \projb \ghT {E'}$, for all $E'\cap E = \emptyset$.
  \end{enumerate}
\end{lemma}
\begin{proof}
  Let us first prove the lemma for $n=2$. In particular, hypothesis $(c)$ gives us $\projb {\ghT} {E} = \hTLone \lmrg \hTLtwo$. We need to prove the thesis for $\ghTone = \pi_1(\ulba \ghT \hTLone \hTLtwo)$ and $\ghTtwo = \pi_2(\ulba \ghT \hTLone \hTLtwo)$, namely that $\projb \ghTone E = \hTLone$, $\projb \ghTtwo E = \hTLtwo$, and $\projb \ghTone E' \mrg \projb \ghTtwo E' = \projb \ghT {E'}$, for all $E'\cap E = \emptyset$. We proceed by induction on $\hdepth \ghT + \sf{max}(\hdepth \hTLone, \hdepth \hTLtwo)$.

  Base cases $\ghT=\hend$, and $\ghT=\hX$ are trivial. Let us consider $\ghT=\hrec \hX \ghTb$; if $\projb \ghT E = \hend$ then $\hTLone = \hTLtwo = \hend$ and the thesis follows by definition. If $\projb \ghT E = \hrec \hX {(\projb \ghTb E)}$ then $\hTLone=\hrec\hX\hTLoneb$ and $\hTLtwo=\hrec\hX\hTLtwob$ and the thesis follows by applying induction hypothesis to $\ghTb$, $\hTLoneb$, and $\hTLtwob$.

  For case $\ghT = \hpar {\dht{\ghTb_1}} {\dht{\ghTb_2}}$, we can assume without loss of generality $\projb{\ghT}E = \projb {\dht{\ghTb_1}}E$, hence $\ulba {(\hpar \ghTone \ghTtwo)} \hTLone \hTLtwo = (\hpar {(\pi_1(\ulba {\dht{\ghTb_1}} \hTLone \hTLtwo))} {\dht{\ghTb_2}}, \hpar {(\pi_2(\ulba {\dht{\ghTb_1}} \hTLone \hTLtwo))} {\dht{\ghTb_2}})$. We conclude by induction hypothesis.

  We are left with the message cases $\hsend \p \q \ell \tS \ghT$, $\hrecv \p \q \ell \tS \ghT$, and $\hmsgi \p \q \ell \tS \ghT$. These are all analogous and we will see in detail only $\ghT=\hmsgi \p \q \ell \tS \ghT$, which is the most general as it can project on every message construct: $\hmsg \p \q\dots$, $\hsnd\p\q\dots$, and $\hrcv\p\q\dots$.

  \ \\{[$\{\p,\q\}\subseteq E$]} then $\projb\ghT E=\hmsg \p \q \dht{\{\lblFmt{\ell}(\tS).(\projb\ghTi E)\}_{i\in I}}$ and, since $\lmrg$ cannot output a global message $\hmsg\p\q\dots$, by hypothesis $(c)$ the thesis holds vacuously.

  \ \\{[$\p\in E$ and $\q\notin E$]} We have (hypothesis $(c)$) $\projb \ghT E = \hTLone \lmrg \hTLtwo$, namely
  \[
  \hsnd \p \q \dht{\{\lblFmt{\ell}(\tS).(\projb\ghTi E)\}_{i\in I}}=\hsnd \p \q \dht{\{\lblFmt{\ell}(\tS).\hTLone_j\}_{j\in J}}\ \lmrg\ \hsnd \p \q \dht{\{\lblFmt{\ell}(\tS).\hTLtwo_k\}_{k\in K}}\ \text{with $J\cup K=I$.}
  \]
We observe that
  \begin{itemize}
  \item if $i\in J\backslash K$ then $\projb \ghTi E = \dht{\hTLone_i}$,
  \item if $i\in K\backslash J$ then $\projb \ghTi E = \dht{\hTLtwo_i}$, and
  \item if $i\in J\cap K$ then $\projb \ghTi E = \dht{\hTLone_i} \lmrg \dht{\hTLtwo_i}$.
  \end{itemize}
 Also, by construction, $\ulba \ghT \hTLone \hTLtwo = (\hsendj \p \q \ell \tS {\ghT^1},\hsendk \p \q \ell \tS {\ghT^2})$, with $\dht{\ghT^1_j}=\ghTj$ for $j\in J\backslash K$, $\dht{\ghT^1_j}= \pi_1(\ulba \ghTj {\dht{\hTLone_j}} {\dht{\hTLtwo_j}})$ for $j\in J\cap K$, $\dht{\ghT^2_k}=\ghTk$ for $k\in K\backslash J$, $\dht{\ghT^2_k}= \pi_2(\ulba \ghTk {\dht{\hTLone_k}} {\dht{\hTLtwo_k}})$ for $k\in J\cap K$. We use the notation $\ghTone = \pi_1(\ulba \ghT \hTLone \hTLtwo)= \hsendj \p \q \ell \tS {\ghT^1}$ and $\ghTtwo = \pi_2(\ulba \ghT \hTLone \hTLtwo)= \hsendk \p \q \ell \tS {\ghT^2}$.

    It is immediate, by construction and induction hypothesis, to prove that $\projb \ghTone E = \hTLone$ and that $\projb \ghTtwo E = \hTLtwo$. Let us consider $E'$, such that $E'\cap E =\emptyset$ and $\q\notin E'$, namely $\projb \ghTone {E'} = \merge_{j\in J}\projb{\dht{\ghT^1_j}}{E'}$ and $\projb \ghTtwo {E'}=\merge_{k\in K}\projb{\dht{\ghT^2_k}}{E'}$; then
  \[
  (\projb \ghTone {E'})\mrg(\projb \ghTtwo {E'})=(\merge_{j\in J\backslash K}\projb{\dht{\ghT^1_j}}{E'})\mrg(\merge_{k\in K\backslash J}\projb{\dht{\ghT^2_k}}{E'})\mrg(\merge_{i\in K\cap J}(\projb{\dht{\ghT^1_i}}{E'}\mrg\projb{\dht{\ghT^2_i}}{E'}))
  \]
  by associativity and commutativity of merge. We conclude by construction of $\dht{\ghT^1_j}$ and $\dht{\ghT^2_k}$ and applying induction hypothesis. Now, if $E'$ is such that $E'\cap E =\emptyset$ and $\q\in E'$, we have  $\projb \ghTone {E'} = \hrcv \p \q \dht{\{\lblFmt{\ell}(\tS).(\projb{\ghTone_j} {E'})\}_{j\in J}}$ and $\projb \ghTtwo {E'}= \hrcv \p \q \dht{\{\lblFmt{\ell}(\tS).(\projb{\ghTtwo_k} {E'})\}_{k\in K}}$. We conclude by definition of $\mrg$ (in particular its behaviour with respect to receive constructs) and induction hypothesis.

 \ \\{[$\p\notin E$ and $\q\in E$]}
We have (hypothesis $(c)$) $\projb \ghT E = \hTLone \lmrg \hTLtwo$, namely
  \[
  \hrcv \p \q \dht{\{\lblFmt{\ell}(\tS).(\projb\ghTi E)\}_{i\in I}}=\hrcv \p \q \dht{\{\lblFmt{\ell}(\tS).\hTLone_i\}_{i\in I}}\ \lmrg\ \hrcv \p \q \dht{\{\lblFmt{\ell}(\tS).\hTLtwo_i\}_{i\in I}}\text{.}
  \]
  We observe that $\projb \ghTi E = \dht{\hTLone_i} \lmrg \dht{\hTLtwo_i}$, and also that, by construction,\\ $\ulba \ghT \hTLone \hTLtwo = (\hrecv \p \q \ell \tS {\ghT^1},\hrecv \p \q \ell \tS {\ghT^2})$, with $\dht{\ghT^1_i}= \pi_1(\ulba \ghTi {\dht{\hTLone_i}} {\dht{\hTLtwo_i}})$ and $\dht{\ghT^2_i}= \pi_2(\ulba \ghTi {\dht{\hTLone_i}} {\dht{\hTLtwo_i}})$. We use the notation \\$\ghTone = \pi_1(\ulba \ghT \hTLone \hTLtwo)= \hrecv \p \q \ell \tS {\ghT^1}$ and \\$\ghTtwo = \pi_2(\ulba \ghT \hTLone \hTLtwo)= \hrecv \p \q \ell \tS {\ghT^2}$.

    It is immediate, by construction and induction hypothesis, to prove that $\projb \ghTone E = \hTLone$ and that $\projb \ghTtwo E = \hTLtwo$. Let us consider $E'$, such that $E'\cap E =\emptyset$ and $\p\notin E'$, namely $\projb \ghTone {E'} = \merge_{i\in I}\projb{\dht{\ghT^1_i}}{E'}$ and $\projb \ghTtwo {E'}=\merge_{i\in I}\projb{\dht{\ghT^2_i}}{E'}$; then
  \[
  (\projb \ghTone {E'})\mrg(\projb \ghTtwo {E'})=\merge_{i\in I}(\projb{\dht{\ghT^1_i}}{E'}\mrg\projb{\dht{\ghT^2_i}}{E'})
  \]
  by associativity and commutativity of merge. We conclude by construction of $\dht{\ghT^1_i}$ and $\dht{\ghT^2_i}$ and applying induction hypothesis. Now, if $E'$ is such that $E'\cap E =\emptyset$ and $\p\in E'$, we have  $\projb \ghTone {E'} = \hsnd \p \q \dht{\{\lblFmt{\ell}(\tS).(\projb{\hTLone_i} {E'})\}_{i\in I}}$ and $\projb \ghTtwo {E'}= \hsnd \p \q \dht{\{\lblFmt{\ell}(\tS).(\projb{\hTLtwo_i} {E'})\}_{i\in I}}$. We conclude by definition of $\mrg$ (in particular its behaviour with respect to send constructs) and induction hypothesis.

  \ \\{[$\{\p,\q\}\cap E=\emptyset$]} We have $\projb \ghT E = \projb {\hmsgi \p \q \ell \tS \ghT} E= \merge_{i\in I}(\projb\ghTi E)=\hTLone\lmrg\hTLtwo$. In this case
  \begin{small}
    \[
    \begin{array}{l}
\ulba {\ghT} {\hTLone} {\hTLtwo} =\\
(\hmsg \p \q \dht{\{\lblFmt{\ell_i} (\tS_i). \pi_1(\ulba \ghTi {\dht{\hTLone_i}} {\dht{\hTLtwo_i}})\}_{i \in I}},\hmsg \p \q \dht{\{\lblFmt{\ell_i} (\tS_i).\pi_2(\ulba \ghTi {\dht{\hTLone_i}} {\dht{\hTLtwo_i}})\}_{i \in I}})=\\
(\ghTone,\ghTtwo)
\end{array}
\]
\end{small}
where $\ulpa {[\projb{\dht{\ghT_i}}E]_{i\in I}} {[\hTLone,\hTLtwo]} = ([\dht{\hTLone_i},\dht{\hTLtwo_i}])_{i\in I}$; we set $\dht{\ghTone_i}=\pi_1(\ulba \ghTi {\dht{\hTLone_i}} {\dht{\hTLtwo_i}})$ and $\dht{\ghTtwo_i}=\pi_2(\ulba \ghTi {\dht{\hTLone_i}} {\dht{\hTLtwo_i}})$. Thanks to Lemma \ref{lem:unmlp}, we have that $\dht{\hTLone_i}\lmrg\dht{\hTLtwo_i}=\projb{\ghTi}E$, $\merge_{i\in I}\dht{\hTLone_i}=\hTLone$, and $\merge_{i\in I}\dht{\hTLtwo_i}=\hTLtwo$. We can then apply induction hypothesis to $\ghTi$ and have that $\projb{\dht{\ghTone_i}}E=\dht{\hTLone_i}$, $\projb{\dht{\ghTtwo_i}}E=\dht{\hTLtwo_i}$, and $\projb{\dht{\ghTone_i}} {E'} \mrg \projb{\dht{\ghTtwo_i}}{E'}=\projb \ghTi {E'}$, for all $i\in I$ and for $E'\cap E=\emptyset$. Let us consider now $\projb \ghT E$ this is equal to:
\[
\projb \ghTone E = \merge_{i\in I}(\projb{\dht{\ghTone_i}} E)= \merge_{i\in I}\dht{\hTLone_i}=\hTLone
\]
thanks to what obtained from Lemma \ref{lem:unmlp} (see above). Analogously for $\projb\ghTone E=\hTLtwo$. If $\{\p,\q\}\cap E'\neq\emptyset$ we have that the projections of $\ghTone$ and $\ghTtwo$ have the same message prefix $\hmsg\p\q$, $\hsnd\p\q$, or $\hrcv\p\q$; are one analogous to the other. Let us consider the first one:
\[
\projb\ghTone{E'}\mrg\projb\ghTtwo{E'}= \hmsg\p\q\dht{\{\lblFmt{\ell_i} (\tS_i).(\projb{\dht{\ghTone_i}}{E'}\mrg\projb{\dht{\ghTtwo_i}}{E'}}\dht{)\}_{i \in I}}
\]
We conclude thanks to what obtained from induction hypothesis (see above). If $\{\p,\q\}\cap E'=\emptyset$ then
\[
\projb\ghTone{E'}\mrg\projb\ghTtwo{E'}= \left(\merge_{i\in I}\projb{\dht{\ghTone_i}}{E'}\right)\mrg\left(\merge_{i\in I}\projb{\dht{\ghTtwo_i}}{E'}\right)
\]
and we conclude thanks to what obtained from induction hypothesis (see above) and by associativity and commutativity of merge.

Now, to have the full lemma for $t=1,\dots,n$, for a generic $n$, it is sufficient to proceed by induction on $n$: we observe that the bas cases are obvious and the binary case is above. We are left with $n\geq 3$, namely:
\[
 \ula {\ghT} {[\hTLone,\hTLtwo,\dots,\dht{\hTL_n}]} = \ghTone\# \ula {\ghTb} {[\hTLtwo,\dots,\dht{\hTL_n}]}
\]
with $\ghTone= \pi_1(\ulba {\ghT} \hTLone {(\lmerge_{t=2,\dots,n} \dht{\hTL_t})})$
and $\ghTb=\pi_2(\ulba {\ghT} \hTLone {(\lmerge_{t=2,\dots,n} \dht{\hTL_t})})$. Since we have proved the theorem for the binary case and since $\ghT=\hTLone \mrg (\lmerge_{t=2,\dots,n} \dht{\hTL_t})$, we have that $\projb\ghTone E = \hTLone$, $\projb\ghTb E=(\lmerge_{t=2,\dots,n} \dht{\hTL_t})$, and $\projb\ghTone {E'}\mrg\projb\ghTb {E'}=\projb\ghT{E'}$. In particular $\projb\ghTb E=(\lmerge_{t=2,\dots,n} \dht{\hTL_t})$ allows us to apply induction hypothesis and we obtain, for\\ $\pi_2(\ulba {\ghT} \hTLone {(\lmerge_{t=2,\dots,n} \dht{\hTL_t})})=[\ghTtwo,\dots,\dht{\ghT_n}]$, that $\projb{\dht{\ghT_t}}E=\dht{\hTL_t}$ and that $\merge_{t=2,\dots,n}\projb{\dht{\ghT_t}}{E'}=\projb\ghTb {E'}$ for $E'\cap E=\emptyset$, and thus we conclude.

\end{proof}

The \emph{binary unmerge for projection} $\upb$ takes a type $\hTL$ (first argument) and ``unmerges'' it, according to the structure of two other types $\ghTone$ and $\ghTtwo$ (second and third arguments); thus $\upb$ returns a pair of types, the unmerged components of $\ghT$. The non-binary version of the \emph{unmerge for projection} $\up$ is defined by iterating the binary one, thus allowing for unmerging into $n$ types. The role of this function is clarified by Lemma \ref{lem:proj-bb}.

To correctly define such function we need an auxiliary function that handles the case in which a localiser merge is equal to a projection one: $\lmerge_{i\in I} \dht{\hTL_i} = \merge_{j\in J} \dht{\ghT_j}$. In this case we want to obtain, for all $i\in I,j\in J$, $\dht{\ghT^i_j}$ such that $\lmerge_i \dht{\ghT^i_j} =  \dht{\ghT_j}$ and $\merge_j \dht{\ghT^i_j} = \dht{\hTL_i}$.

\begin{definition}[Projection Unmerge for Localiser Merge]\label{def:unmpl}
  The partial function \emph{binary projection-unmerge for localiser merge} $\uplba \ghTone \ghTtwo \hTLone \hTLtwo$ is recursively defined by the following equations, whenever the recursive call is defined (we use $\pi_1$ and $\pi_2$ to indicate standard pair projections):
  \begin{small}
  \[
  \begin{array}{l}
    \uplba \hend \hend \hend \hend = ((\hend, \hend), (\hend,\hend))\\
    \uplba \hX \hX \hX \hX = ((\hX, \hX), (\hX,\hX))\\
    \uplba {\hrec \hX \hTLone} {\hrec \hX \hTLtwo} {\hrec \hX \ghTone}  {\hrec \hX \ghTtwo}  = ((\hrec\hX\hToneone,\hrec\hX\hTonetwo),(\hrec\hX\hTtwotwo,\hrec\hX\hTtwotwo))\\
    \qquad \text{where}\ \dht{\hT^n_m}=\pi_n(\pi_m(\uplba {\hTLone} {\hTLtwo} {\ghTone}  {\ghTtwo})))\ \text{for $n,m=1,2$}\\
    \uplba {(\hsendj \p \q \ell \tS \hTLone)}
           {(\hsendk \p \q \ell \tS \hTLtwo)}
           {(\hsend \p \q \ell \tS \ghTone)}
           {(\hsend \p \q \ell \tS \ghTtwo)} = \\
    \qquad ((\hsendj \p \q \ell \tS \hToneone,
            \hsendj \p \q \ell \tS \hTonetwo),
            (\hsendk \p \q \ell \tS \hTtwoone,
            \hsendk \p \q \ell \tS \hTtwotwo))\\
    \qquad \text{where $J\cup K = I$ and, for $n,m=1,2$,}\ \dht{{\hT^n_m}_i}=
    \begin{cases}
      \pi_n(\pi_m(\uplba {\dht{{\hTLone}_i}} {\dht{{\hTLtwo}_i}} {\dht{{\ghTone}_i}} {\dht{{\ghTtwo}_i}}))\ \text{if $i\in J\cup K$}\\
       \pi_n(\pi_m(\uplba {\dht{{\hTLone}_i}} {\dht{{\hTLone}_i}} {\dht{{\ghTone}_i}} {\dht{{\ghTtwo}_i}}))\ \text{if $i\in J\backslash K$}\\
       \pi_n(\pi_m(\uplba {\dht{{\hTLtwo}_i}} {\dht{{\hTLtwo}_i}} {\dht{{\ghTone}_i}} {\dht{{\ghTtwo}_i}}))\ \text{if $i\in K\backslash J$}
    \end{cases}\\
    \uplba {(\hrecv \p \q \ell \tS \hTLone)}
           {(\hrecv \p \q \ell \tS \hTLtwo)}
           {(\hrecvj \p \q \ell \tS \ghTone)}
           {(\hrecvk \p \q \ell \tS \ghTtwo)} = \\
    \qquad ((\hrecvj \p \q \ell \tS \hToneone,
            \hrecvk \p \q \ell \tS \hTonetwo),
            (\hrecvj \p \q \ell \tS \hTtwoone,
            \hrecvk \p \q \ell \tS \hTtwotwo))\\
    \qquad \text{where $J\cup K = I$ and, for $n,m=1,2$,}\ \dht{{\hT^n_m}_i}=
    \begin{cases}
      \pi_n(\pi_m(\uplba {\dht{{\hTLone}_i}} {\dht{{\hTLtwo}_i}}{\dht{{\ghTone}_i}} {\dht{{\ghTtwo}_i}}))\ \text{if $i\in J\cup K$}\\
       \pi_n(\pi_m(\uplba {\dht{{\hTLone}_i}} {\dht{{\hTLtwo}_i}} {\dht{{\ghTone}_i}} {\dht{{\ghTone}_i}} ))\ \text{if $i\in J\backslash K$}\\
       \pi_n(\pi_m(\uplba  {\dht{{\hTLone}_i}} {\dht{{\hTLtwo}_i} {\dht{{\ghTone}_i}} {\dht{{\ghTone}_i}}}))\ \text{if $i\in K\backslash J$}
    \end{cases}\\
    \uplba {(\hmsgi \p \q \ell \tS \hTLone)}
           {(\hmsgi \p \q \ell \tS \hTLtwo)}
           {(\hmsgi \p \q \ell \tS \ghTone)}
           {(\hmsgi \p \q \ell \tS \ghTtwo)}\\ =
     ((\hmsgi \p \q \ell \tS \hToneone,
            \hmsgi \p \q \ell \tS \hTonetwo),
            (\hmsgi \p \q \ell \tS \hTtwoone,
            \hmsgi \p \q \ell \tS \hTtwotwo))\\
    \qquad \text{where, for $n,m=1,2$,}\ \dht{{\hT^n_m}_i}=
    \pi_n(\pi_m(\uplba {\dht{{\hTLone}_i}} {\dht{{\hTLtwo}_i}} {\dht{{\ghTone}_i}} {\dht{{\ghTtwo}_i}}))\\[2mm]
      \text{undefined otherwise}
  \end{array}
  \]
  \end{small}
  The partial function \emph{projection-unmerge for localiser merge} $\upla {L_L} {L_P}$ takes as input two lists of types instead of two pairs. It returns a tuple of lists: if $L_L= [\hTLone,\dots,\dht{\hTL_n}]$ and $L_P=[\ghTone,\dots,\dht{\ghT_m}]$ the result list is of the kind $([\dht{\hT^1_1},\dots,\dht{\hT^1_m}],\dots,[\dht{\hT^n_1},\dots,\dht{\hT^n_m}])$ ($\pi_i$ indicate the standard projection of the $i$-th component of the tuple). $\upl$ is defined by recursion (on the sum of the lengths of $L_L$ and $L_P$) by the following equations.
  \begin{small}
  \[
  \begin{array}{l}
    \upla {[]} {L_P} = ([])\\
    \upla {[\hTL]} {L_P} = ([L_P])\\
    \upla  {[\hTLone, \hTLtwo]} {L_P} =
    \begin{cases}
      ([],[])\ \text{if $L_P=[]$}\\
      ([\pi_1(\pi_1(\uplba \hTLone \hTLtwo \ghT \ghT)),\pi_1(\pi_2(\uplba \hTLone \hTLtwo \ghT \ghT))])\ \text{if $L_P=[\ghT]$}\\
      ([\pi_1(\pi_1(\uplba \hTLone \hTLtwo \ghTone \ghTtwo)),
        \pi_2(\pi_1(\uplba \hTLone \hTLtwo \ghTone \ghTtwo))],\\
      \ \
       [\pi_1(\pi_2(\uplba \hTLone \hTLtwo \ghTone \ghTtwo)),
        \pi_2(\pi_2(\uplba \hTLone \hTLtwo \ghTone \ghTtwo))])\
      \text{if $L_P=[\ghTone,\ghTtwo]$}\\
      (\pi_1(\pi_1(U))\#L^1,\pi_1(\pi_2(U))\#L^2)\\
         \qquad\text{if $L_P=[\ghTone,\dots,\dht{\ghT_m}]$, where, for $i=1,2$}\\
         \qquad L^i=\pi_i(\upla {[(\pi_2(\pi_1(U))),(\pi_2(\pi_2(U)))]} {[\ghTtwo,\dots,\dht{\ghT_m}]})\ \text{and}\\
         \qquad U= \uplba \hTLone \hTLtwo \ghTone {(\lmerge_{k=2\dots m}\dht{\ghT_k})}
    \end{cases}\\
    \upla  {[\hTLone, \hTLtwo,\dots,\dht{\hTL_n}]} {L_P}
    = (L^1,\pi_1(\upla {[\hTLtwo,\dots,\dht{\hTL_n}]} {L^2}),\dots,\pi_{n-1}(\upla {[\hTLtwo,\dots,\dht{\hTL_n}]} {L^2}))\\
    \qquad \text{where, for $i=1,2$,}\ L^i = \upla {[\hTLone,(\merge_{k=2,\dots,n}\dht{\hTL_k})]} {L_L}
  \end{array}
  \]
  \end{small}
\end{definition}

\begin{lemma}\label{lem:unmpl}
  We are given hybrid types $\hTLone,\dots,\dht{\hTL_n}$ and $\ghTone,\dots,\dht{\ghT_m}$ such that $\lmerge_{i=1,\dots,n} \dht{\hTL_i} = \merge_{j=1,\dots,m} \dht{\ghT_j}$.\\

  Then $\upla [\hTLone,\dots,\dht{\hTL_n}] [\ghTone,\dots,\dht{\ghT_m}]$ is defined and, if we set
  \[
  \upla [\hTLone,\dots,\dht{\hTL_n}] [\ghTone,\dots,\dht{\ghT_m}] = ([\dht{\ghT^1_1},\dots,\dht{\ghT^1_m}],\dots,[\dht{\ghT^n_1},\dots,\dht{\ghT^n_m}])\text{,}
  \]
  we have that:
  \begin{itemize}
  \item $\merge_{j=1,\dots,m} ,\dht{\ghT^i_j} = \dht{\hTL_i}$, for all $i=1,\dots,n$, and
  \item $\lmerge_{i=1,\dots,n} ,\dht{\ghT^i_j} = \dht{\ghT_j}$, for all $j=1,\dots,n$.
  \end{itemize}
\end{lemma}

\begin{proof}\ \\
  \begin{itemize}
  \item We first prove the base cases, in particular the case in which $n=2$ and $m=2$ (formally we first prove the result for $\uplb$). This is done by structural induction on $\hTLone$, exploiting the hypothesis and the inductive nature both of $\mrg$ and $\lmrg$.
  \item We conclude by induction on $n+m$, by following the definition of $\upl$.
  \end{itemize}
\end{proof}

We now define the unmerge for the projection (Definition \ref{def:unml}) and ensure that it has the desired behaviour (Lemma \ref{lem:unmp}).

\begin{definition}[Unmerge for Projection]\label{def:unmp}
  The partial function \emph{binary unmerge for projection} $\upba \hTE \ghTPone \ghTPtwo$ is defined by the following equations, by recursion on $\hdepth\hTE$, whenever the recursive call is defined (we use $\pi_1$ and $\pi_2$ to indicate standard pair projections):
  \begin{small}
  \[
  \begin{array}{l}
    \upba \hend \ghTPone \ghTPtwo = (\hend,\hend)\\
    \upba \hX \ghTPone \ghTPtwo = (\hX,\hX)\\
    \upba {(\hrec \hX \hTE)} \ghTPone \ghTPtwo =
    \begin{cases}
    (\hrec \hX {\pi_1(\upba {\hTE}  {\dht{\ghTPone'}}  {\dht{\ghTPtwo'}})},
      \hrec \hX {\pi_2(\upba {\hTE}  {\dht{\ghTPone'}}  {\dht{\ghTPtwo'}})})
      \\ \qquad \text{if $\ghTPone=\hrec \hX {\dht{\ghTPone'}}$ and $\ghTPtwo=\hrec \hX {\dht{\ghTPtwo'}}$}
      \\
      (\hrec \hX \hTE,\hrec \hX \hTE)  \text{if $\ghTPone=\ghTPtwo=\hend$}
      \end{cases}\\
    \upba {(\hpar \ghTPone \ghTPtwo)} \hTEone \hTEtwo = \text{undefined}
    \\
    \upba {(\hsend \p \q \ell \tS \hTE)} {\ghTPone} {\ghTPtwo} =
    \\ \qquad
    \begin{cases}
       \hspace*{-1.7mm}\begin{array}{l}
         (\hsendj \p \q \ell \tS {\hTE^1},\hsendk \p \q \ell \tS {\hTE^2})\\
         \qquad \text{with $\dht{\hTE^1_i}= \pi_1(\upba {\dht{\hTE_i}} \ghTPone \ghTPtwo)$ and $\dht{\hTE^2_k}= \pi_2(\upba {\dht{\hTE_i}} \ghTPone \ghTPtwo)$,}
         \\ \qquad \text{for $i\in I$, if $\ghTPone = \hsend \p \q {\ell} \tS \ghTPone$, $\ghTPtwo= \hsend \p \q {\ell} \tS \hTEtwo$}
       \end{array}\\
    \end{cases}\\
    \\
    \upba {(\hrecv \p \q \ell \tS \hTE)} {\ghTPone} {\ghTPtwo} =
    \\ \qquad
    \begin{cases}
      \hspace*{-1.7mm}\begin{array}{l}
        (\hrecvj \p \q \ell \tS {\hTE^1},\hrecvk \p \q \ell \tS {\hTE^2})\\
        \qquad \text{with $\dht{\hTE^1_j}=\dht{\hTE_j}$ for $j\in J\backslash K$, $\dht{\hTE^1_j}= \pi_1(\upba {\dht{\hTE_j}} \ghTPone \ghTPtwo)$ for $j\in J\cap K$,}
    \\
        \qquad \text{$\dht{\hTE^2_k}=\dht{\hTE_k}$ for $k\in K\backslash J$, $\dht{\hTE^2_k}= \pi_2(\upba {\dht{\hTE_k}} \ghTPone \ghTPtwo)$ for $k\in J\cap K$}
        \\ \qquad \text{if $\ghTPone = \hrecvj \p \q {\ell} \tS \ghTPone$, $\ghTPtwo= \hrecvk \p \q {\ell} \tS \ghTPtwo$ and $I=J\cup K$}
      \end{array}\\
    \end{cases}\\
    \upba {(\hmsgi \p \q \ell \tS \hTE)} {\ghTPone} {\ghTPtwo} =
    \\ \qquad
    \begin{cases}
       \hspace*{-1.7mm}\begin{array}{l}
        (\hmsg \p \q \dht{\{\lblFmt{\ell_i} (\tS_i). \pi_1(\upba {\dht{\hTE_i}} {\dht{\ghTPone_i}} {\dht{\ghTPtwo_i}})\}_{i \in I}},\\
        \qquad\hmsg \p \q \dht{\{\lblFmt{\ell_i} (\tS_i). \pi_2(\upba {\dht{\hTE_i}} {\dht{\ghTPone_i}} {\dht{\ghTPtwo_i}})\}_{i \in I}})
        \\ \qquad \text{with $\uplba {[\loc{}{\dht{\hTE_i}}]_{i\in I}} {[\ghTPone,\ghTPtwo]} = ([\dht{\ghTPone_i},\dht{\ghTPtwo_i}])_{i\in I}$}
      \end{array}
    \end{cases}
  \end{array}
  \]
  \end{small}
  The partial function \emph{unmerge for projection} $\upa {\hTE} {\dht{L}}$ is defined by the following equations, by recursion on the list of hybrid types $\dht{L}$, whenever the recursive call is defined:
  \begin{small}
  \[
  \begin{array}{l}
    \upa {\hTE} {[]} = \text{undefined}\\
    \upa {\hTE} {[\ghTP]} = [\hTE]\\
    \upa {\hTE} {[\ghTPone,\ghTPtwo]} = [\pi_1(\upba {\hTE} \ghTPone \ghTPtwo),\pi_2(\upba {\hTE} \ghTPone \ghTPtwo)]\\
    \upa {\hTE} {(\ghTPone\#(\ghTPtwo\#\dht{L}))} =\\ \qquad
    \pi_1(\upba {\hTE} \ghTPone {(\merge({\sf set} (\ghTPtwo\#\dht{L}))}) \#
    \upa {\pi_2(\upba {\hTE} \ghTPone {(\lmerge({\sf set} (\ghTPtwo\#\dht{L}))})} {(\ghTPtwo\#\dht{L})}
  \end{array}
  \]
  \end{small}
\end{definition}

\begin{lemma}[Unmerge for Projection]\label{lem:proj-bb} \label{lem:unmp}
  We fix a set of participants $E$, and we are given hybrid types $\dht{\ghTP_t}$, $t=1,\dots,n$, and $\hTE$ such that:
  \begin{itemize}
  \item[(a)] $\hpart {\dht{\hTE}} \subseteq E$, for all $t=n$,
  \item[(b)] $\hepart  {\dht{\hTE}} \cap E = \emptyset$,
  \item[(c)] $\loc{} {\hTE} = \merge_{t=1,\dots,n} {\dht{\ghTP_t}}$;
  \end{itemize}
  then, for all $t=1,\dots,n$, we call $\dht{\hTE_t}$ the $t$-th element of the list $ \upEa {\hTE} {[\ghTPone,\dots,\dht{\ghTP_n}]}$ and we have
  \begin{enumerate}
  \item $\loc{} {\dht{\hTE_t}}  =  \dht{\ghTP_t}$, for all $t=1,\dots,n$ and
  \item $\merge_{t\in T} {\dht{\hTE_t}} = \hTE$.
  \end{enumerate}
\end{lemma}
\begin{proof}
  As for Lemma \ref{lem:unml}, the proof is pretty technical: we first prove the binary case $n=2$---which goes by induction on the $\hdepth\hTE$ and needs lemma \ref{lem:unmpl}---and then we generalise it to any $n$ by induction. The proof is totally analogous to the proof of Lemma \ref{lem:unml}, with the difference that localiser and projection ``play inverted roles''.
\end{proof}

We give an algorithm to explicitly build ``back'' a hybrid type, from two compatible types $\ghT$ and $\hTE$. We first describe the algorithm as a (partial) recursive function in the next Definition \ref{def:bb1-a}, then, with Theorem \ref{thm:comp1-a}, we prove that such functions has the desired behaviour.

\begin{definition}[Build-Back of a Single Component]\label{def:bb1-a}
  Given a set of participants $E$, the \emph{build-back of a single component} is the partial recursive function $\bbeoa \ghT \hTE$ defined below, by recursion on $\hdepth{\ghT}+\hdepth{\hTE}$. (In what follows we indicate by $\pi_t$ the function that given a list returns its $t$-th element, when it exists.)
  \begin{small}
  \[
  \begin{array}{l}
    \text{If $\hpart \hTE\nsubseteq E$, $\bbeoa \ghT \hTE$ is undefined, otherwise it is defined by the following equations:}\\
    \bbeoa \hend \hTE = \hTE\\
    \bbeoa \hX \hTE = \hTE\\
    \bbeoa {(\hrec \hX \ghTb)} {\hTE} =
    \begin{cases}
      \hpar {(\hrec \hX \ghTb)} {\hTE}\ \text{if $\hpart \ghTb \cap E = \emptyset$ and both $\hrec \hX \ghTb$ and $\hTE$ are closed}\\
      \hmsg \ps \pr \dht{\{\lblFmt{\ell}(\tS).(\bbeoa {(\hrec \hX \ghTb)} {\hTEi})\}_{i\in I}}
      \\ \qquad\text{if $\hpart \ghTb \cap E = \emptyset$, $\hTE=\hmsgi\ps\pr\ell\tS\hTE$,}
      \\ \qquad \text{and one of $\hrec \hX \ghTb$ and $\hTE$ is not closed}\\
      \hrec \hX {(\bbeoa \ghTb \hTEb)}\ \text{if $\hpart \ghTb \cap E \neq \emptyset$ and $\hTE = \hrec \hX \hTEb$}
    \end{cases}\\
    \bbeoa {(\hpar \ghTone \ghTtwo)} {\hTE} =
    \begin{cases}
      \hpar {(\bbeoa \ghTone \hTE)} \ghTtwo\ \text{if $\hpart \ghTtwo \cap E = \emptyset$}\\
      \hpar \ghTone {(\bbeoa \ghTtwo \hTE)}\ \text{if  $\hpart \ghTtwo \cap E \neq \emptyset$ and $\hpart \ghTone \cap E = \emptyset$}
    \end{cases}\\
    \bbeoa {\hsend \p \q \ell \tS \ghT} {\hTE} =
    \\ \qquad
    \begin{cases}
      \hsnd \p \q \dht{\{\lblFmt{\ell_i} (\dte{\tS_i}).(\bbeoa \ghTi \hTEi)\}_{i\in I}}\ \text{if $\hTE = \hsend \p \q \ell \tS \hTE$}\\
      \hmsg \ps \pr \dht{\{\lblFmt{m_j} (\dte{\tS'_j}).(\bbeoa {{\sf UL}_j} \hTEj)\}_{j\in J}}\
      \\ \qquad \text{with ${\sf UL}_j = \pi_j(\ula \ghT {[\loc{} {\hTEj}]_{j\in J}})$, if $\hTE = \hmsgj \ps \pr m {\tS'} \hTE$ and $\p\in E$}\\
      \hsnd \p \q \dht{\{\lblFmt{l_i} (\dte{\tS_i}).(\bbeoa \ghTi {{\sf UP}_i})\}_{i\in I}}\
      \\ \qquad \text{with ${\sf UP}_i = \pi_i(\upa \hTE {[\projb {\ghTi} E]_{i\in I}})$, if $\p\notin E$}
    \end{cases}\\
        \bbeoa {\hrecv \p \q \ell \tS \ghT} {\hTE} =
        \\ \qquad
        \begin{cases}
      \hrcv \p \q \dht{\{\lblFmt{\ell_i} (\dte{\tS_i}).(\bbeoa \ghTi \hTEi)\}_{i\in I}}\ \text{if $\hTE = \hrecv \p \q \ell \tS \hTE$}\\
      \hmsg \ps \pr \dht{\{\lblFmt{m_j} (\dte{\tS'_j}).(\bbeoa {{\sf UL}_j} \hTEj)\}_{j\in J}}\
      \\ \qquad \text{with ${\sf UL}_j = \pi_j(\ula \ghT {[\loc{} {\hTEj}]_{j\in J}})$, if $\hTE = \hmsgj \ps \pr m {\tS'} \hTE$ and $\q\in E$}\\
      \hrcv \p \q \dht{\{\lblFmt{l_i} (\dte{\tS_i}).(\bbeoa \ghTi {{\sf UP}_i})\}_{i\in I}}\
      \\ \qquad \text{with ${\sf UP}_i = \pi_i(\upa \hTE {[\projb {\ghTi} E]_{i\in I}})$, if $\q\notin E$}
    \end{cases}\\
        \bbeoa {\hmsgi \p \q \ell \tS \ghT} {\hTE} =
        \\ \qquad
    \begin{cases}
      \hmsg \p \q \dht{\{\lblFmt{\ell_i} (\dte{\tS_i}).(\bbeoa \ghTi \hTEi)\}_{i\in I}}\ \text{if $\hTE = \hsend \p \q \ell \tS \hTE$ or $\hTE = \hrecv \p \q \ell \tS \hTE$}\\
      \hmsg \ps \pr \dht{\{\lblFmt{m_j} (\dte{\tS'_j}).(\bbeoa {{\sf UL}_j} \hTEj)\}_{j\in J}}\
      \\ \qquad \text{with ${\sf UL}_j = \pi_j(\ula \ghT {[\loc{} {\hTEj}]_{j\in J}})$,}
      \\ \qquad \text{if $\hTE = \hmsgj \ps \pr m {\tS'} \hTE$ and $\p\in E$ or $\q\in E$}\\
      \hmsg \p \q \dht{\{\lblFmt{l_i} (\dte{\tS_i}).(\bbeoa \ghTi {{\sf UP}_i})\}_{i\in I}}\
      \\ \qquad \text{with ${\sf UP}_i = \pi_i(\upa \hTE {[\projb {\ghTi} E]_{i\in I}})$, if $\{\p,\q\}\cap E = \emptyset$}
    \end{cases}\\
    \text{and undefined if none of the above applies.}
  \end{array}
  \]
  \end{small}
\end{definition}
We observe that, in the above definition, for $\ghT=\hmsg \p \q \dots$, we have left undefined the case in which $\{\p,\q\}\subseteq E$. That is because $E$ is the set of participants for a certain component and the messages $\hmsg\p\q$ are for internal communication within this component. Namely such messages belong to $\hTE$, not to $\ghT$, which instead is the type for inter-component communication.

Finally, we prove the technical result that supports our whole theory. This theorem allows us to compose two hybrid types---$\ghT$, disciplining the inter-component communication of the system, and $\hTE$, taking care of interactions within a single component---into one that retains all the information of the two, in terms of projections. The proof of this theorem relies on the just-proved lemmas and goes by a combination of induction on the depth of both given types $\ghT$ and $\hTE$, and of case analysis on these.

\begin{theorem}[Building Back a Single Component]\label{thm:comp1-a} \ \\
  We fix a set of participants $E$, and we are given hybrid types $\ghT$ and $\hTE$, such that:
  \begin{itemize}
  \item[(a)] $\hpart \hTE \subseteq E$,
  \item[(b)] $\hepart \hTE \cap E = \emptyset$, and
  \item[(c)] $\projb {\ghT} {E} = \loc{} \hTE$;
  \end{itemize}
  we set $\hT=\bbeoa \ghT \hTE$ and we have:
  \begin{enumerate}
  \item $\projb {\hT} {E} = \hTE$ and
  \item for all $E'$, such that $E'\cap E = \emptyset$, $\projb \hT {E'} = \projb \ghT {E'}$.
  \end{enumerate}
  Moreover if \isglobal{\ghT} then \isglobal{\hT}.
\end{theorem}
\begin{proof}
  By induction on $d=\hdepth{\ghT}+\hdepth{\hTE}$. The case $d=0$ is trivial; let us then have $d=n+1$ and assume the theorem true for all $\ghTb,\hTb$ as in hypothesis, for which $\hdepth{\ghTb}+\hdepth{\hTb}\leq n$. We proceed by case analysis on $\ghT$. We call $\hT=\bbeoa \ghT \hTE$.

  \ \\{[$\ghT = \hend$]}
  By definition $\hT = \hTE$; we observe that $\mathbf{1.}$ projection onto a set containing all the internal participants is the identity function (Lemma \ref{lem:proj-fix-a}), and $\mathbf{2.}$ projection of a closed type ($\hTE$ is closed since its localisation is, Lemma \ref{lem:loc-clo-a}), onto a set of roles disjoint from the internal participants, is $\hend$ (Lemma \ref{lem:proj-end-a}).

  \ \\{[$\ghT = \hX$]} We observe that, since $\loc{}\hTE=\hX$, either $\hTE=\hX$, for which our proof is trivial, or $\hTE=\hmsgi\ps\pr\ell\tS\hTE$, in which case we observe that $\hX=\lmerge_{i\in I} \loc{}\hTEi = \loc{}\hTEi$. By induction on $\hdepth\hTE$ we prove our result.

  \ \\{[$\ghT = \hrec \hX \ghTb$]} All cases go by induction hypothesis, but the second case, where $\ghT$ and $\hTE$ share some participant, but they are not both closed. In fact, both are not for condition $(c)$. Thanks to condition $(c)$, Lemma \ref{lem:proj-end-a}, and a similar argument to the one used in the [$\ghT = \hX$] case above, we conclude by induction hypothesis.

  \ \\{[$\ghT = \hpar \ghTone \ghTtwo$]} We observe that $\projb \ghT E = \projb \ghTone E$, (or $\projb \ghT E = \projb \ghTtwo E$) without loss of generality. We conclude by induction hypothesis.

   \ \\{[$\ghT = \hsend \p \q \ell \tS \ghT$]}
   \begin{itemize}
   \item[$(a)$] If $\p\in E$ and $\hTE = \hsend \p \q \ell \tS \hTE$ we conclude by construction of $\hT=\bbeoa \ghT \hTE$ and induction hypothesis.
   \item[$(b)$] If $\p\in E$ and $\hTE = \hmsgj \pr \ps {m} {\tS'} \hTE $, then $\hsnd \p \q \dht{\{\lblFmt{\ell_i}(\tS_i).(\projb \ghTi E)\}} = \lmerge_{j\in J} \hTEj$. We can therefore apply Lemma \ref{lem:loc-bb} and obtain that for ${\sf UL}_j = \pi_j(\ula \ghT {[\loc{} {\hTEj}]_{j\in J}})$ we have $\loc {} \hTEj = \projb {{\sf UL}_j} E$ (for all $j\in J$) and $\merge_{j\in J}\projb{\dht{{\sf UL}_j}}{E'}= \projb \ghT {E'}$ (for all $E'\cap E =\emptyset$). By construction (namely $\bbeoa \ghT \hTE = \hsnd \ps \pr \dht{\{\lblFmt{m_j} (\dte{\tS'_j}).(\bbeoa {{\sf UL}_j} \hTEj)\}_{j\in J}}$) and induction hypothesis we conclude.

   \item[$(c)$] If $\p\notin E$, we have that $\loc {} \hTE = \merge_{i\in I} (\projb \ghTi E)$. We apply Lemma \ref{lem:proj-bb} and, for ${\sf UP}_i = \pi_i(\upa \hTE {[\projb {\ghTi} E]_{i\in I}})$, we have $\loc {} {{\sf UP}_i} = (\projb \ghTi E)$ and $\merge_{i\in I} {{\sf UP}_i} = \hTE$. By construction (namely $\bbeoa \ghT \hTE = \hsnd \p\q \dht{\{\lblFmt{\ell_i} (\dte{\tS_i}).(\bbeoa \ghTi{{\sf UP}_i})\}_{i\in I}}$) and induction hypothesis we conclude.

   \end{itemize}

   \ \\{[$\ghT = \hrecv \p \q  \ell \tS \ghT$]} and {[$\ghT = \hmsgi \p \q \ell \tS \ghT$]} These are analogous to to the above ``send'' case: in essence, they follow the different cases of the Definition \ref{def:bb1} of $\bbeo$ which are similar for all three different kind of messages.

\end{proof}

\subsection{Projection for Hybrid Types and Set Inclusion}
\label{subsec:appendix2}



The following lemma is needed for proving preservation of projection (Theorem \ref{thm:proj-comp-a} below), which is essential for achieving distributed specification for communication protocols.

\begin{lemma}[Projection Distributes over Merge]
  \label{lem:mer-proj-step}
  Given $\hT = \merge_{t= 1,\dots, n} \hTt$, then $\projb \hT E = \merge_{n=1,\dots N} {\projb \hTi E}$
\end{lemma}
\begin{proof}

  Let us first prove the lemma for $n=2$: $\hT=\hTone\mrg\hTtwo$. We proceed by induction on $\hT$. Cases $\hT=\hend$ and $\hT=\hX$ are trivial.

  Let us consider $\hT=\hrec \hX \hTb$, then $\hTone = \hrec\hX{\dht{\hTone'}}$ and $\hTtwo = \hrec\hX{\dht{\hTtwo'}}$, with $\hTb = \dht{\hTone'}\mrg\dht{\hTtwo'}$. If $\hT$ is closed and $\hrec \hX {\projb \hTb E}$ is not guarded, than neither are $\hrec \hX {\projb {\dht{\hTone'}} E}$ and $\hrec \hX {\projb {\dht{\hTtwo'}} E}$ (by reasoning on the participants of the result of $\mrg$ and the well-formedness of initial types), thus we conclude. In the other case, the thesis follows by induction hypothesis.

  If $\hT=\hpar {\dht{\hT_a}}{\dht{\hT_b}}$ the thesis follow simply by applying induction hypothesis.\\ \

  Let us consider the send case $\hT =\hsend \p \q \ell \tS \hT$, and hence $\hTone=\hsend \p \q {\ell} \tS \hTone$ and $\hTtwo=\hsend \p \q {\ell} \tS \hTtwo$. If $\p\in E$ then we conclude by definition of merge and induction hypothesis. If $\p\notin E$, then
  \[
  \begin{array}{l}
    \projb \hT E = \merge_{i\in I} \projb \hTi E =
     (\merge_{i\in I} \projb {(\dht{\hTone_i} \mrg \dht{\hTtwo_i})} E) \overset{IH}{=}\\
    \overset{IH}{=}  (\merge_{i\in I} ((\projb{\dht{\hTone_i}} E \mrg \projb{\dht{\hTtwo_i}} E))) =\\
    = (\merge_{i\in I} \projb {\dht{\hTone_i}} E) \mrg  (\merge_{i\in I} \projb {\dht{\hTtwo_i}} E) = \projb \hTone E \mrg \projb \hTtwo E
    \end{array}
  \]
  We have marked with $\overset{IH}{=}$ where we have exploited the induction hypothesis; all the other equalities follow by definitions or by associativity and commutativity of merge.\\ \

If $\hT =\hrecv \p \q \ell \tS \hT$, we have $\hTone=\hrecvj \p \q {\ell} \tS \hTone$ and $\hTtwo=\hrecvk \p \q {\ell} \tS \hTtwo$.  Since, $\hT = \hTone \mrg \hTtwo$, we observe that:
  \begin{itemize}
  \item if $i\in J\backslash K$ then $\hTi = \hTonei$,
  \item if $i\in K\backslash J$ then $\hTi = \hTtwoi$, and
  \item if $i\in J\cap K$ then $\hTi = \hTonei \mrg \hTtwoi$.
  \end{itemize}
  If $\q\in E$ we conclude by applying the definition of projection to $\hT =\hrecv \p \q \ell \tS \hT$ and induction hypothesis when $i\in J\cap K$. If $\q\notin E$) we have
  \[
  \begin{array}{l}
    \projb \hT E = \merge_{i\in I} \projb \hTi E =\\
    = (\merge_{i\in J\backslash K} \projb \hTi E) \mrg  (\merge_{i\in K\backslash J} \projb\hTi E) \mrg  (\merge_{i\in J\cap K} \projb \hTi E) =\\
    = (\merge_{j\in J\backslash K} \projb \hTonej E) \mrg  (\merge_{k\in K\backslash J} \projb \hTtwok E) \mrg  (\merge_{i\in J\cap K} \projb {(\hTonei \mrg \hTtwoi)} E) \overset{IH}{=}\\
    \overset{IH}{=}  (\merge_{j\in J\backslash K} \projb \hTonej E) \mrg  (\merge_{k\in K\backslash J} \projb \hTtwok E) \mrg  (\merge_{i\in J\cap K} ((\projb\hTonei E \mrg \projb\hTtwoi E))) =\\
    = (\merge_{j\in J} \projb \hTonej E) \mrg  (\merge_{k\in K} \projb \hTtwok E) = \projb \hTone E \mrg \projb \hTtwo E
    \end{array}
  \]
  As above induction hypothesis, definition, and associativity and communtativity of merge.\\ \

  The case $\hT =\hmsgi \p \q \ell \tS \hT$ has four cases:
  \begin{itemize}
  \item $\{\p,\q\}\subseteq E$, analogous to the send case above $\hsnd\p\q\dots$, when $\p in E$.
  \item $\p\in E$ and $\q\notin E$, analogous to the send case above $\hsnd\p\q\dots$, when $\p in E$.
  \item $\p\notin E$ and $\q\in E$, analogous to the receive case above $\hrcv\p\q\dots$, when $\q in E$.
  \item  $\{\p,\q\}\cap E=\emptyset$, analogous to the send case above $\hsnd\p\q\dots$, when $\p \notin E$.
  \end{itemize}
  \ \\ \

  Now, to have the full lemma for $t=1,\dots,n$, for a generic $n$, it is sufficient to proceed by induction on $n$: we observe that $\hT = \merge_{t=1,\dots,n} \hTt = (\merge_{t=1,\dots,n-1} \hTt) \mrg\hTn$. Now, thanks to the case $n = 2$ above, we have $\projb \hT E = \projb {( (\merge_{t=1,\dots,n-1} \hTt) \mrg\hTn)} E = \projb {(\merge_{t=1,\dots,n-1} \hTt)} E \mrg \projb {\hTn} E$. Since, by induction hypothesis $\projb {(\merge_{t=1,\dots,n-1} \hTt)} E = \merge_{t=1,\dots,n-1} \projb\hTt E$, we conclude.

\end{proof}

The following theorem proves that generalised projection on hybrid types is sound with respect to set inclusion. This is the conclusive step of our theory: after composing subprotocols into a general global type, thanks to the theorem below, we are guaranteed that, by projecting a hybrid type for some subprotocol onto one of its participant, we obtain \emph{the same local type} as we would if we projected directly from the global type that we have obtained by composition, and that disciplines the whole communication.

\begin{theorem}[Projection Composes over Set Inclusion]
  \label{thm:proj-comp-a}
   Given $\hT$, $E_1$ and $E_2$, with $E_2\subseteq E_1$, assuming $ \projb {(\projb \hT {E_1})} {E_2}$ is defined,
   \[
   \projb {(\projb \hT {E_1})} {E_2} = \projb \hT {E_2}
   \]
\end{theorem}

\begin{proof}[Proof of Theorem \ref{thm:proj-comp}]
  By induction on $\hT$.\\

  \noindent [$\hT = \hend$] and [$\hT = \hX$] are immediate.\\

  \noindent [$\hT = \hrec X {\hTb}$]  We observe that, if $\hrec \hX {\projb \hTb {E_1}} $ is not guarded, then neither is $\hrec \hX {\projb \hTb {E_2}} $, since $E_2 \subseteq E_1$ (this result is obtained simply by structural induction over $\hT$). Also, we know that the property of being closed is preserved by projection. In case  $\hrec \hX {\projb \hTb {E_1}} $ is guarded, but $\hrec \hX {\projb \hTb {E_2}} $ is not, then (since $\projb \hTb {E_2}$ is guarded, by Lemma \ref{lem:proj-guar-a}) we have that $\hpart \hTb \cap E_2 = \emptyset$ and hence also  $\hpart{(\projb \hTb {E_1})}\cap E_2=\emptyset$ ($\hpart{(\projb \hT E)}\subseteq\hpart\hT$; by structural induction on $\hT$). We can then conclude by following the defining rule for projection, in the recursion case ($\rulename{proj-rec}$ in Definition \ref{def:proj}), and by applying induction hypothesis and Lemma \ref{lem:proj-end-a} (depending on the cases).\\

  \noindent [$\hT = \hpar \hTone \hTtwo$] Since projections are defined, w.l.o.g., we can assume $E_2 \subseteq E_1\subseteq \hpart \hTone$. We conclude by direct application of induction hypothesis on $\hTone$.\\

  \noindent [$\hT = \hsend \ps \pext \ell S \hT$] We observe that $\pext\notin E_1$.

  \textbf{(a)} If $\ps\notin E_1$, then $\projb {(\projb \hT {E_1})} {E_2} = \projb {\merge_{i\in I} {(\projb \hTi {E_1}})} {E_2}$; by applying Lemma \ref{lem:mer-proj-step} and induction hypothesis we have $\projb {(\projb \hT {E_1})} {E_2} = \merge_{i\in I} (\projb \hTi {E_2})$ which is exactly what we need (since  $\ps\notin E_1$ entails $\ps\notin E_2$).

  \textbf{(b)} If $\ps\in E_1\backslash E_2$, then $\projb \hT {E_1} = \hsnd \ps \pext \dht{\{\lblFmt{\ell_i}(S_i).(} \projb \hTi {E_1} \dht{)\}_{i\in I}}$. By induction hypothesis, we then have $\projb {(\projb \hT {E_1})} {E_2} = \merge_{i\in I}  \projb {(\projb \hTi {E_1})} {E_2}  = \merge_{i\in I} \projb \hTi {E_2}$ and we conclude.

  \textbf{(c)} If $\ps\in E_2$, no merge happens and we conclude by simply unfolding the definition of projection and applying induction hypothesis.\\

  \noindent [$\hT = \hrecv \pext \pr \ell S \hT$] Analogous to the previous ``send'' case.\\

  \noindent [$\hT = \hmsgi \ps \pr \ell S \hT$]

  \textbf{(a)} If $\ps,\pr\notin E_1$, then $\projb {(\projb \hT {E_1})} {E_2} = \projb {\merge_{i\in I} (\projb \hTi {E_1})} {E_2}$ and $\projb \hT {E_1}= \merge_{i\in I} \projb \hTi {E_2}$. We conclude analogously to \textbf{(a)} in case [$\hT = \hsend \ps \pext \ell S \hT$] above, by applying Lemma \ref{lem:mer-proj-step} and induction hypothesis.

  \textbf{(b)} If 1. $\ps,\pr\in E_1\setminus E_2$, 2. $\ps\in E_1 \setminus E_2$ and $\pr\notin E_1$, or 3. $\pr\in E_1 \setminus E_2$ and $\ps\notin E_1$, then we have, respectively:
  \begin{enumerate}
  \item $\projb \hT {E_1} = \hmsg \ps \pr \dht{\{\lblFmt{\ell_i}(S_i).(} \projb \hTi {E_1} \dht{)\}_{i\in I}}$,
  \item $\projb \hT {E_1} = \hsnd \ps \pr \dht{\{\lblFmt{\ell_i}(S_i).(} \projb \hTi {E_1} \dht{)\}_{i\in I}}$, or
  \item  $\projb \hT {E_1} = \hrcv \ps \pr \dht{\{\lblFmt{\ell_i}(S_i).(} \projb \hTi {E_1} \dht{)\}_{i\in I}}$.
  \end{enumerate}
   In all cases, we conclude analogously to \textbf{(b)} in case [$\hT = \hsend \ps \pext \ell S \hT$] above.

  \textbf{(c)} If 1. $\ps\in E_2$ and $\pr\notin E_1$, or 2. $\pr\in E_2$ and $\ps\notin E_1$, then we have, respectively:
  \begin{enumerate}
  \item $\projb \hT {E_1} = \hsnd \ps \pr \dht{\{\lblFmt{\ell_i}(S_i).(} \projb \hTi {E_1} \dht{)\}_{i\in I}}$, or
  \item  $\projb \hT {E_1} = \hrcv \ps \pr \dht{\{\lblFmt{\ell_i}(S_i).(} \projb \hTi {E_1} \dht{)\}_{i\in I}}$.
  \end{enumerate}
  In both cases, we conclude analogously to \textbf{(c)} in case [$\hT = \hsend \ps \pext \ell S \hT$] above.

    \textbf{(d)} If 1. $\ps\in E_2$ and $\pr\in E_1\setminus E_2$, or 2. $\pr\in E_2$ and $\ps\in E_1\setminus E_2$, then we have, in both cases, $\projb \hT {E_1} = \hmsg \ps \pr \dht{\{\lblFmt{\ell_i}(S_i).(} \projb \hTi {E_1} \dht{)\}_{i\in I}}$. Then, projecting on $E_2$, respectively:
  \begin{enumerate}
  \item $\projb {(\projb \hT {E_1})} {E_2} = \hsnd \ps \pr \dht{\{\lblFmt{\ell_i}(S_i).(} \projb {(\projb \hTi {E_1})} {E_2} \dht{)\}_{i\in I}}$, or
  \item  $\projb {(\projb \hT {E_1})} {E_2} = \hrcv \ps \pr \dht{\{\lblFmt{\ell_i}(S_i).(} \projb {(\projb \hTi {E_1})} {E_2} \dht{)\}_{i\in I}}$.
  \end{enumerate}
  In both cases, we conclude by induction hypothesis.

  \textbf{(e)} If $\ps,\pr\in E_2$ then we have: $\projb \hT {E_1} = \hmsg \ps \pr \dht{\{\lblFmt{\ell_i}(S_i).(} \projb \hTi {E_1} \dht{)\}_{i\in I}}$.\\
  Then,  $\projb {(\projb \hT {E_1})} {E_2} =\allowbreak \hmsg \ps \pr \dht{\{\lblFmt{\ell_i}(S_i).(} \projb {(\projb \hTi {E_1})} {E_2} \dht{)\}_{i\in I}}$. We conclude by induction hypothesis.

\end{proof}

\subsection{Projection and Localiser for Hybrid Types with Delegation}
\label{sec:del-proj-a}

We define here \emph{projection} and \emph{localiser}
for our extension of hybrid types to delegation and
explicit connections, from \S\ref{subsec:del}.

\begin{definition}[Projection for Hybrid Types with Delegation]
\label{def:proj-del-a}
The projection on the set of participants $E$
for hybrid types with delegation, is a partial
operator, recursively
defined by the clauses in Figure \ref{fig:proj-del-a} (whenever the recursive call is defined).

 \emph{Merging} ($\mrg$) is defined as a partial commutative operator
 over two hybrid types such that: for all $\hT$, $\hT \mrg \hT = \hT$, it delves inductively inside all constructors (see Definition \ref{def:proj}), and \\[1mm]
\begin{small}
\centerline{$
\begin{array}{lll}
  (\dht{\bigwedge_{i\in I}\roleFmt{\p}_i\dqusi\roleFmt{\q}_i(\lblFmt{\ell_i});\hTi
  })
  \mrg (\dht{\drs \p \q(\lblFmt{\ell});\hT})&= &(\dht{\bigwedge_{i\in I}\roleFmt{\p}_i\dqusi\roleFmt{\q}_i(\lblFmt{\ell_i});\hTi})
  \vee (\dht{\drs \p \q(\lblFmt{\ell});\hT})\\
  &&\qquad\text{if the resulting intersection is non-ambiguous}\\
   (\dht{\bigwedge_{i\in I}\roleFmt{\p}_i\dqusi\roleFmt{\q}_i(\lblFmt{\ell_i});\hTi
  })
  \mrg (\dht{\drs \p \q(\lblFmt{\ell});\hT})&=&(\dht{\bigwedge_{i\in I}\roleFmt{\p}_i\dqusi\roleFmt{\q}_i(\lblFmt{\ell_i});(\hTi\mrg\hT)})\\
  &&\qquad\text{if $\p=\p_i$, $\q=\q_i$, and $\ell=\ell_i$, for some $i\in I$}\\
    (\dht{\bigwedge_{i\in I}\roleFmt{\p}_i?^{e}\roleFmt{\q}_i(\lblFmt{\ell_i});\hTi
  })\mrg\hend&= &(\dht{\bigwedge_{i\in I}\roleFmt{\p}_i?^{e}\roleFmt{\q}_i(\lblFmt{\ell_i});\hTi
  })
\end{array}
$}
\end{small}
\end{definition}

We use an inductive version of merge,
which is more restrictive than (and hence compatible with)
the equirecursive merge in \cite{CASTELLANI2020128},
Definition $4.8$ and following discussion. 
An input intersection $\hint \ell \hT$ is \emph{non-ambiguous}
if all inputs are distinct and if $\q\neq\pr$,
whenever the intersection combines
a simple input $\dr\p\q$ and connecting input $\dre\ps\pr$.




The partial function
$\projd \_ \de E$ (defined by the rules below)
generalises the delegation projection functions
$\projdone \_ \p \q$ and $\projdtwo\_\p\q$ from \cite{CASTELLANI2020128}
(Figure $6$); as in that paper, this function
is of a sequential nature: it is not defined
for branching, recursion, and parallel constructs.
In Figure \ref{fig:del-proj-del-a},
we provide the full definition.
To obtain well-delegated hybrid types
(see Definition $4.11$ in \cite{CASTELLANI2020128})
and correctly nest delegation,
we have added the set argument $\de$, that keeps track
of the delegation pairs. We ask that,
if $\{(\p,\q),(\ps,\pr)\}\subseteq\de$,
then $\{(\p,\q)\}\cap\{(\ps,\pr)\}=\emptyset$.
We can use one funtion $\projd {}{}{}$,
for generalising both $\projdone{}{}{}$ and $\projdtwo{}{}{}$,
since we rely on the set $E$, namely:
$\projdone \_ \p \q$ is generalised by $\projd \_ {\pqe}E$,
when $\p\in E$ and $\q \notin E$, while
$\projdtwo \_ \p \q$ is generalised by $\projd \_ {\pqe}E$
when $\q\in E$ and $\p\notin E$
In Figure \ref{fig:del-proj-del-a},
we use the following notation:
\begin{itemize}
\item $\projb \ms E$, for the projection extended to prefixes $\ms$, in the obvious way;
\item $\projd \hT \emptyset E$, for $\projb \hT E$;
\item $\disjpair \p \q \de$, when the following holds: for all $(\ps,\pr)\in\de$,
$\{\p,\q\}\cap\{\ps,\pr\}=\emptyset$.
\end{itemize}

The intuition behind these rules is
the following:
whenever a forward delegation is
encountered while projecting, the control is passed to
$\projd{}{}{}$ (e.g.,
$\projb{\hfdel\p\q\hT}E = \hafdel\p\q{\projd\hT{\{(p,q)\}}E}$
if $\p\in E$), which takes care
of projecting the sequence (single branch) in $\hT$ of
elementary interactions (messages or
well-nested delegations, see \cite{CASTELLANI2020128},
Definition 4.11), until a dual backward delegation
$\hbdel\q\p\hTb$ is encountered and the control is
passed back to $\projb{}{}$, e.g.,
$\projd{\hbdel\q\p\hTb}{\{(p,q)\}}E=\hpbdel\q\p{\projb\hTb E}$.

\begin{definition}[Localiser for Hybrid Types with Delegation]
\label{def:loc-del-a}
  The localiser for hybrid types with delegation,
  is a partial operator, 
  recursively defined by the clauses in
  Figure \ref{fig:loc-del-a} (whenever the recursive call is defined).

\emph{Merging for the localiser} ($\mrg$) is defined as a partial
commutative operator
 over two hybrid types such that:
 for all $\hT$, $\hT \lmrg \hT = \hT$,
 it delves inductively inside all constructors
 (see Definition \ref{def:loc}), and\\[1mm]
\begin{small}
\centerline{$
\begin{array}{lll}
  (\dht{\boxplus_{i\in I}\roleFmt{\p}_i\dexs\roleFmt{\q}_i(\lblFmt{\ell_i});\hTi
  })
  \lmrg (\dht{\drs \p \q(\lblFmt{\ell});\hT})&= &(\dht{\boxplus_{i\in I}\roleFmt{\p}_i\dexs\roleFmt{\q}_i(\lblFmt{\ell_i});\hTi})
  \boxplus (\dht{\dss \p \q(\lblFmt{\ell});\hT})\\
   (\dht{\boxplus_{i\in I}\roleFmt{\p}_i\dexs\roleFmt{\q}_i(\lblFmt{\ell_i});\hTi
  })
  \lmrg (\dht{\dss \p \q(\lblFmt{\ell});\hT})&=&(\dht{\boxplus_{i\in I}\roleFmt{\p}_i\dexs\roleFmt{\q}_i(\lblFmt{\ell_i});(\hTi\lmrg\hT)})\\
  &&\qquad\text{if $\p=\p_i$, $\q=\q_i$, and $\ell=\ell_i$, for some $i\in I$}\\
    (\dht{\boxplus_{i\in I}\roleFmt{\p}_i!^{e}\roleFmt{\q}_i(\lblFmt{\ell_i});\hTi
  })\lmrg\hend&= &(\dht{\boxplus_{i\in I}\roleFmt{\p}_i!^{e}\roleFmt{\q}_i(\lblFmt{\ell_i});\hTi
  })
\end{array}
$}
\end{small}
\end{definition}

\begin{figure}
\begin{small}
\[
\begin{array}{l}
\text{Equations for
$\projb\hend E$, $\projb\hX E$, $\projb{\hrec \hX \hT} E$, and
$\projb {\hpar \hTone \hTtwo} E$ are the same
as in Definition \ref{def:proj}.
}
\\
\begin{array}{ll}
\rulename{proj-$\alpha$} &\\
\projb {(\dht{\dms \p \q (\lblFmt{\ell});\hT})} E =
\begin{cases}
\dht{\dms \p \q (\lblFmt{\ell});(\projb\hT E)}\ \text{if $\{\p,\q\}\subseteq E$}\\
\dht{\dss \p \q (\lblFmt{\ell});(\projb\hT E)}\ \text{if $\p\in E$ and $\q\notin E$}\\
\dht{\drs \p \q (\lblFmt{\ell});(\projb\hT E)}\ \text{if $\p\notin E$ and $\q\in E$}\\
\projb \hT E\ \text{if $\{\p,\q\}\cap E=\emptyset$}
\end{cases}
&
\begin{array}{l}
\text{We use the following notation:}\\
\quad\dtos\in\{\dht{\to},\dht{\to^{e}}\}\text{,}\\
\quad\dexs\in\{\dht{!},\dht{!^{e}}\}\text{, and}\\
\quad\dqus\in\{\dht{?},\dht{?^{e}}\text{.}\}
\end{array}
\end{array}\\
\begin{array}{ll}
\rulename{proj-$\pi^{\textsf{out}}$}&\rulename{proj-$\pi^{\textsf{in}}$}
\\
\projb {(\dht{\dss \p \q (\lblFmt{\ell});\hT})} E =
\begin{cases}
\dht{\dss \p \q (\lblFmt{\ell});(\projb\hT E)}\\
\qquad\text{if $\p\in E$ and $\q\notin E$}\\
\projb \hT E\ \text{if $\{\p,\q\}\cap E=\emptyset$}
\end{cases}
&
\projb {(\dht{\drs \p \q (\lblFmt{\ell});\hT})} E =
\begin{cases}
\dht{\drs \p \q (\lblFmt{\ell});(\projb\hT E)}\\
\qquad\text{if $\p\notin E$ and $\q\in E$}\\
\projb \hT E\ \text{if $\{\p,\q\}\cap E=\emptyset$}
\end{cases}

\end{array}
\\
\begin{array}{ll}
  \rulename{proj-int} &  \rulename{proj-choice}\\
        \projb{(\hint \ell \hT)} {E}  = 
          \merge_{i\in I} (\projb {(\dht{\pi^{\textsf{in}}_i(\lblFmt{\ell_i});\hT_i})} {E}) 
        &
         \projb{(\hchoicep \ell \hT)} {E}  =
         \merge_{i\in I} (\projb {(\dht{\ms^p_i(\lblFmt{\ell_i});\hT_i})} {E})
         \end{array}\\
  \begin{array}{ll}
  \rulename{proj-fd}&\rulename{proj-bd}\\
  \projb{(\hfdel \p \q \hT)} E =
    \begin{cases}
      \hfdel \p \q {(\projb\hT E)}\ \text{if $\{\p,\q\}\subseteq E$}\\
      \hafdel\p \q {(\projd \hT {\{(\p,\q)\}} E)}\ \text{if $\p\in E$ and $\q\notin E$}\\
      \hpfdel\p \q {(\projd \hT {\{(\p,\q)\}} E)}\ \text{if $\p\notin E$ and $\q\in E$}\\
      \projb{\hT}E\ \text{if $\{\p,\q\}\cap E=\emptyset$}
    \end{cases}
    & \projb{(\hbdel \q \p \hT)} E =
    \begin{cases}
    \hbdel \q \p {(\projb\hT E)}\\ \quad \text{if $\{\p,\q\}\subseteq E$}\\
    \projb{\hT}E\\ \quad \text{if $\{\p,\q\}\cap E=\emptyset$}
    \end{cases}
  \end{array}\\
  \begin{array}{ll}
  \rulename{proj-afd}&\rulename{proj-pbd}\\
  \projb{(\hafdel \p \q \hT)}E=
  \begin{cases}
  \hafdel \p \q {(\projb \hT E)}) \\
  \qquad\text{if $\p\in E$ and $\q\notin E$}\\
  \projb \hT E \ \text{if $\{\p,\q\}\cap E=\emptyset$}
  \end{cases}
  &
  \projb{(\hpbdel \q \p \hT)}E=
  \begin{cases}
  \hpbdel \q \p {(\projb \hT E)}) \\
  \qquad\text{if $\p\in E$ and $\q\notin E$}\\
  \projb \hT E \ \text{if $\{\p,\q\}\cap E=\emptyset$}
  \end{cases}\\
  \rulename{proj-pfd}&\rulename{proj-abd}\\
  \projb{(\hpfdel \p \q \hT)}E=
  \begin{cases}
  \hpfdel \p \q {(\projb \hT E)}) \\
  \qquad\text{if $\p\notin E$ and $\q\in E$}\\
  \projb \hT E \ \text{if $\{\p,\q\}\cap E=\emptyset$}
  \end{cases}
  &
  \projb{(\habdel \q \p \hT)}E=
  \begin{cases}
  \habdel \q \p {(\projb \hT E)}) \\
  \qquad\text{if $\p\notin E$ and $\q\in E$}\\
  \projb \hT E \ \text{if $\{\p,\q\}\cap E=\emptyset$}
  \end{cases}
  \end{array}
\end{array}
\]
\end{small}
\caption{Recursive Rules for Projection for Hybrid Types with Delegation}
\label{fig:proj-del-a}
\end{figure}

\begin{figure}
\begin{small}
\[
\begin{array}{l}
\text{If $\p\in E$, $\q\notin E$, and $\pqp\in\de$,}
\ \ \projd {\dht{\ms(\lblFmt{\ell});\hT}}{\de} E =
\begin{cases}
\projd {\hT} \de E\ \text{if $\hpart{\ms}\cap (\{\q\}\cup E)\subseteq \{\p\}$}\\
\dht{(\projb\ms {E\setminus\{\p\}})(\lblFmt{\ell});(\projd\hT \de E)}\ \text{if $\hpart{\ms}\subseteq E\setminus\{\p\}$}
\end{cases}\\
\begin{array}{l}
\text{If $\p\notin E$, $\q\in E$, and $\pqp\in\de$,}\\
\qquad\qquad\qquad \projd {(\dht{\dms \ps \pr(\lblFmt{\ell});\hT})}\de E =
\\ \
\end{array}
\begin{cases}
\dht{\dms \q\pr(\lblFmt{\ell});(\projd {\hT} \de E)}\ \text{if $\p=\ps$, $\q\neq\pr$, and $\pr\in E$}\\
\dht{\dms \ps \q(\lblFmt{\ell});(\projd{\hT} \de E)}\ \text{if $\p=\pr$, $\q\neq\ps$, and $\ps\in E$}\\
\dht{\dss \q\pr(\lblFmt{\ell});(\projd {\hT} \de E)}\ \text{if $\p=\ps$, $\q\neq\pr$, and $\pr\notin E$}\\
\dht{\drs \ps \q(\lblFmt{\ell});(\projd{\hT} \de E)}\ \text{if $\p=\pr$, $\q\neq\ps$, and $\ps\notin E$}\\
\projd {\hT} \de E\ \text{if $\{\ps,\pr\}\cap E=\emptyset$}\\
\dht{(\projb {\dms \ps \pr} {E\setminus\{\q\}})(\lblFmt{\ell});(\projd\hT \de E)}\
\text{if $\{\ps,\pr\}\cap E\neq\emptyset$ and $\q\notin\{\ps,\pr\}$}
\end{cases}\\
\begin{array}{l}
\text{If $\p\notin E$, $\q\in E$, and $\pqp\in\de$,}\\
\qquad\qquad\qquad\qquad \projd {(\dht{\dss \ps \pr(\lblFmt{\ell});\hT})}\de E =
\\ \
\end{array}
\begin{cases}
\dht{\dss \q\pr(\lblFmt{\ell});(\projd {\hT} \de E)}\ \text{if $\p=\ps$, $\q\neq\pr$, and $\pr\notin E$}\\
\projd {\hT} \de E\ \text{if $\{\ps,\pr\}\cap E=\emptyset$}\\
\dht{\dss \ps \pr (\lblFmt{\ell});(\projd\hT \de E)}\ \text{if $\ps\in E$, $\pr\notin E$, and $\q\notin\{\ps,\pr\}$}
\end{cases}\\
\begin{array}{l}
\text{If $\p\notin E$, $\q\in E$, and $\pqp\in\de$,}\\
\qquad\qquad\qquad\qquad \projd{(\dht{\drs \ps \pr(\lblFmt{\ell});\hT})}\de E =
\\ \
\end{array}
\begin{cases}
\dht{\drs \ps\q(\lblFmt{\ell});(\projd {\hT} \de E)}\ \text{if $\p=\pr$, $\q\neq\ps$, and $\ps\notin E$}\\
\projd {\hT} \de E\ \text{if $\{\ps,\pr\}\cap E=\emptyset$}\\
\dht{\drs \ps \pr (\lblFmt{\ell});(\projd\hT \de E)}\ \text{if $\ps\notin E$, $\pr\in E$, and $\q\notin\{\ps,\pr\}$}
\end{cases}\\
\text{If $\disjpair \p\q\de$,}\ \ \projd {(\hfdel \p\q\hT)} \de E =
\begin{cases}
\hfdel \p \q {(\projd {\hT} \de E)} \ \text{if $\{\p,\q\}\subseteq E$}\\
\hafdel \p\q  {(\projd {\hT} {\pqe\cup\de} E)}\ \text{if $\p\in E$ and $\q\notin E$}\\
\hpfdel \p\q  {(\projd {\hT} {\pqe\cup\de} E)}\ \text{if $\p\notin E$ and $\q\in E$}\\
\projd {\hT} \de E\ \text{if $\{\p,\q\}\cap E=\emptyset$}
\end{cases}\\
\projd {(\hbdel \q\p\hT)} \de E =
\begin{cases}
\hbdel \p \q {(\projd \hT\de E)}\ \text{if $\disjpair \p\q\de$ and $\{\p,\q\}\subseteq E$}\\
\hpbdel \q \p {(\projd \hT {\de\setminus\pqe} E)}\ \text{if $\p \in E$, $\q\notin E$, and $\pqp\in\de$}\\
\habdel \q \p {(\projd \hT {\de\setminus\pqe} E)}\ \text{if $\p \notin E$, $\q\in E$, and $\pqp\in\de$}\\
\projd {\hT} \de E\ \text{if $\disjpair \p\q\de$ and $\{\p,\q\}\cap E=\emptyset$}
\end{cases}
\\
\text{If $\disjpair \p\q\de$,}\ \ \projd {(\hafdel \p\q\hT)} \de E =
\begin{cases}
\hafdel\p \q (\projd \hT\de E)\ \text{if $\p\in E$ and $\q\notin E$}\\
\projd \hT\de E \ \text{if $\{\p,\q\}\cap E=\emptyset$}\\
\end{cases}\\
\text{If $\disjpair \p\q\de$,}\ \ \projd {(\hpbdel \q\p\hT)} \de E =
\begin{cases}
\hpbdel\q \p (\projd \hT\de E)\ \text{if $\p\in E$ and $\q\notin E$}\\
\projd \hT\de E \ \text{if $\{\p,\q\}\cap E=\emptyset$}\\
\end{cases}\\
\text{If $\disjpair \p\q\de$,}\ \ \projd {(\hpfdel \p\q\hT)} \de E =
\begin{cases}
\hpfdel\p \q (\projd \hT\de E)\ \text{if $\p\notin E$ and $\q\in E$}\\
\projd \hT\de E \ \text{if $\{\p,\q\}\cap E=\emptyset$}\\
\end{cases}\\
\text{If $\disjpair \p\q\de$,}\ \ \projd {(\habdel \q\p\hT)} \de E =
\begin{cases}
\habdel\q \p (\projd \hT\de E)\ \text{if $\p\notin E$ and $\q\in E$}\\
\projd \hT\de E \ \text{if $\{\p,\q\}\cap E=\emptyset$}\\
\end{cases}
\end{array}
\]
\end{small}
\caption{Delegation Projection for Hybrid Types with Delegation}
\label{fig:del-proj-del-a}
\end{figure}

\begin{figure}
\begin{small}
\[
\begin{array}{l}
\text{Equations for
$\loc{}\hend $, $\loc{}\hX $, $\loc{}{\hrec \hX \hT} $, and
$\loc{} {\hpar \hTone \hTtwo} $ are the same
as in Definition \ref{def:loc}.
}\\
\begin{array}{lll}
\rulename{loc-msg}&\rulename{loc-send}&\rulename{loc-fd}\\
\loc{}{\dht{\da(\lblFmt{\ell});\hT}}=\loc{}\hT&
\loc{}{\dht{\dpo(\lblFmt{\ell});\hT}}=\dpo(\lblFmt{\ell});(\loc{}\hT)&
\loc{}{\hfdel\p\q\hT}=\loc{}\hT\\
\rulename{loc-bd}&\rulename{loc-afd}&\rulename{loc-pbd}\\
\loc{}{\hbdel\q\p\hT}=\loc{}\hT&
\loc{}{\hafdel\p\q\hT}=\hafdel\p\q{(\loc{}\hT)}&
\loc{}{\hpbdel\q\p\hT}=\hpbdel\q\p{(\loc{}\hT)}
\end{array}\\
\begin{array}{ll}
\rulename{loc-pfd}&\rulename{loc-abd}\\
\loc{}{\hpfdel\p\q\hT}=\hpfdel\p\q{(\loc{}\hT)}&
\loc{}{\habdel\q\p\hT}=\habdel\q\p{(\loc{}\hT)}\\
\rulename{loc-int}&\rulename{loc-choicep}\\
\loc{}{\hint\ell\hT}=\dht{\bigwedge_{i\in I}\pi^{\textsf{in}}_i(\lblFmt{\ell_i});
(\loc{}\hTi)}
&
\loc{}{\hchoicep\ell\hT}=\dht{\lmerge_{i\in I}(\ms^{\p}_i(\lblFmt{\ell_i});
(\loc{}\hTi))}
\end{array}
\end{array}
\]
\end{small}
\caption{Recursive Rules for the Localiser for Hybrid Types with Delegation}
\label{fig:loc-del-a}
\end{figure}









\end{document}